\documentclass[version=preprint,floatrow]{iacrcc}

\usepackage{algorithmic}
\usepackage{amsmath, amsfonts, amsthm}
\usepackage{booktabs}
\usepackage{dsfont}
\usepackage{enumitem}
\usepackage{environ}
\usepackage{fancyvrb}
\usepackage{graphicx}
\usepackage{keyval}
\usepackage{lipsum}
\usepackage{listings}
\usepackage{multirow}
\usepackage{stmaryrd}
\usepackage{subcaption}
\usepackage{tabularx}
\usepackage{textcomp}
\usepackage{url}
\usepackage{varwidth} 
\usepackage{xspace}

\usepackage[dvipsnames]{xcolor}
\usepackage[linesnumbered,ruled,vlined]{algorithm2e}
\usepackage[most]{tcolorbox}
\usepackage[normalem]{ulem}
\usepackage{fontspec}

\newtheorem{fact}{Fact}

\newcolumntype{C}{>{\centering\arraybackslash}X}
\SetAlgorithmName{Protocol}{protocol}{List of Protocols}

\newcommand{\ours}{{\sc Orq}\xspace}

\newfontfamily\stixmath{STIXTwoMath}[
  Path = ./stix/,
  Extension = .otf,
]

\newcommand{\leftouterjoin}{\mathrel{\text{{\stixmath ⟕}}}}
\newcommand{\rightouterjoin}{\mathrel{\text{{\stixmath ⟖}}}}
\newcommand{\fullouterjoin}{\mathrel{\text{{\stixmath ⟗}}}}

\makeatletter
\renewcommand{\@algocf@capt@plain}{bottom}
\makeatother

\let\oldnl\nl
\newcommand{\nonl}{\renewcommand{\nl}{\let\nl\oldnl}}

\SetKwInOut{Input}{Input}
\SetKwInOut{vars}{vars}
\SetKwInOut{Output}{Output}
\SetKwInOut{Parameter}{Param.}
\SetKwInOut{Result}{Result}
\SetKwInput{Procedure}{Procedure}
\SetKwComment{Comment}{{//}}{}
\newcommand{\Com}[1]{\Comment*{\small #1}}
\SetEndCharOfAlgoLine{}

\newcounter{protocol}

\newcounter{parentAlgoLine}
\SetKwBlock{BEGIN}{}{}

\makeatletter
\NewEnviron{substeps}{%
  \protected@edef\theparentequation{\theequation}%
  \setcounter{parentAlgoLine}{\value{AlgoLine}}%
  \setcounter{AlgoLine}{0}%
  \BEGIN{\BODY}%
  \setcounter{AlgoLine}{\value{parentAlgoLine}}%
  \ignorespacesafterend
}
\makeatother

\newcommand{\share}[1]{\ensuremath{\llbracket #1 \rrbracket}\xspace}

\newtcblisting[auto counter]{apilisting}[2][]{sharp corners, 
    fonttitle=\bfseries, colframe=black, listing only, 
    listing options={basicstyle=\ttfamily,language=c}, 
    title=Listing \thetcbcounter: #2, #1}

\definecolor{codegreen}{rgb}{0,0.6,0}
\definecolor{codegray}{rgb}{0.5,0.5,0.5}
\definecolor{codepurple}{rgb}{0.58,0,0.82}
\definecolor{backcolour}{rgb}{0.95,0.95,0.92}

\lstdefinestyle{apistyle}{
    backgroundcolor=\color{white},
    commentstyle=\color{codegreen},
    keywordstyle=\color{magenta},
    numberstyle=\tiny\color{codegray},
    stringstyle=\color{codepurple},
    basicstyle=\ttfamily\small,
    breakatwhitespace=false,         
    breaklines=true,                 
    captionpos=b,                    
    keepspaces=true,                 
    numbers=left,                    
    numbersep=3pt,                  
    showspaces=false,                
    showstringspaces=false,
    showtabs=false,                  
    tabsize=2,
    frame=lines,	
	escapechar=\%
}

\newcommand{\camera}[1]{#1}
\newcommand{\tocs}[1]{#1}
\newcommand{\tocsCam}[1]{#1}

\newcommand{\protoPrefixAgg}{\text{\sc AggNet}}
\newcommand{\genperm}{\texttt{genShardedPerm}\xspace}
\newcommand{\genpermpair}{\texttt{genShardedPermPair}\xspace}
\newcommand{\applyperm}{\texttt{applyShardedPerm}\xspace}
\newcommand{\applyinvperm}{\texttt{applyInverseShardedPerm}\xspace}

\begin{document}

\title{ORQ: Complex Analytics on Private Data with Strong Security Guarantees}

\addauthor[email = {elibaum@bu.edu},
           orcid = {0009-0003-7311-1705},
           inst = {1},
           onclick = {https://elibaum.com},
]{Eli Baum}

\addauthor[email = {sambux@bu.edu},
           orcid = {0000-0003-0627-5091},
           inst = {1},
           onclick = {https://sambux.org},
]{Sam Buxbaum}

\addauthor[email = {nitinm@cs.utexas.edu},
           orcid = {0009-0000-4928-6386},
           inst = {2},
           onclick = {https://nitinm25.github.io/},
           footnote = {Work done while at Boston University}
]{Nitin Mathai}

\addauthor[email = {mfaisal@bu.edu},
           orcid = {0009-0005-0040-9134},
           inst = {1},
           onclick = {https://cs-people.bu.edu/mfaisal/},
]{Muhammad Faisal}

\addauthor[email = {vkalavri@bu.edu},
           orcid = {0000-0001-8219-4862},
           inst = {1},
           onclick = {https://cs-people.bu.edu/vkalavri/},
]{Vasiliki Kalavri}

\addauthor[email = {varia@bu.edu},
           orcid = {0000-0002-7460-3272},
           inst = {1},
           onclick = {https://www.mvaria.com/},
]{Mayank Varia}

\addauthor[email = {liagos@bu.edu},
           orcid = {0000-0003-4692-3022},
           inst = {1},
           onclick = {https://cs-people.bu.edu/liagos/index.html},
]{John Liagouris}

\addaffiliation[ror = {05qwgg493},
                city = {Boston, MA},
                country = {USA}]{Boston University}

\addaffiliation[ror = {00hj54h04},
                city = {Austin, TX},
                country = {USA}]{University of Texas at Austin}
                
\authorrunning{Baum, Buxbaum, Mathai, Faisal, Kalavri, Varia and Liagouris}

\genericfootnote{Extended version of \cite{orq-conf} (SOSP 2025).}

\maketitle

\begin{abstract}
We present \ours, a system that enables collaborative analysis of large private datasets using cryptographically secure multiparty computation (MPC). \ours protects data against semi-honest or malicious parties and can efficiently evaluate relational queries with multi-way joins and aggregations that have been considered notoriously expensive under MPC. 
To do so, \ours eliminates the quadratic cost of secure joins by leveraging the fact that,  in practice,  the structure of many real queries allows us to join records and apply the aggregations ``on the fly'' while keeping the result size bounded. 
On the system side, \ours contributes generic oblivious operators, a data-parallel vectorized query engine, a communication layer that amortizes MPC network costs, and a dataflow API for expressing relational analytics --- all built from~the~ground~up.

We evaluate \ours in LAN and WAN deployments on a diverse set of workloads, including complex queries with multiple joins and custom aggregations. When compared to state-of-the-art solutions, \ours significantly reduces MPC execution times and can process one order of magnitude larger datasets. For our most challenging workload, the full TPC-H benchmark, we report results entirely under MPC with Scale Factor 10 --- a scale that had previously been achieved only with information leakage or the use of trusted compute.
\end{abstract}\

\section{Introduction}\label{sec:intro}

Secure multiparty computation (MPC) \cite{10.1145/3387108} is a cryptographic technique that distributes trust across non-colluding parties, which perform a computation collectively.  
Thanks to its decentralized nature, it offers a robust solution that protects data both from individual parties and from powerful adversaries, who may have compromised a subset of the parties. 

To avoid leaking information about the data, MPC programs are \emph{oblivious}, i.e., they perform the same operations and memory accesses for all inputs of the same size.  In practice, oblivious execution hides access patterns and the data distribution by generating worst-case outputs. For example, a secure MPC filter on a table with $n$ records returns a table of the same size. This guarantees that the filter selectivity remains hidden from computing parties and external adversaries.
However, hiding the data distribution for binary operators incurs a quadratic blowup in the worst case.  An oblivious relational join on two input tables with $n$ records will return the Cartesian product of size $n^2$. 
The problem gets worse for large $n$ and multi-way join queries, due to a cascading effect: in general, a naive oblivious evaluation of a query with $k$ joins has $O(n^{k+1})$ time and space complexity.

Secure join queries have been extensively studied in the context of peer-to-peer MPC, where data owners also act as computing parties~\cite{Volgushev2019Conclave, Poddar2021Senate, fang2024secretflow, bater2018shrinkwrap, bater2020saqe, Bater2017SMCQL, peng2024mapcomp, luo2024secure, Wang2021Secure}.
In this setting, performance can be improved by introducing controlled leakage~\cite{djoin,bater2018shrinkwrap,bater2020saqe}, offloading computation to trusted parties~\cite{Volgushev2019Conclave},  or by applying optimizations tailored to specific data ownership schemes~\cite{fang2024secretflow, Wang2021Secure, luo2024secure, Poddar2021Senate, Bater2017SMCQL}.
For example, SecretFlow~\cite{fang2024secretflow} assumes the participation of exactly two data owners and Senate~\cite{Poddar2021Senate} expects each party to be a data owner who contributes an entire table to the analysis. 
Despite their particular differences, all of these works are restricted to a single (often custom) MPC protocol and threat model,  offload expensive operations to data owners for plaintext execution in their trusted compute, and do not scale beyond a small number of participants.

To lift the restriction on the number of data owners and the need for trusted resources, MPC has also been deployed in the outsourced setting~\cite{mozilla, google-sec-agg, prio, bwwc-privacy, FengScape2022,crypten,Asharov2023Secure, sefspu, secrecy, tva, rathee2024private}, where computation is performed by a small set of untrusted servers.
Although promising, this approach makes oblivious join queries even more challenging, as servers never get access to plaintext data and cannot perform operations in the clear. Prior works on outsourced systems for relational analytics suggest that avoiding the cascading effect is inherently incompatible with strong security guarantees: Secrecy~\cite{secrecy} ensures no information leakage but its join operator has quadratic cost, while Scape~\cite{FengScape2022} reduces the complexity to $O(n \log^2n)$ but leaks the join result size to the servers.

\paragraph{Core problem.} Making secure analytics practical in the outsourced setting requires tackling a hard open problem: \emph{how to efficiently evaluate queries that involve joins where both input tables may contain duplicates}. This is the crux of the cascading effect and becomes the norm when multiple data owners contribute data to a joint analysis. 
To evaluate such queries, all prior MPC works that support fully oblivious execution either compute the Cartesian product of the input tables (which is expensive), enforce a subquadratic upper bound on the join output (which silently drops joined records), or leak the join result size (which is insecure).
In fact, the problem is not specific to MPC but manifests itself in the broader area of oblivious analytics. When it comes to joins on duplicate keys, even state-of-the-art approaches for trusted hardware either require an upper bound on the join result~\cite{opaque, DBLP:conf/eurosys/DaveLPGS20} or introduce volume leakage~\cite{cryptoeprint:2025/183}.

This paper describes the design and implementation of \ours, an \emph{Oblivious Relational Query engine} that enables efficient analytics in the outsourced MPC setting with strong security guarantees.  \ours takes a holistic approach that considers the entire workload and supports a broad class of analytics without suffering from the limitations of prior approaches.
We make the observation that all relational queries used in prominent MPC works (even those that include multi-way joins on potentially duplicate keys) produce \emph{results whose size is worst-case bounded by the input size}; that is, there exists an upper bound $O(n)$ on the query output that is independent of the data distribution. 
By analyzing 31 real and synthetic workloads we collected from the standard TPC-H benchmark~\cite{tpc-h} and prior MPC works, we find that the cascading effect is avoidable in all of them, even in the  queries that have been used to pinpoint the problem and motivate the need for leakage~\cite{bater2018shrinkwrap}. \ours leverages the data-independent upper bound on the query result size to achieve fully oblivious tractable execution --- entirely under MPC --- with no leakage and without relying on trusted compute.

\paragraph{A new scalable approach to oblivious analytics.}
In a typical MPC scenario, data owners would only agree to participate in analyses that return aggregated results (to preserve their privacy), and,  \camera{oftentimes, the aggregation function is amenable to partial evaluation}. 
Driven by this insight, \ours's core novelty is a technique that avoids materializing the Cartesian product of joins by decomposing and applying the aggregation \camera{eagerly} (whenever possible) to bound the size of intermediate results.
In practice, \ours performs a series of composite join-aggregation operations resembling a MapReduce-style evaluation~\cite{pig-latin}. To achieve competitive performance, \ours generalizes recent oblivious shuffling and sorting techniques~\cite{ahi22-radixsort,permutation-correlations} to different MPC protocols, optimizes them for tabular data (to avoid redundant column permutations), and integrates them with a 
butterfly-style control flow~\cite{B_goodrich2011data,A_blanton2012private} (to further reduce communication rounds between parties). It evaluates all queries we have collected from prior MPC works with \emph{asymptotic cost equal to that of sorting the input tables}: $O(n\log n)$ operations, $O(\log n)$ rounds, and $O(n)$ memory, where $n$ is the total number of input records. 

\paragraph{Contributions.} We make the following contributions:

\begin{itemize}

\item We develop \ours, a system that enables complex analytics on large private datasets using cryptographically secure MPC.  \ours provides a novel system runtime that employs data-parallel vectorized execution,  a high-level dataflow API for composing secure analytics on tabular~data, and an efficient communication layer that amortizes MPC costs in LAN and WAN settings.

\item We design an oblivious join-aggregation operator that supports all join types (inner, outer, semi- and anti-join), a broad class of aggregation functions, and composition with itself or other operators (e.g., user-defined filters) to form \emph{arbitrary} data pipelines.  Using this operator and known transformations from relational algebra, \ours can evaluate all queries we have collected from prior MPC works, without suffering from the quadratic blowup of intermediate results. 

\item We design all \ours operators to make black-box use of MPC primitives (e.g., $+, \times, \texttt{div}, \oplus, \land$).  As such, they can be easily instantiated with \emph{any} MPC protocol to support use cases with different threat models and security requirements. To showcase generality, we instantiate $\ours$ with three protocols: two semi-honest (for honest and dishonest majorities) and one maliciously secure protocol in the honest-majority setting.

\item We present a comprehensive experimental evaluation on diverse workloads, including the \emph{complete} TPC-H benchmark (for the first time under MPC)\footnote{We do not yet support substring operations.} and nine queries from prior works inspired by real use cases. We show that \ours significantly outperforms state-of-the-art relational MPC systems while also scaling complex analytics to one order of magnitude larger inputs.

\end{itemize}

\urlstyle{tt}

\noindent Additionally, we contribute data-parallel vectorized implementations of oblivious shuffling, quicksort, and radixsort that achieve better concrete performance than existing approaches. For oblivious sorting, we report results with up to half a billion rows --- $10\times$ larger than the best published results ~\cite{ahi22-radixsort, 10.1145/3460120.3484560} (whose implementations are not publicly available). We have released \ours as open-source at:
\begin{center}
    \url{https://github.com/CASP-Systems-BU/orq}.
\end{center}

\subsection{Overview}

\tocsCam{
This extended version of \ours provides further background material and additional benchmarks. In Section \ref{sec:system}, we introduce the \ours system, its supported workloads, and its security guarantees. Section \ref{sec:secure-comp} provides readers with a gentle introduction to secure computation, with a focus on multiparty protocols. Next, Section \ref{sec:approach} introduces \ours's major operators, while Section \ref{sec:sort} describes our oblivious sorting framework. Section \ref{sec:join} illustrates our new composite join operator. Finally, Section \ref{sec:impl} describes our implementation, and Section \ref{sec:eval} evaluates \ours on a series of microbenchmarks and real-world queries. We discuss related work in Section \ref{sec:related}. We also include appendices with more technical details on oblivious shuffle (\S\ref{apdx:shuffle}), sorting (\S\ref{sec:radixsort-analysis} and \S\ref{sec:bound-quicksort}), and joins (\S\ref{apdx:join}). Bandwidth measurements for the experiments in \S\ref{sec:eval} are included in Appendix \ref{apdx:bw}.
}
\section{The \ours System}
\label{sec:system}

\begin{figure}[t]
    \centering
            \begin{subfigure}{0.6\textwidth}
            \centering
            \includegraphics[width=\textwidth]{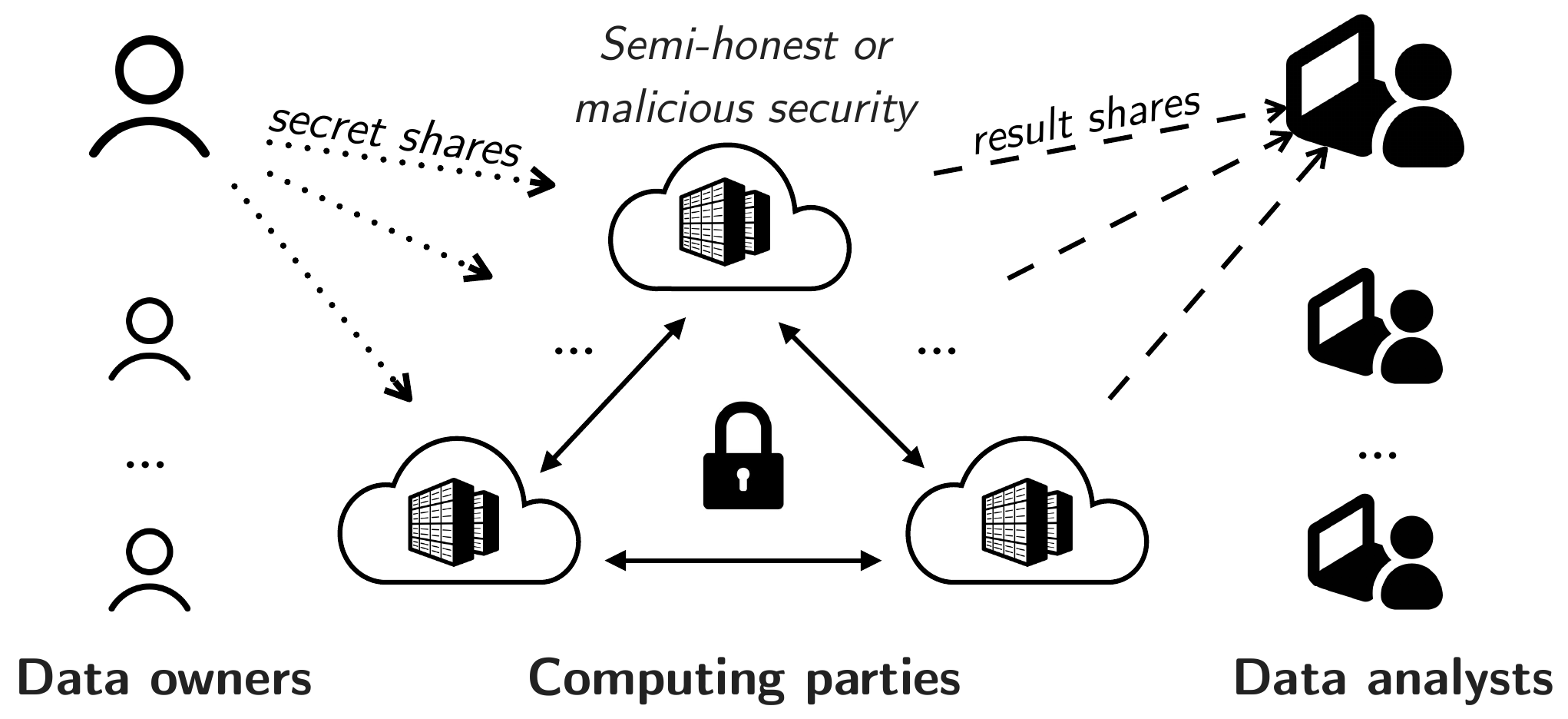}
            \end{subfigure}
         \caption{\ours targets the typical outsourced setting that uses a small set of computing parties (whose number depends on the MPC protocol) to support any number of data owners and analysts. }\label{fig:overview}
\end{figure}

\ours targets the outsourced setting in Figure~\ref{fig:overview} that has recently gained popularity in industry~\cite{mozilla,google-sec-agg, sefspu, FengScape2022, crypten},  academia \cite{prio,waldo,rathee2024private,secrecy,tva}, and non-profit organizations \cite{un-case-studies}. In our setting, \emph{data owners} distribute \emph{secret shares} of their data to a set of untrusted non-colluding \emph{computing parties}, in practice, servers managed by different infrastructure providers. Computing parties execute a query on the input shares, following a specified MPC protocol, and send shares of the result only to the designated \emph{data analysts}.  \ours currently supports three MPC protocols and can be instantiated with semi-honest or malicious security. We describe \ours's secret-sharing techniques and protocols in \S\ref{sec:secret-sharing}--\ref{sec:model}.

\ours is designed for the typical outsourced setting with no access to a Trusted Compute Base (TCB)~\cite{secrecy, FengScape2022}, but can also be deployed on premises in scenarios where data owners want to participate in the computation~\cite{Volgushev2019Conclave, Bater2017SMCQL, Poddar2021Senate} (and in that case it could use optimizations from prior work that leverage trusted resources). Naturally, \ours also supports the scenario in Flock~\cite{kaviani2024flock},  where a single data owner wants to securely outsource some computation~to~the~cloud.

\subsection{Supported workloads}\label{sec:workloads}
\ours supports a rich class of relational workloads, allowing an arbitrary composition of the following operators to form query plans: \texttt{SELECT}, \texttt{PROJECT}, \texttt{INNER\_JOIN}, \texttt{(RIGHT} / \texttt{LEFT} / \texttt{FULL)-OUTER\_JOIN}, \texttt{SEMI-JOIN},  \texttt{ANTI-JOIN}, \texttt{GROUP BY}, \texttt{DISTINCT}, \texttt{ORDER BY} and \texttt{LIMIT}. It also supports all common aggregations (\texttt{COUNT}, \texttt{SUM}, \texttt{MIN}, \texttt{MAX}, \texttt{AVG}) and user-defined aggregations, which can be constructed with its secure primitives: addition, multiplication, bitwise operations, comparisons, and division with private or public divisor.
\camera{\ours does not impose any restrictions on the database schema,  such as the existence of integrity constraints (e.g.,  primary or foreign keys),  but analysts can leverage these constraints, if they exist, to improve execution performance. }

\paragraph{What can \ours compute efficiently?} \ours's focus is on complex, multi-way join-aggregation queries that have so far been considered impractical in the outsourced MPC setting~\cite{FengScape2022,Asharov2023Secure}. Specifically, \ours can efficiently compute acyclic conjunctive queries~\cite{yannakakis1981algorithms} that include:
(i) only one-to-many joins, i.e., joins where at least one input has unique keys, or (ii) many-to-many joins (with duplicate keys in both inputs) as long as there exists a decomposable (or algebraic) aggregation function~\cite{10.1145/1629575.1629600, Gray1996Data}, and any group-by keys appear in a single input table.
Interestingly, all queries we have collected from prior relational MPC works~\cite{Volgushev2019Conclave,secrecy,Poddar2021Senate,bater2018shrinkwrap,fang2024secretflow, Wang2021Secure,Bater2017SMCQL,luo2024secure, bater2020saqe} fall in one of these two categories.

\camera{We note that many-to-many joins (also known as cross joins) are particularly common in collaborative data analysis because integrity constraints (e.g., the existence of unique join keys) may not be public.  Moreover, even if integrity constraints are known, it is hard to enforce them outside MPC,  unless data owners rely on a trusted third party or are willing to reveal some information about their data.}

\paragraph{What can \ours not compute efficiently?} There are queries outside this class that are inherently difficult to perform under MPC, and for these queries \ours falls back to an oblivious $O(n^2)$ join algorithm, like prior work. Examples of such queries (in SQL syntax) are: \texttt{SELECT COUNT(*) FROM A,B,C WHERE A.X=B.Y AND B.Y=C.W AND C.W=A.Z} (cyclic), or  \texttt{SELECT COUNT(*) FROM A,B WHERE A.X=B.Y AND A.W>B.Z}, when \texttt{A.X} contains duplicates (aggregation across tables).

\subsection{Programming model}
\ours users write queries in a dataflow-style API,  similar to the ones provided by Apache Spark~\cite{DBLP:journals/cacm/ZahariaXWDADMRV16} and Conclave~\cite{Volgushev2019Conclave}.  Relational operators are modeled as transformations on input tables and can be chained to construct complex queries.  Users submit their queries to the \ours engine, which compiles them into secure MPC programs for the computing parties.

\begin{lstlisting}[label={example}, caption={TPC-H Q3 with the \ours API},language=C++, float]
// Initialize C (Customers), O (Orders), LI (Lineitem)
...

// Apply filters
C.filter("MktSegment" == SEGMENT); %\label{line:filter-start}%
O.filter("OrdDate" < DATE);
LI.filter("ShipDate" > DATE); %\label{line:filter-end}%

// Calculate the revenue
LI["Revenue"] = LI["Price"] * (100 - LI["Discount"]) / 100;

// Joins and group-by with aggregation
auto RES = C.inner_join(O, {"CustKey"}) %\label{line:join-custkey}%
            .inner_join(LI, {"OrdKey"}, { %\label{line:join-ordkey}%
                {"OrdDate",  "OrdDate",  copy},
                {"Priority", "Priority", copy}})
            .aggregate({"OrdKey", "OrdDate", "Priority"}, {%\label{line:agg}%
                {"Revenue", "TotalRevenue", sum}})
            .project({"OrdKey", "TotalRevenue",
                      "OrdDate", "Priority"})
            .sort({{"TotalRevenue", DESC}, {"OrdDate", ASC}})
            .limit(10);
\end{lstlisting}

Listing~\ref{example} shows an implementation of TPC-H Q3 in the \ours API (see the benchmark specification~\cite{tpc-h} for the SQL definition). The query combines filters, joins, and aggregations to compute the shipping priority and total revenue of customer orders that had not been shipped as of a given date.  

Filters (lines~\ref{line:filter-start}--\ref{line:filter-end})  expect logical predicates that are constructed with \ours secure primitives (e.g., $\land$, $=$, $\neq$, $\geq$, $\leq$).  Joins expect a key and an optional list of aggregation functions (each with an input and output column)
that will be applied on the same key as the join.  Line~\ref{line:join-custkey} joins tables \texttt{C} and \texttt{O} on column \texttt{CustKey}. The result is a new (anonymous) table that is then joined with table \texttt{LI} on column \texttt{OrdKey}, in line~\ref{line:join-ordkey}. 
The join propagates \texttt{OrdDate} and \texttt{Priority} values from the left input to the output by applying the function \texttt{copy}.

The \texttt{aggregate} operator (line~\ref{line:agg}) expects a list of columns as grouping keys (\texttt{\{"OrdKey", "OrdDate", "Priority"\}}) and a list of functions along with their input and output columns (\texttt{\{"Revenue", "TotalRevenue", sum\}}).  \ours can apply multiple aggregation functions in a single control flow (see \S\ref{sec:composition}) and can also update columns in place (e.g., \texttt{\{"C",~"C",~sum\}}) to reuse memory.  \ours provides a set of built-in aggregation functions (e.g., \texttt{count}, \texttt{sum}, \texttt{min}, \texttt{max}, and \texttt{avg}) but also allows users to define custom ones.

Our rationale for exposing a dataflow-style API in \ours instead of a SQL interface was motivated by a recent critique of SQL~\cite{DBLP:conf/cidr/0001L24}.  We also found that multi-way join queries with nested sub-queries written with the \ours API were often more concise and easier to understand than their SQL counterparts.  For this paper,  we implemented 31 queries of varying complexity in a few hours. 

\subsection{Threat models and security guarantees}\label{sec:model}
Like in prior works~\cite{Volgushev2019Conclave,Poddar2021Senate,secrecy,tva},  data owners and analysts who want to use \ours must agree beforehand on the computation to execute.  The query and data schema are public and also known to the computing parties.  To ensure data privacy,  \ours uses end-to-end oblivious computation and protects all input, intermediate, and output data --- including the sizes of intermediate and final results --- throughout the analysis.

\paragraph{Threat models.}
MPC protocols protect data against computing parties and external adversaries who may control $T$ out of $N$ computing parties, where $T<N$.  \ours can withstand adversaries who can be \emph{semi-honest} (``honest-but-curious'') or \emph{malicious}.  A semi-honest adversary passively attempts to break privacy without deviating from the protocol (by monitoring the state of the parties it controls, e.g., via inspection of messages and access patterns).
A malicious adversary actively attempts to break privacy or correctness and can deviate arbitrarily from the protocol for the parties it controls.

\paragraph{Supported protocols.} 
\ours operators are protocol-agnostic, that is, they use the underlying MPC functionalities ($+, \times, \oplus, \land$) as black boxes \cite{DBLP:conf/crypto/DamgardN03}.
This way, they can be instantiated for any MPC protocol by simply replacing the corresponding functionalities. In its current version, \ours supports three state-of-the-art protocols with different threat models and guarantees: (i) the semi-honest ABY protocol \cite{aby} for dishonest majorities ($T=1$, $N=2$), (ii) the semi-honest protocol by Araki et al.~\cite{ArakiFLNO16} for honest majorities ($T=1$, $N=3$), and (iii) the malicious-secure Fantastic Four protocol~\cite{fantastic4} for honest majorities ($T=1$, $N=4$). 

\paragraph{Security guarantees.} 
\ours inherits the security guarantees of its underlying MPC protocols, ensures a fully oblivious control flow, and always operates over secret-shared data. It provides (i) \emph{privacy}, which means that computing parties and adversaries do not learn anything about the data, and (ii) \emph{correctness}, which means that all participants are convinced that the computation result is accurate.
Our malicious-secure protocol provides security \emph{with abort}, meaning that honest parties  halt the computation if they detect malicious behavior (to protect data privacy). That said, \ours can easily be amended to support robust execution if desired.
\section{Background on secure computation}
\label{sec:secure-comp}

The field of secure computation investigates efficient ways to execute \textit{public programs} on \textit{private data}. We wish to run computations in such a way that nothing more than the output is revealed: inputs and intermediate values remain hidden. There are three general approaches to solving this problem:

\begin{enumerate}
    \item \textbf{Trusted Execution Environments} (TEEs) provide security via hardware protections. Information is encrypted at rest, and only decrypted within a secure enclave which performs the actual computations.

    \item \textbf{Multiparty Computation} (MPC) distributes trust among multiple distinct parties and assumes that an adversary can only compromise a strict subset of the parties. Programs are decomposed into distributed protocols, where parties only exchange masked values which reveal nothing about the data.
    
    \item \textbf{Fully Homomorphic Encryption} (FHE) bases secure computation on encryption schemes which are designed such that operations can be performed \textit{directly} on ciphertexts and fully outsourced to an untrusted party, with clients performing only (relatively) cheap encryption and decryption operations.
\end{enumerate}

The current state of affairs is such that TEEs provide the lowest overhead for secure computation --- but carry with them the greatest risks. FHE represents the holy grail of secure computation, but current approaches are orders of magnitude slower. MPC represents a middle ground, where trust assumptions are minimal but protocols are efficient enough to evaluate real-world workloads. We do not address in this paper the related, but much more challenging, goal of hiding the program itself.

\subsection{Approaches to multiparty computation}

There are two main kinds of multiparty computation: \textit{garbled circuits} and \textit{secret sharing}. We focus exclusively on the latter in \ours, but for completeness we provide a brief overview of garbled circuit protocols.
The \textit{garbler} party ``encrypts'' the truth table for each gate in the computation circuit (represented as a directed acyclic graph of logic gates)
in such a way that the \textit{evaluator} party can obliviously execute the circuit. 
Garbled circuit protocols execute in a constant number of rounds but typically have higher bandwidths: for every bit of data in a circuit, they must transmit multiple ciphertexts. While recent theoretical works bring this communication overhead down \cite{liu2025bitgc} at the expense of increased computation, garbling is usually not the practical approach for real-world workloads, where we deal with very wide circuits (i.e., vectorized inputs) and relatively low-latency connections between all parties.

\label{sec:secret-sharing}

\begin{figure}
    \centering
    \includegraphics[width=0.35\columnwidth]{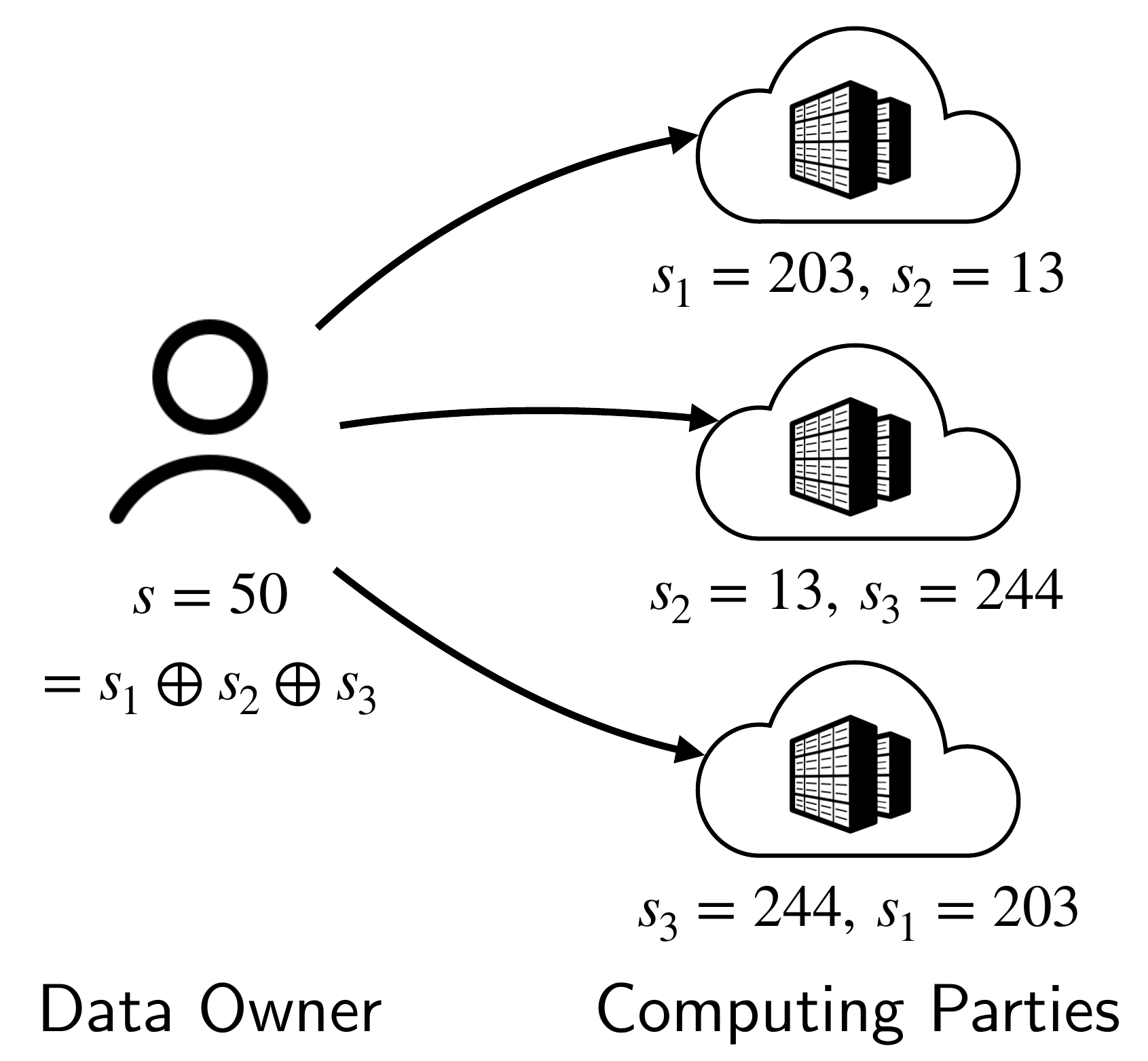}
    \caption{\tocs{Example of replicated boolean secret sharing}}
    \label{fig:ss-example}
\end{figure}

By contrast, \ours protects data using \emph{secret sharing}~\cite{shamir}. Given $N$ computing parties, a secret $\ell$-bit string $s$ is encoded using randomized $\ell$-bit shares $s_1, s_2, \cdots, s_N$ that individually are uniformly distributed over all possible $\ell$-bit strings and collectively are sufficient to reconstruct the secret $s$. 
\ours supports two types of secret sharing: \emph{arithmetic} (for numbers) and \emph{boolean} (for strings or numbers). In arithmetic sharing, the secret value $s$ is encoded with $N$ random shares such that $s = s_1 + s_2 + \cdots + s_N\bmod 2^\ell$, where $\bmod$ handles overflow. Similarly, in boolean sharing, $s$ is encoded such that $s = s_1 \oplus s_2 \oplus \cdots \oplus s_N$, where $\oplus$ is bitwise XOR. \ours also provides efficient primitives to convert between the two representations without relying on data owners. In this paper, when not clear from context, we will use $\share{x}$ to denote a secret sharing of $x$. Figure \ref{fig:ss-example} shows an example: the data owner, with secret $s=50$, randomly generates three secret shares $(203, 13, 244)$ and shares two with each of the three parties.

Data owners distribute secret shares to non-colluding computing parties, which in practice are \ours servers running in different trust domains. Each computing party receives a strict subset of the shares (per secret) so that no single party can reconstruct the original data. 
Given random shares, parties can execute a secure MPC protocol to jointly perform arbitrary computations on the shares --- as if they had access to the secrets --- and end up with shares of the result.
Operations on shares break down into a set of low-level secure primitives (e.g., $+, \times, \oplus, \land$), which are more expensive than the respective plaintext primitives because they require communication: for two $\ell$-bit shares, secure bitwise AND requires exchanging $O(\ell)$ bits between computing parties. XOR is a local operation which parties can apply directly to their secret shares:
\begin{gather*}  
x=x_1 \oplus x_2 \oplus x_3\qquad y=y_1 \oplus y_2 \oplus y_3\\
\mathcal P_1 : z_1 := x_1 \oplus y_1 \\
\mathcal P_2 : z_2 := x_2 \oplus y_2 \\
\mathcal P_3 : z_3 := x_3 \oplus y_3 \\
(x_1 \oplus y_1) \oplus (x_2 \oplus y_2) \oplus (x_3 \oplus y_3)=x\oplus y
\end{gather*}

\begin{figure}[t]
    \centering
    \includegraphics[width=0.8\columnwidth]{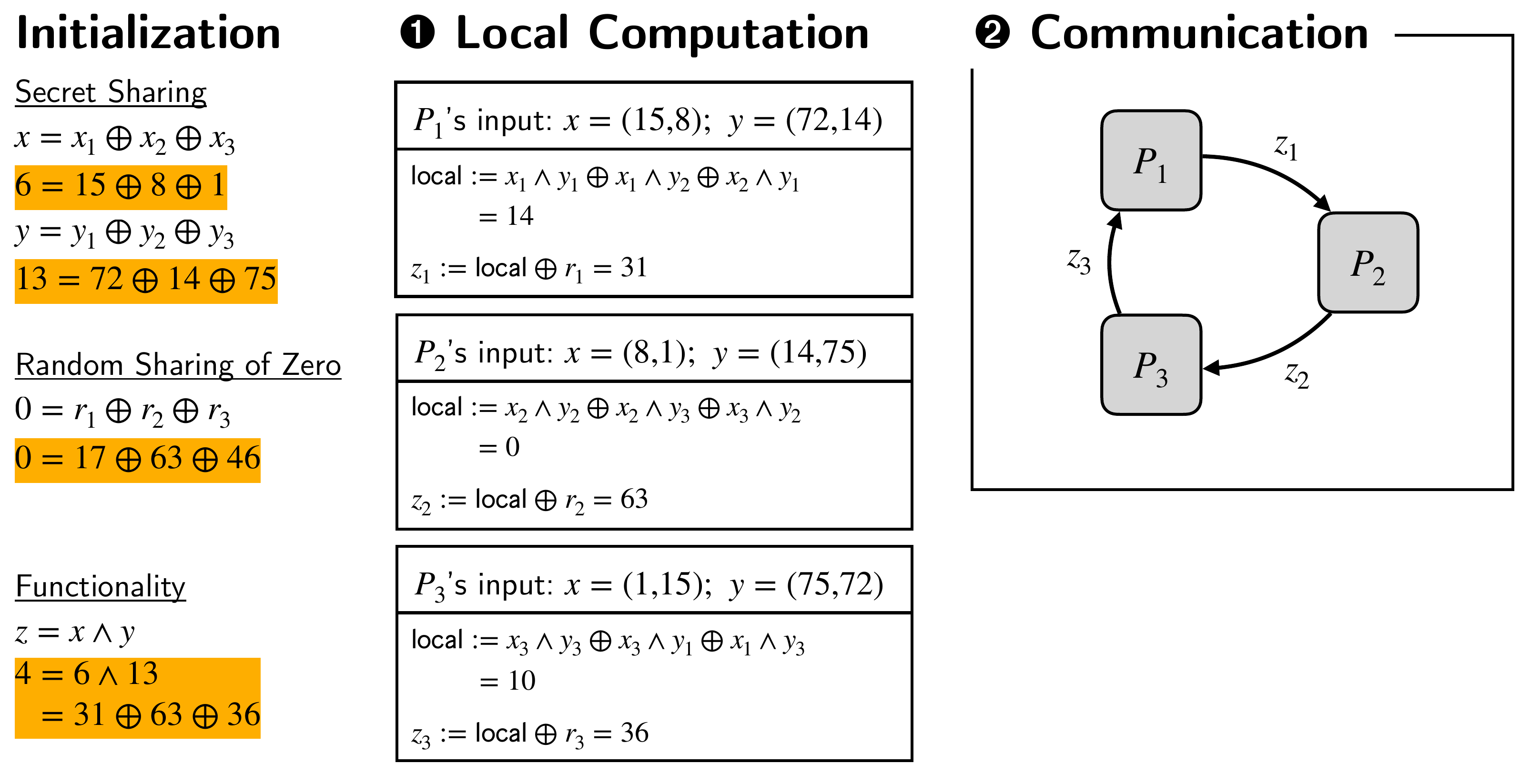}
    \caption{\tocs{Overview of secure AND protocol by Araki et al.~\cite{ArakiFLNO16}. Example values are highlighted in orange. (1) Parties perform local computation on their shares. (2) They create a replicated sharing by sending one share to their neighboring party. The protocol requires a random sharing of zero, which can be generated non-interactively using pairwise pseudorandom generators with pre-shared~keys.}}
    \label{fig:and}
\end{figure}

AND gates require communication, and thus a more complex protocol. \ours supports multiple underlying MPC protocols (each of which targets a different threat model); Figure \ref{fig:and} and the example below discuss the AND protocol for our replicated three-party (semi-honest, honest majority) protocol \cite{ArakiFLNO16}.

So-called \textit{replicated} secret-sharing protocols distribute multiple shares to each party. Observe that since all shares are random, each party still has no information about the secret: the third missing share always masks the true value. In the following discussion, we write $ab$ to represent the componentwise (bitwise) AND $a \wedge b$.
$$
\mathcal P_1 : (x_1, x_2)
\qquad
\mathcal P_2 : (x_2, x_3)
\qquad \mathcal P_3 : (x_3, x_1)
$$
To compute a sharing of $\share{x\wedge y}$, we consider the expansion of the product:
$$
\begin{matrix}
x\wedge y = & x_1 y_1 & \oplus & x_1 y_2 & \oplus & x_1 y_3& \oplus \\
& x_2 y_1 & \oplus & x_2 y_2 & \oplus & x_2 y_3 &\oplus \\
& x_3 y_1 & \oplus & x_1 y_2 & \oplus & x_3 y_3
\end{matrix}
$$
Each party has sufficient information to compute a single share of the product.
\begin{align*}
\mathcal P_1 &: z_1 := x_1 y_1 \oplus x_1 y_2 \oplus x_2 y_1\\
\mathcal P_2 &: z_2 := x_2 y_2 \oplus x_2 y_3 \oplus x_3 y_2\\
\mathcal P_3 &: z_3 := x_3 y_3 \oplus x_3 y_1 \oplus x_1 y_3
\end{align*}
However, while $\bigoplus_i z_i = x\cdot y$, the parties no longer possess a replicated sharing. Therefore, the final step is a communication phase, where parties mask their calculated shares with a random sharing of zero (i.e., $\bigoplus r_i = 0$) and send one to their neighbor. This communication is inherent to all MPC protocols and becomes the major bottleneck for large computations.

We did not yet discuss how to implement OR gates. Taking advantage of De Morgan's Law and the fact that bitwise negation is a local operation under boolean secret sharing (one party takes the negation of their shares), an OR gate has exactly the same cost as an AND gate: $x\vee y = \neg [(\neg x) \wedge (\neg y)]$. Together with the AND and XOR gates described above, this gives a complete boolean basis.
When working over the ring of integers $\bmod\ 2^k$, the same protocols apply: XOR becomes the addition protocol; the AND gate is analogous to multiplication.

Historically, secret-sharing-based MPC was considered to be constrained by its higher number of communication rounds as compared to garbled circuits. However, \ours demonstrates a viable approach to evaluate complex analytics in only $O(\log n)$ rounds; experimental evidence, even in the challenging WAN setting, confirms that we pay only a modest overhead compared to LAN.

\subsection{From secure primitives to more complex operations}\label{sec:inequality-example}

\begin{figure}[t]
    \centering
    \includegraphics[width=0.75\textwidth]{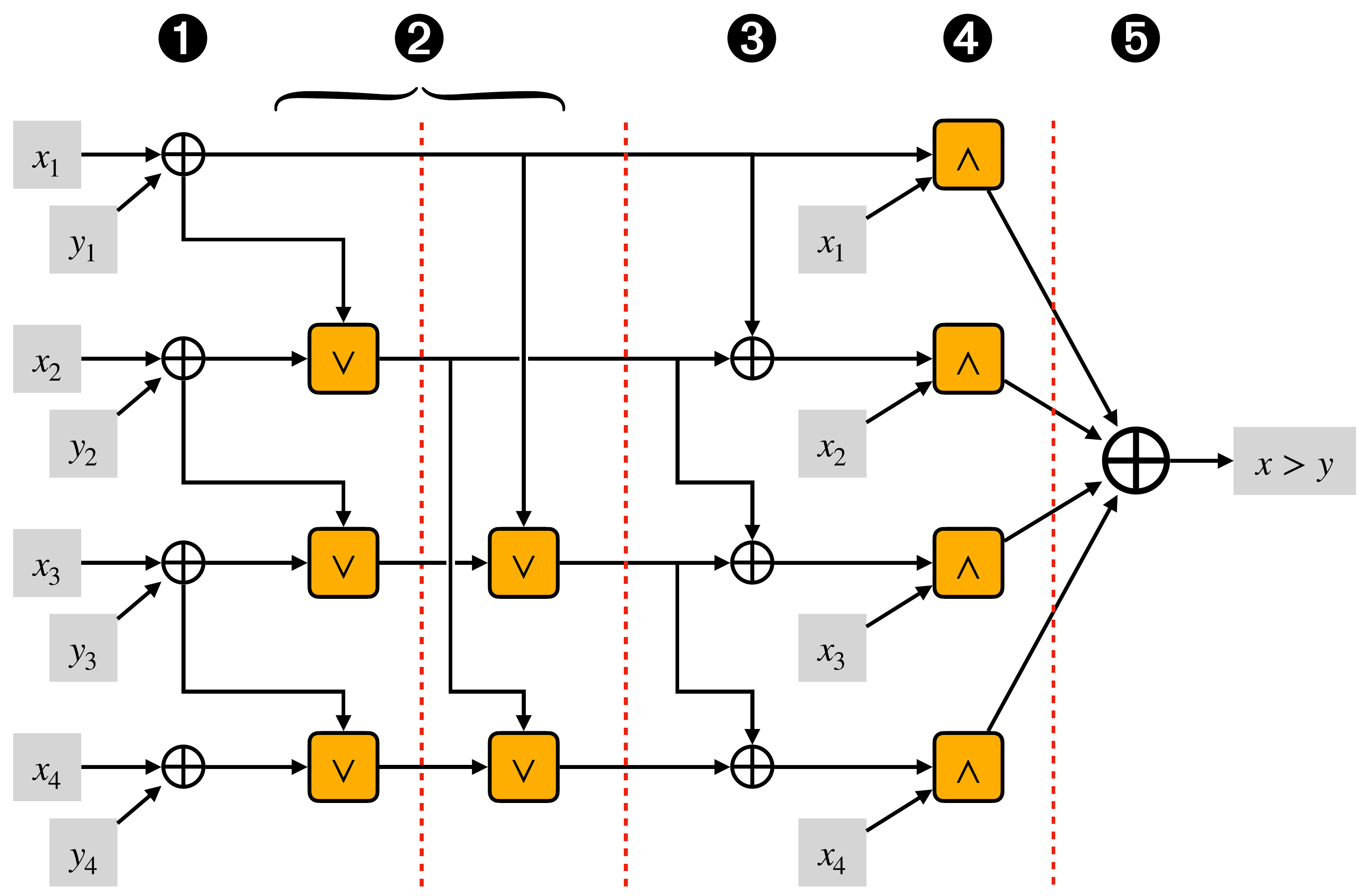}
    \caption{\tocs{Example inequality circuit that computes $x>y$, where $x$, $y$ are 4-bit unsigned integers. (1) Compute the bitwise XOR between $x$ and $y$. (2) Use a prefix-OR circuit to identify the location of the first difference between $x$ and $y$. (3) XOR this value with itself, offset by one, to generate a one-hot difference vector marking the location of the difference. (4) AND this value with the input $x$, to find the value of $x$ at the bit-position where $x$ and $y$ first differ. If $x$ is a one, then $y$ must be a zero, so $x > y$. Otherwise, $x \leq y$. (5) Compute the parity of the outputs. \textbf{Orange gates} (AND, OR) require communication. XOR gates are local. Dashed red barriers denote dependent computational stages, and thus sequential rounds of communication. Gates in the same stage are independent and can be computed in parallel; the entire computation takes three rounds. OR gates have the same cost as AND gates.}}
    \label{fig:inequality}
\end{figure}

In this section, we discuss an example of a multiparty computation problem. We show how to compute the secure comparison of two secrets, by decomposing the problem into XOR and AND gates, as introduced above. Figure \ref{fig:inequality} shows a circuit representation of the procedure described in the text.

We assume $x$ and $y$ are 4 bits wide and secret shared in the boolean domain (although we show plaintext values below for simplicity). For illustrative purposes, we will set $x=7$ and $y=5$. In binary, these are:
$$
\begin{matrix}
x = & 0 & 1 & 1 & 1\\
y = & 0 & 1 & 0 & 1\\
\end{matrix}
$$
First, we compute a bitstring, $\mathsf{sameBits}$, which consists of all zeros up until the first bit that differs between $x$ and $y$. This can be seen as $\mathsf{prefixOR(x\oplus y)}$ (and therefore computed by any efficient parallel-prefix circuit, which decomposes into a series of XORs and ANDs).
$$
\begin{array}{rcccccccc}
x \oplus y =        & 0 & 0 & 1 & 0\\
\mathsf{sameBits} = & 0 & 0 & 1 & 1\\
\end{array}
$$
At this point, observe that the LSB of $\mathsf{sameBits}$ gives (a secret sharing of) the equality predicate. Next, we calculate the expression:
$
(\mathsf{sameBits} \oplus (\mathsf{sameBits} \gg 1)) \wedge x
$, which extracts the bit of $x$ coincident with the first difference between $x$ and $y$. If 1, $x$ must be greater than $y$ (since this implies $y$ is 0 at that location, and the same as $x$ at all prior locations). Otherwise, $x$ must be less than or equal to $y$.
$$
\begin{array}{rcccccccc}
\mathsf{sameBits}               = & 0 & 0 & 1 & 1\\
\mathsf{sameBits} \gg 1         = & 0 & 0 & 0 & 1\\
\dots \oplus \mathsf{sameBits}  = & 0 & 0 & 1 & 0 \\
\dots \wedge x                  = & 0 & 0 & 1 & 0
\end{array}
$$
To move this bit to the LSB, we compute the XOR of all bits (implemented with \texttt{popcount} or a parity-check instruction) and place the result in the LSB.
$$
\begin{array}{rcccccccc}
\share{x > y} = & 0 & 0 & 0 & 1
\end{array}
$$

An important optimization in the design of efficient systems for MPC is the observation that independent messages within a stage can be batched together and sent in a single communication round. In the above computation, we do not need to perform eight sequential calls to the AND functionality. Instead, all run in parallel, and we communicate in a single round trip. This general structure also applies to the \textit{vectorized} operations common in collaborative analytics.

To compute other inequalities, we use the standard identities:
$a > b = \neg(a \leq b) = \neg(b \geq a)$. Further handling is necessary for signed integers; the result must be inverted if $\mathsf{sign}(x)\neq\mathsf{sign}(y)$ and $y<0$. We omit a full description of this protocol for brevity and refer readers to our implementation for more details.

\subsection{Obliviousness}
\ours programs must be \textit{oblivious}. A computation is oblivious if its control flow and observable side effects (such as data access patterns and timing) are input-independent, so that a computing party learns nothing about input, output, or intermediate values.  Obliviousness prevents data reconstruction attacks \cite{DBLP:conf/ccs/GrubbsLMP18,10.1145/2810103.2813651,DBLP:conf/ndss/BlackstoneKM20,DBLP:conf/eurosp/KamaraKMSTY22,DBLP:journals/popets/KamaraKMDPT24} but has an inherent overhead compared to plaintext computation,  as it requires using worst-case sizes and transforming all data-dependent conditionals (if-then-else statements) into data-independent ones. For example, consider the expression
$$\texttt{if a > b then c = a else d = b}$$
As written, a computing party (or adversary) could infer the relationship \texttt{a > b} by observing whether memory location \texttt{c} or \texttt{d} was written to.  Such expressions exist in relational analytics, e.g., in filters and \texttt{CASE} statements. To avoid leaking information about the data, \ours transforms and evaluates this expression  as shown below:

\begin{lstlisting}[numbers=none, language=C++]
condition = (a > b);
c = (1 - condition) * c + condition * a;
d = (1 - condition) * b + condition * d;
\end{lstlisting}

We now write to both memory locations, irrespective of the condition, and follow no data-dependent branches.  In practice, \texttt{a},  \texttt{b}, \texttt{c}, and \texttt{d} are secret-shared values, operator $>$ triggers a secure inequality protocol like the one described in \S\ref{sec:inequality-example}, and operators $\mathtt{*,+,-}$ correspond to secure functionalities provided by the underlying MPC protocol.

\subsection{Dishonest-majority MPC in the preprocessing model}

The two-party dishonest majority protocol \cite{aby} requires the generation of input-independent preprocessing material. We commonly refer to this preprocessing step as the \textit{offline phase} and measure it separately. The \textit{online phase} refers to the actual execution of the computation, after the inputs to the program are available.

Three kinds of preprocessing material, or correlated randomness, are required for our protocols: (1) boolean Beaver triples, for evaluating AND gates; (2) arithmetic Beaver triples, for evaluating multiplication gates; and (3) permutation correlations, for performing oblivious shuffle.

\subsubsection{Beaver triples}\label{sec:triples}

Beaver \cite{Beaver91a} made the simple but groundbreaking observation that correlated randomness could allow for the efficient evaluation of circuits in the online phase. For an arbitrary finite ring $\mathcal R$, a Beaver triple is a secret sharing of three values $a,b,c\in\mathcal R$ such that $a\times_{\mathcal R} b = c$, where $\times_{\mathcal R}$ represents the multiplication operation in the ring, and $a,b$ are sampled randomly. We will often write such a triple as $(\share a, \share b, \share c)$. Protocol \ref{beaver-mult} shows how we can ``derandomize'' a single Beaver triple to perform an arbitrary secret-shared multiplication.

\begin{algorithm}[t]
    \caption{Beaver multiplication protocol}
    \label{beaver-mult}
    \Input{$\share x,\share y\in\mathcal R$}
    \Output{$\share{z := x\times y}$}

	$\share a, \share b, \share c \gets \mathsf{getTriple}(\mathcal R)$ \\

	$\alpha\gets\mathsf{open}(x + a)$\\
	$\beta\gets\mathsf{open}(y + b)$\\

    \Return{$\share z:=(\alpha\times \share y) - (\share a \times \beta) + \share c$}
\end{algorithm}

Correctness follows directly:
\begin{align*}
\alpha\share y - \share a\beta + \share c &= (x + a)\share y - \share a(y + b) + \share{ab}\\
&=\share{xy}+\share{ay}-\share{ay}-\share{ab}+\share{ab}\\
&=\share{xy}
\end{align*}

Security comes from the fact that $a$ and $b$ are uniform and thus perfectly mask $x$ and $y$ in the opened values $\alpha,\beta$. The triple acts like a one-time pad; accordingly, each triple can be used only once (and, in fact, using a triple twice leaks information about the secret).

\paragraph{Example of rings.} This general template applies over any (finite) ring.
\begin{itemize}
	\item When $\mathcal R = \mathbb Z_{2^k}$, the integers $\bmod\ 2^k$, we have standard integer arithmetic. Generating such triples is expensive, but this ring is often the most natural one to express real-world problems in collaborative analytics. We implement the approaches of Gilboa \cite{gilboa} and Doerner et al. \cite{doerner2025}.
	\item When $\mathcal R = \mathbb F_2$, we have the boolean triples. Here, operations are performed bitwise; addition is XOR $(\oplus)$ and multiplication is AND $(\land)$. These triples reduce down to a cryptographic primitive known as \textit{random oblivious transfer}; we use the primitives provided by \texttt{libOTe} \cite{libote}.
\end{itemize}

In the semi-honest setting, Beaver triples are generally generated non-interactively given \textit{oblivious linear evaluations} (OLE) --- correlations of the form $w_0 = x_1\cdot y_0 + z_1$, where one party has $(w_0, y_0)$ and the other has $(x_1, z_1)$.\footnote{The term ``oblivious linear evaluation'' refers to the view of $P_1$'s randomness as defining a random line, and $P_0$'s as a random point on that line. In a finite ring, exponentially many points lie on each line, and exponentially many lines pass through each point, so neither party can learn anything about the other's input.} Given two OLEs, we can construct a Beaver triple with no further communication. Subscripts denote holders of each secret share.
\begin{align*}
\text{First OLE} \qquad s_0 &= t_1 u_0 + v_1 \quad\implies s_0 - v_1 = t_1 u_0\\
\text{Second OLE} \qquad w_0 &= x_1 y_0 + z_1 \quad\implies w_0 - z_1 = x_1 y_0\\
\end{align*}
$$
\underbrace{(u_0 - x_1)}_{\share a}\underbrace{(y_0 - t_1)}_{\share b} =
\underbrace{u_0 y_0 - u_0 t_1 - x_1 y_0 + x_1 t_1}_{\share c}
$$
The first and last terms of $\share c$ can be computed locally by $P_0$ and $P_1$, respectively, and each OLE gives a secret sharing of the middle terms. The final shares of the Beaver triple are:
$$
\begin{array}{rcccc}
      &   \share a  & \share b  & \share c \\
P_0 : (&  u_0, & y_0, & u_0y_0-s_0-w_0&)\\
P_1 : (& -x_1, &-t_1, & x_1t_1+v_1+z_1&)
\end{array}
$$

In the evaluation section (\S \ref{sec:eval-preproc}), we consider two protocols for generating base OLEs: the original one, due to Gilboa \cite{gilboa}, which directly emulates the schoolbook multiplication circuit using oblivious transfer, and has communication complexity $O(\ell^2)$; and a more recent construction by Doerner et al. \cite{doerner2025} which uses the Chinese Remainder Theorem and a prime-field decomposition to evaluate a multiplication gate with only $O(\ell\log\ell)$ communication.

\subsubsection{Permutation correlations}

A thematically similar construction allows us to perform oblivious shuffles of large vectors. Let $\vec x\in \mathcal R^n$ be a vector of length $n$ over some finite ring. A length-$n$ \textit{permutation correlation} is a pair of tuples $(\pi, \vec B_0), (\vec A, \vec B_1)$, where $\pi\in S_n$ (the symmetric group of size $n$) and $\vec A, \vec B_0, \vec B_1\in \mathcal R^n$. They are correlated such that $\pi(\vec{A}) = \vec{B}_0 + \vec{B}_1$ for a suitable definition of the addition operation.

Say we have two parties, $P_0$ and $P_1$, and $P_1$ has the secret vector $\vec{x}$. We will show an example of how to use the correlation $(\mathcal P_0: (\pi, \vec{B}_0), \mathcal P_1: (\vec{A}, \vec{B}_1))$ to permute $\vec{x}$ by $P_0$'s secret permutation $\pi$. First, $P_1$ masks the secret vector $\vec{x}$ with the random vector $\vec{A}$, and sends $\vec{\Delta} = \vec{x} - \vec{A}$ to $P_0$. $P_1$'s share of the shuffled vector is $\vec{B}_1$. $P_0$, upon receiving $\vec{\Delta} = \vec{x} - \vec{A}$ from $P_1$, permutes $\vec{\Delta}$ by $\pi$ to obtain $\vec{C}' = \pi(\vec{\Delta}) + \vec{B}_0$, and $P_1$'s share of the shuffled vector is $\vec{C}'$. We can observe that the output is a secret sharing of $\pi(\vec{x})$, where the final step follows from the fact that $\pi(\vec{A}) = \vec{B}_0 + \vec{B}_1$.
$$\vec{B}_1 + \vec{C}' = \vec{B}_1 + \pi(\vec{x} - \vec{A}) + \vec{B}_0 = (\vec{B}_0 + \vec{B}_1) - \pi(\vec{A}) + \pi(\vec{x}) = \pi(\vec{x})$$

A single permutation correlation suffices for one party (the party holding the permutation $\pi$) to permute a secret-shared vector according to $\pi$. To perform an oblivious shuffle, where \textit{neither} party knows the permutation that has been applied, we use two permutation correlations, $\pi_0$ and $\pi_1$, where each party acts as the sender (the party holding the permutation) in one of the correlations. The resulting permutation $\pi' = \pi_0 \circ \pi_1$, where $\circ$ represents composition, is random and hidden from both parties.

\section{\tocs{Oblivious building blocks}}\label{sec:approach}
\tocs{In this section, we review the oblivious primitives used as building blocks in \ours.}

\paragraph{Protocol primitives.} Protocol developers are expected to implement the basic cryptographic operations of a multiparty computation protocol: $+,\times,\oplus,\wedge$, along with utility primitives $\mathtt{open}$ and $\mathtt{secret\_share}$. From there, \ours provides implementations of comparison protocols ($=,\neq,>$), boolean circuits for addition (both ripple-carry and parallel-prefix) and division (non-restoring), along with conversion operators \texttt{a2b} and \texttt{b2a} to change between secret-sharing domains.

\paragraph{Table operations.} To hide true table sizes,  \ours tables (input, intermediate, or output) have a special \emph{validity column} of secret-shared bits, encoding which rows are ``valid'' $(1)$ or ``dummy'' $(0)$.  Dummy rows are those invalidated by oblivious operators and can also exist in input tables as a result of padding by data owners.  Since valid bits are secret-shared, \ours servers cannot distinguish real from invalid rows and operate on worst-case table sizes during the entire execution. Validity columns are never exposed to parties; they are managed by \ours operators transparently. Invalid rows are masked and shuffled before opening to prevent leakage of ``deleted'' data.

\ours provides efficient implementations of core oblivious operators for relational tables. An oblivious \textbf{filter} applies a predicate to each row in the input table and outputs a secret-shared bit that denotes whether the row passes the filter. Oblivious \textbf{deduplication} (the ``distinct'' operator) is applied to one or more columns treated as a composite \emph{key}. It obliviously sorts the table on these columns and then compares adjacent rows to mark the first occurrence of each distinct key with a secret-shared bit.
Oblivious \textbf{multiplexing} is used in sorting, join, and aggregation operators. An arithmetic multiplexer $\texttt{Mux}(b, x, y)$ implements $\texttt{b\,?\,y\,:\,x}$ by evaluating $(1-b)\cdot x + b \cdot y$, where $b \in \{0, 1\}$ and $x$, $y$ are the elements to multiplex (boolean multiplexers are defined similarly).

\paragraph{Oblivious aggregation.} Aggregation identifies groups of records based on a (composite) key and applies one or more aggregation functions to each group, updating the input table in place.
\ours implements oblivious aggregation using sort followed by a butterfly-style network \cite{B_goodrich2011data, A_blanton2012private, mpc-sorting} which requires $O(n\log n)$ operations and messages, $O(\log n)$ communication rounds, and $O(n)$ space per party, where $n$ is the cardinality of the input table. \camera{Protocol~\ref{alg:prefixagg} shows the  control flow of the aggregation, which resembles the Hillis-Steele network \cite{hillis1986data}.  We give a proof of correctness in Appendix \ref{apdx:agg}}. \ours also provides an implementation of the asymptotically-optimal Brent-Kung network \cite{brentkung}.

The algorithm expects a table $T$ sorted on column $K$ (the grouping key) and compares keys at distance $d$, while doubling the distance after every iteration.  Each step applies an aggregation function $f$ to pairs of input values and multiplexes the result (in-place) into column $G$, based on whether the respective records belong to the same group $(b=1)$.
For simplicity, the pseudocode assumes a single key column $K$ and a single aggregation function $f$. In practice, however, \ours's operator expects one or more key columns and can evaluate more than one aggregation function \emph{on the fly} (on the same or different columns) by reusing $b$ (line~\ref{line:bit_b}). 
In~\S\ref{sec:join}, we show how \ours leverages this functionality to reduce the cost of joins followed by aggregation.

\begin{algorithm}[t]
\caption{\camera{\textsc{Aggregation Network (\protoPrefixAgg{})}}}\label{alg:prefixagg}
\Input{Table $T$ with n rows, group key column $K$, function~$f$, input column $A$, result column $G$}
\Result{Table $T$ with one value in $G$ per unique $K$}
$r.G \gets r.A$ \Com{\ Copy input into result}
\For{$d = 1;\;d < n - d;\;d\gets 2\cdot d$}{
	\For{each pair of rows ($r_i$, $r_{i+d}$), $0 \leq i < n-d$}{
		$b~\gets~r_i.K == r_{i+d}.K$\Com{\ Compare keys}\label{line:bit_b}
		$g~\gets~f(r_i.G,~r_{i+d}.G)$\Com{\ Aggregate}
		$r_{i+d}.G~\gets~\texttt{Mux}(b,~r_{i+d}.G, ~g)$\Com{\ Update}
	}
}    
\end{algorithm}

\section{Sorting protocols for tabular data}\label{sec:sort}

\tocs{

We implement two oblivious sorting protocols, quicksort and radixsort, each competitive with the state of the art in terms of performance.
Both algorithms rely on state-of-the-art oblivious shuffling~\cite{ahi22-radixsort, permutation-correlations}, which we generalize to work with all MPC protocols in \ours.

In \S \ref{sec:shuffling-primitives}, we give background on the oblivious shuffling primitives upon which our sorting protocols are based.
In \S \ref{sec:base-sorting}, we describe the ``base'' version of each protocol as an in-place, ascending-order sorting protocol that requires unique keys.
In \S \ref{sec:wrapper}, we describe a wrapper protocol that removes these restrictions and allows us to extract the sorting permutation applied to the input as a secret-shared permutation, which we then use to sort multiple columns in a table through our \textsc{TableSort} protocol in \S \ref{sec:table-sort}.
In Appendix \ref{sec:radixsort-analysis}, we compare our radixsort protocol in detail with the prior state of the art, and Appendix \ref{sec:bound-quicksort} discusses a probabilistic bound on the amount of preprocessing required for quicksort when using the two-party ABY \cite{aby} protocol.

\paragraph{Note on sorting complexity in MPC}
The lower bound on the complexity of general-purpose comparison-based sorting algorithms for input size $n$ is $\Omega(n \log n)$ comparisons. In MPC, comparisons are expensive operations, typically requiring $\Omega(\ell)$ communication, where $\ell$ is the bitwidth of the elements.
Hence, quicksort in MPC requires $\Omega(n \log n)$ comparison operations over $\ell$-bit strings, for a total of $\Omega(\ell n \log n)$ bits of communication.
Alternatively, sorting networks require $O(n \log^2 n)$ comparisons and $\Omega(\ell n \log^2 n)$ communication \cite{Batcher68}.
Finally, radixsort~\cite{ahi22-radixsort} has either $O(\ell^2 n)$ or $O(\ell n \log n)$ communication depending on subtleties in the protocol that we discuss in Appendix \ref{sec:radixsort-analysis}.

In the general case, where the plaintext vector might contain unique elements, our input space must have a bitwidth of at least $\ell = \Omega(\log n)$. Both quicksort and radixsort thus require $\Omega(n \log^2 n)$ communication, while sorting networks require $\Omega(n \log^3 n)$ communication. Hence, $\Omega(n \log^2 n)$ communication is the gold standard of general MPC sorting protocols.

\subsection{Oblivious shuffling primitives}\label{sec:shuffling-primitives}
While a detailed discussion of our oblivious shuffling functionality is presented in Appendix~\ref{apdx:shuffle}, here we outline the basic primitives and notation that are necessary to understand our sorting protocols.

\tocsCam{
We support two forms of secret sharing for permutations: \textit{sharded permutations}, where the permutation is split into individually random permutations that compose to form the underlying permutation ($\langle \pi \rangle=\pi_n\circ\cdots\circ\pi_1$), and \textit{elementwise permutations}, written as $\share \pi$, where the destination indices of the permutation are secret shared as any other vector (either arithmetic or boolean).
}
Additionally, we support $\texttt{applyElementwisePerm}$ to apply an elementwise permutation to a vector, $\texttt{invertElementwisePerm}$ for in-place inversion of the underlying permutation, and $\texttt{convertElementwisePerm}$ to convert between a boolean and arithmetic secret sharing. We will generally use $\pi$ to refer to random or unstructured permutations, and we will use $\sigma$ to refer to sorting permutations.

\paragraph{Example.} We briefly walk through an example of oblivious shuffling using a sharded permutation. Consider the two-party setting with parties $P_0$ and $P_1$, and let the input be a secret sharing of the vector $(1, 2, 3, 4)$. Each party generates a local random permutation: $P_0$ generates $\pi_0 = (3, 1, 2, 4)$, and $P_1$ generates $\pi_1 = (1, 4, 3, 2)$. The parties will take turns applying their permutation to the secret-shared input obliviously (that is, without revealing to the other party which permutation has been applied) using the protocol by Peceny et al. \cite{permutation-correlations}. The outcome is that first we apply $\pi_0 = (3, 1, 2, 4)$ to the vector $(1, 2, 3, 4)$, yielding $(2, 3, 1, 4)$, and then we apply $\pi_1 = (1, 4, 3, 2)$ to the intermediate vector to obtain the final shuffled vector $(2, 4, 1, 3)$.

\subsection{Base sorting protocols}\label{sec:base-sorting}

\ours provides two base sorting protocols: an iterative version of quicksort and a radixsort protocol that combines features of Asharov et al.~\cite{ahi22-radixsort} and Bogdanov et al.~\cite{sharemind-radixsort}.
The protocols only sort in ascending order, and quicksort requires unique keys in the sorting column; these restrictions simplify the protocols and allow for more effective optimization.

\subsubsection{Quicksort} \label{sec:quicksort}

Quicksort is typically considered to be the most efficient general sorting algorithm in practice. However, it typically has a recursive and data-dependent control flow, both of which make a na\"ive application of quicksort to the MPC setting problematic in both efficiency and security. The typical instantiation of the algorithm would incur $O(n)$ rounds of comparisons, which is prohibitively slow. 
To address this challenge, we devise an iterative control flow that applies all independent comparisons at a given recursion depth simultaneously, facilitating vectorization and reducing the number of rounds to $O(\log n \log \ell)$.

To benefit from the efficiency of quicksort in the MPC setting, we use the \textit{shuffle-then-sort} paradigm by Hamada et al.~\cite{shuffle-then-sort,10.1145/3460120.3484560} in which we obliviously shuffle the input vector first. As a result, the subsequent data-dependent sorting protocol only reveals the (meaningless) order of the shuffled vector. Thus, these random bits provably leak nothing about the data, and the control flow of the comparisons in quicksort is entirely independent of the input vector (this should not be confused with one-bit leakage protocols, such as the protocol by Zhang et al.~\cite{DBLP:conf/sp/ZhangGYZYW24}).

The resulting protocol is described in Protocol \ref{iterative-quicksort} below.
All sorting happens in place, and the protocol alternates between local computation and rounds of full vector comparisons.
The protocol initializes a plaintext set of pivots $S$ with the first element of the vector and iterates until $S$ contains all elements in the vector.
In each iteration, we compare every element at index $i$ with one of the pivots: namely, the highest pivot index $j$ such that $j \leq i$ via a function $j \leftarrow \texttt{prevPivot}(i)$ that can be computed locally by each party. We create a secret-shared vector $\share{\vec{c}}$ of all previous-pivot elements, perform all comparisons between $\vec{x}$ and $\vec{c}$ in parallel, and omit comparisons for pivot elements (which would be compared with themselves).

After each round of secure comparisons, we reveal the result vector $\vec{r}$ and then partition the vector based on the desired pivot locations.
Concretely, for each pivot $p$,
we define a partitioning subroutine $\texttt{partition}(\share{\vec{x}}, p, i, \vec{r})$ that takes as input the vector $\share{\vec{x}}$, the start index $p$ where the pivot is located, an end index $i$, and $\vec{r}$.
The subroutine moves items in the subvector according to standard quicksort: all items less than the pivot, then the pivot itself, and all items greater than the pivot. It returns a set $S'$ of up to three locations for the pivot position and the next-round pivots for the left and right halves.
This step is entirely local since $S'$ and $\vec{r}$ are common knowledge.

To improve concrete efficiency, we can represent $S$ as a scaled indicator vector, $S_i=i\cdot\mathds I[i\text{ is a pivot}]$; i.e., if index $i$ is a pivot, index $i$ of $S$ contains $i$; otherwise, 0. Then, \texttt{prevPivot} is implementable via a single call to a prefix-max function (which computes a running maximum of a list). Likewise, the result of \texttt{partition} is no longer used to compute $S\cup S'$; instead, we update the new indices to contain their value, $\forall j \in S':S_j=j$.

\begin{algorithm}[ht]
    \caption{Quicksort}
    \label{iterative-quicksort}
    \SetKwInOut{Input}{input}
    \SetKwInOut{Output}{output}
    \Input{$\share{\vec{x}}$ containing unique elements}
    \Output{$\share{\vec{y}} = \sigma(\share{\vec{x}})$}

    $\share{\vec{x}} \gets \texttt{shuffle}(\share{\vec{x}})$\Com{\ Oblivious shuffle}
    $S \leftarrow \{1\}$\Com{\ Initial set of pivots}
    \While{$|S| \neq n$}{
        $\share{\vec{c}} \leftarrow [\,]$\Com{\ List of prev pivots to compare}
        \For{$i$ from $1$ to $n$}{
            $j \leftarrow \texttt{prevPivot}(i)$\Com{\ Previous pivot}
            $\share{\vec{c}}_i \leftarrow \share{\vec{x}}_j$\Com{\ Set value to compare with}
        }
        $\share{\vec{r}} \leftarrow \texttt{compare}(\share{\vec{x}}, \share{\vec{c}})$\Com{\ Perform the oblivious comparisons}
        $\vec{r} \leftarrow \texttt{open}(\share{\vec{r}})$\Com{\ Reveal the results}
        $i \leftarrow n$\Com{\ End index of current partition}
        \For{$p \in \text{\tt reversed}(S)$}{
            $S' \leftarrow \texttt{partition}(\share{\vec{x}}, p, i, \vec{r})$\Com{\ Partition subvector around the pivot}
            $S \leftarrow S \cup S'$\Com{\ Update pivot set}
            $i \leftarrow p - 1$\Com{\ Update index}
        }
    }
    \Return{$\share{\vec{x}}$}
\end{algorithm}

To see (informally) why this is secure, consider the information revealed to the adversary through the opening of comparison results. After the oblivious shuffle, each comparison reveals the relationship between two randomly chosen (and unknown) elements in the input list. Combining the revealed comparison results from each iteration together amounts to revealing the permutation applied to the shuffled list. However, the random and unknown shuffling permutation perfectly masks this information, so the adversary learns nothing about the ordering of elements in the original list.

\paragraph{Example.} We will now look at a toy example of the first iteration. Consider the input list $\vec{x} = (40, 50, 20, 80, 70, 30, 10, 60)$. The first step of the algorithm is to obliviously shuffle the list. Say the permutation applied is $\pi = (6, 1, 2, 3, 4, 8, 5, 7)$, which gives the shuffled vector $\pi(\vec{x}) = (50, 20, 80, 70, 10, 40, 60, 30)$. In the first iteration, we always choose the first element, $50$, as our initial pivot and compare it with all elements, giving us the vector $(\pi(\vec{x}) \geq 50) = (1, 0, 1, 1, 0, 0, 1, 0)$. We then (locally) rearrange the elements according to the indicator vector, so values greater than or equal to $50$ are moved to the back, yielding $(20, 10, 40, 30, 50, 80, 70, 60)$. The list is now correctly partitioned, where all elements preceding $50$ are less than $50$, and all elements following $50$ are greater than $50$. In the next iteration, our set of pivot elements will additionally contain $20$ and $80$, as they are the first elements in the newly created partitions.

We remark that Protocol \ref{iterative-quicksort} is only secure if all input elements in $\share{\vec{x}}$ are unique, so that each comparison check in $\vec{r}$ is equally likely to be true or false based only on the random shuffle and not the data $\vec{x}$. We remove this limitation in \S \ref{sec:wrapper} by making each element unique.

\subsubsection{Radixsort} \label{sec:radixsort}

Radixsort sorts a vector one bit at a time, moving from the least to most significant bit. 
Bogdanov et al.~\cite{sharemind-radixsort} permute the entire vector after each round of bit comparisons, while the later approach by Asharov et al.~\cite{ahi22-radixsort} composes the permutations bit by bit and applies them to the vector at the end, which is similar to what we do in \textsc{TableSort} at the table level (\S \ref{sec:table-sort}).
For radixsort, we have found that a hybrid approach which combines the eager technique by Bogdanov et al.~\cite{sharemind-radixsort} with the shuffling primitives by Asharov et al.~\cite{ahi22-radixsort} reduces the number of rounds for a small increase in bandwidth and outperforms both protocols by up to $44\%$ in practice. \ours implements a vectorized version of this hybrid approach. We provide a detailed analysis in Appendix~\ref{sec:radixsort-analysis}.

The radixsort algorithm is shown in Protocol \ref{radix-sort}. Our protocol takes two inputs beyond the data vector: the total number of bits to sort, $\ell$, and the number of least significant bits to skip, $\ell_s$.
The basic idea is simple: for each bit of the sorting key, we invoke the $\texttt{genBitPerm}$ protocol of Asharov et al., which returns an elementwise sharing of the sorting permutation for that bit, $\sigma_i$. Although $\texttt{genBitPerm}$ was proposed in the honest-majority three-party setting, we observe that it is actually agnostic to the setting and the number of parties, and it only makes black-box use of basic MPC primitives; therefore, it applies just as well in our dishonest-majority and malicious settings. We then apply the permutation to the full vector by invoking the $\texttt{applyElementwisePerm}$ protocol.

\begin{algorithm}[ht]
    \caption{Radixsort}
    \label{radix-sort}
    \SetKwInOut{Input}{input}
    \SetKwInOut{Output}{output}
    \Input{$\share{\vec{x}}$, $\ell$, $\ell_s$}
    \Output{$\share{\vec{y}} = \sigma(\share{\vec{x}})$}

    $\share{\vec{y}} \leftarrow \share{\vec{x}}$ \\
    \For{$i$ from $1$ to $\ell$}{
        $\share{\sigma_i} \leftarrow \texttt{genBitPerm}(\share{\vec{y}} \gg (i + \ell_s))$
        \Com{\ Single bit sort}
        $\share{\vec{y}} \leftarrow \texttt{applyElementwisePerm}(\share{\vec{y}}, \share{\sigma_i})$\Com{\ Permute according to sort}
    }

    \Return{$\share{\vec{y}}$}
\end{algorithm}

\paragraph{Example.} We demonstrate the radixsort protocol with a small example. Consider the following vector of two-bit elements, written as bitvectors: $(01, 10, 01, 11, 10, 00, 10, 00)$.
First, we isolate the least significant bit to obtain $10110000$. We sort the bitvector to obtain $00000111$, which gives us the applied permutation $\sigma_0 = (6, 1, 7, 8, 2, 3, 4, 5)$. We then apply $\sigma_0$ to the original input to sort the input by the least significant bit: $(10, 10, 00, 10, 00, 01, 01, 11)$. Next, we do the same to the most significant bit. We isolate the most significant bit: $11010001$. Then, we sort the bit to obtain $00001111$, giving us the applied permutation $\sigma_1 = (5, 6, 1, 7, 2, 3, 4, 8)$. We apply $\sigma_1$ to the current full vector to obtain our output: $(00, 00, 01, 01, 10, 10, 10, 11)$. We can see that the total permutation that has been applied from input to output is $\sigma_1 \circ \sigma_0$.

In the above example, we assumed that we could efficiently sort a single-bit vector, so we now demonstrate how that can be achieved. The insight (following from \cite{ahi22-radixsort}) is that a single-bit sort can be decomposed into two prefix sums and a multiplexing. First, we negate every bit in the vector, and a prefix sum of the result yields a vector of destination indices for each value that was original a zero. Consider the example in the prior paragraph of the least significant bit: $10110000$. We invert each bit to obtain $01001111$, and a prefix sum yields $(0, 1, 1, 1, 2, 3, 4, 5)$. This matches the permutation $\sigma_0$ for indices where the original bit was a zero. We apply another prefix sum to the original (non-negated) vector to obtain $(1, 1, 2, 3, 3, 3, 3, 3)$. Here, each element corresponding to a one in the input is the destination index of that element if we index starting from the first position that will contain a one. That is, we add $5$ (the number of zeros, obtained from the last element of the first prefix sum) to all elements, and we obtain $(6, 6, 7, 8, 8, 8, 8, 8)$. Finally, we multiplex between the two vectors, so if bit $i$ in the input is a zero, we choose the $i^\mathrm{th}$ element from the first list, and if bit $i$ is a one, we choose the $i^{th}$ element from the second list. The output is the permutation $\sigma_0 = (6, 1, 7, 8, 2, 3, 4, 5)$.

\subsection{Wrapper sort protocol} \label{sec:wrapper}

We now describe a general sorting wrapper protocol, \textsf{VectorSort}, that takes as input a boolean secret-shared vector $\share{\vec{x}}$ to sort, the order in which to sort (either ascending (\texttt{ASC}) or descending (\texttt{DESC})), and a base sorting protocol (either quicksort or radixsort). As output, it returns the sorted vector $\share{\vec{y}} = \sigma(\share{\vec{x}})$ and an elementwise sharing of the sorting permutation $\share{\sigma}$.
The wrapper protocol performs two important operations: it pads the input before sorting, and it extracts the sorting permutation after sorting

\newcommand{\brackets}[1]{\share{#1}\xspace}

\paragraph{Input padding.}
Following Hamada et al.~\cite{shuffle-then-sort}, for each element of the input vector $\vec{x}$, we append its index. This padding ensures that the elements are unique, that we can extract the sorting permutation, and that the sort is \emph{stable} in the sense that if $x_i = x_j$ with $i < j$, then $x_i$ will remain before $x_j$ in the sorted vector.
Concretely, to sort an input vector $\brackets{\vec{x}}$ in ascending order, we form the concatenated string $\brackets{x_i'} = \brackets{x_i} \; || \; \brackets{i}$, with the sharing of $i$ being a local operation using a protocol $\texttt{publicShare}$, which secret-shares a publicly known vector without communication. To ensure that all $n$ elements in $\vec{x'}$ are unique, we must add at least $\lceil \lg n \rceil$ bits of padding; our implementation always adds 32 bits of padding.

\paragraph{Permutation extraction.}
After the base sort of $\vec{x'}$ is completed, the $i^{\mathrm{th}}$ element in the input vector, with the padded value $i$, winds up at position $\sigma_i$. We then extract the sorting permutation from the padding bits so that we can apply it to other columns in a table; to the best of our knowledge, no prior works have detailed how to extract the sorting permutation from the padding. We separate the sorted list $\vec{y'}$ into the data $\vec{y} = \sigma(\vec{x})$ and a permutation $\pi$. The padding bits contain the result of applying the sorting permutation $\sigma$ to the identity permutation, so $\pi = \sigma(\mathds{I}_n)$, where $\mathds{I}_n$ represents the identity permutation $(1, \dots, n)$. By Fact \ref{perm-fact} (see Appendix \ref{apdx:shuffle}), $\sigma(\mathds{I}_n) = \mathds{I}_n \circ \sigma^{-1} = \pi$, so $\sigma = \pi^{-1}$. Hence, we obtain the sorting permutation $\sigma$ by inverting $\pi$ using the $\texttt{invertElementwisePerm}$ protocol.

\paragraph{Descending order.}
Sorting an input vector $\vec{x}$ in descending order can be done similarly.
As a thought experiment: it almost works to use the ascending-sort protocol above and locally reverse the output vector, except it breaks stability by reversing the tiebreak condition. That is, if $x_i = x_j$ and $i < j$, then ascending-sort preserves stability, but the local reverse ensures that $x_i$ always comes after $x_j$ instead.

To solve this issue, we negate the padding for each element. That is, each secret-shared element $\brackets{x_i}$ now gets padded to $\brackets{x_i'} = \brackets{x_i} \; || \; \brackets{-i}$.
Since we are attaching the padding bits as the least significant bits, we need any two values $-i$ and $-j$ where $i < j$ to satisfy $-i > -j$ even when interpreted as an unsigned integer. Our implementation uses $\pm\{1,\dots,n\}$ for the padding.

\paragraph{Putting it all together.}
We write the general sorting wrapper in Protocol \ref{sorting-wrapper} and summarize it here. We start by padding the input $\brackets{\vec{x}}$ with $\mathds{I}_n$ in order to perform an ascending sort for input size $n$, or with $-1 \cdot \mathds{I}_n$ to sort in descending order. The resulting padded input vector $\share{\vec{x'}}$ is fed into the base quicksort or radixsort protocol from \S \ref{sec:base-sorting}, which results in the padded sorted vector $\share{\vec{y'}}$.
Finally, we locally split $\share{\vec{y'}}$ into the unpadded result $\share{\vec{y}}$ and a permutation $\share{\pi}$, and invert $\share{\pi}$ to obtain the sorting permutation $\share{\sigma}$.
If sorting in descending order, before inverting $\share{\pi}$, we convert it from boolean to arithmetic sharing to negate each element in the permutation.
We return both $\share{\vec{y}}$ and $\share{\sigma}$.

\begin{algorithm}[t]
    \caption{\textsf{VectorSort}}
    \label{sorting-wrapper}
    \SetKwInOut{Input}{input}
    \SetKwInOut{Output}{output}
    \Input{$\share{\vec{x}}$, $\texttt{order} \in \{\texttt{ASC}, \texttt{DESC}\}$, $\texttt{sort} \in \{\texttt{quicksort}, \texttt{radixsort}\}$}
    \Output{$\share{\vec{y}} = \sigma(\share{\vec{x}})$, $\share{\sigma}$}

    $\vec{p} \leftarrow \mathds{I}_n$\Com{\ Padding is the identity permutation}
    \If {$\texttt{order} == \texttt{DESC}$}{
        $\forall i \in [n], p_i \leftarrow -p_i$\Com{\ Negate padding if sorting in reverse}
    }
    $\share{\vec{p}} \leftarrow \texttt{publicShare}(\vec{p})$\Com{\ Secret-share padding, known to both parties}
    $\forall i \in [n], \share{\vec{x'}}_i \leftarrow \share{\vec{x}}_i \; || \; \share{\vec{p}}_i$\Com{\ Update data vector with padding bits}
    $\share{\vec{y'}} \gets \texttt{sort}(\share{\vec{x'}})$\Com{\ Do the ascending-order sort}
    $\share{\vec{y}} \; || \; \share{\pi} \leftarrow \texttt{split}(\share{\vec{y'}})$\Com{\ Separate sorted data from sorted padding}
    \If {$\texttt{order} == \texttt{DESC}$}{
    \Comment{\ Type conversion on $\pi$ so we can negate it}
        $\share{\pi} \leftarrow \texttt{convertElementwisePerm}(\share{\pi}, \texttt{B}\to\texttt{A})$\\
        \Comment{\ $\pi$ represents permuted padding, convert back to domain $[n]$}
        $\forall i \in [n], \share{\pi}_i \leftarrow -\share{\pi}_i$\\
            }
    \Comment{\ Permuted padding is the inverse of the sorting permutation}
    $\share{\sigma} \leftarrow \texttt{invertElementwisePerm}(\share{\pi})$\\    \Return{$\share{\vec{y}}$, $\share{\sigma}$}
\end{algorithm}

Next, we describe how to use Protocol~\ref{sorting-wrapper} to efficiently sort a table with multiple columns.

}

\subsection{Table sort}\label{sec:table-sort}

\tocs{\ours sorts tables in place in Protocol~\ref{alg:sorting-wrapper} (\textsc{TableSort}) using the sorting wrapper protocol \textsf{VectorSort}. Each call to $\texttt{sort}$ executes the sorting wrapper protocol, and the output is the secret-shared sorting permutation. The protocol expects a table and a subset of columns ($1$ to $k$) to use as sorting keys, along with directions to sort each sorting key.}
The strawman approach is to sort the entire table for each key, which would incur expensive communication under MPC. \camera{This strawman is commonly used in other MPC systems that sort multiple columns, such as Secrecy \cite{secrecy}.} \textsc{TableSort} avoids this overhead as follows: it extracts sorting permutations (i.e., mappings from old to new positions) from key columns in line \ref{line:single-column-sort}, composes them from right to left (line 6), and applies the final permutation to all columns of the table once at the end~(lines~7-8).

\tocs{
Unlike prior works, \textsc{TableSort} uses different base sorting protocols depending on the key bitwidth.}
By default, it uses quicksort for large bitwidths (e.g., $\ell = 64$) and radixsort for small bitwidths ($\ell \leq 32$).

\begin{algorithm}[t]
    \caption{\textsc{TableSort}}
    \label{alg:sorting-wrapper}
    \Input{Table $T$ with $p$ columns $\{C_1, \dots, C_p\}$, integer $k \leq p$, orders $o =\{o_1, \dots, o_k\}$, $o_i \in \{\texttt{ASC}, \texttt{DESC}\}$}
    \Output{Table $T$ sorted on $\{C_1, \dots, C_k\}$}

    $\sigma \gets \mathbb{I}$\Com{\ Identity permutation}
    \For{$i$ from $k$ to $1$ descending}{
    	$t~\gets~C_i$\Com{\ Temp copy of column i}
        $\texttt{applyPerm}(t,~\sigma)$\Com{\ Permutation so far}
        $\_, \sigma' \gets  \textsf{VectorSort}(t,~\texttt{o}_i)$\Com{\ Sort temp column} \label{line:single-column-sort}
        $\sigma \gets \texttt{composePerm}(\sigma,~\sigma')$\Com{\ Build iteratively}
    }
    \Comment{\ Apply final permutation to all columns}
    \For{$i$ from $1$ to $p$}{
        $\texttt{applyPerm}(C_i,~\sigma)$
    }
\end{algorithm}

\section{A composite join-aggregation operator}\label{sec:join}

\tocs{
    In this section, we introduce our approach to evaluating the composite join-aggregation operator, its extensions (\S\ref{sec:generalization}--\ref{sec:duplicates}), and walk through a step-by-step example of a multi-way join query (\S\ref{sec:full-example}).
} \tocsCam{Our new operator runs in $O(n \log n)$ time, a significant speedup over the quadratic-time Cartesian product, and, crucially, allows us to bound the size of intermediate results and chain multiple joins together to execute complex analytics.}

\tocs{
\paragraph{Notation.}
To simplify the presentation, we avoid secret-sharing notation and describe the control flow of operators as if they are applied to plaintext data. We denote a column $C$ in table $T$ as $T.C$. For a row $r \in T$, we denote its $C$ value as $r.C$. When saying \emph{``an operator is applied to table $T$,''} we mean that \ours parties perform MPC operations on the secret-shared table $T$,  and produce secret shares~of~the~result.
}

Let $L$, $R$ be two tables with dimensions $n \times p$ and $m \times q$, respectively.
$L.V$ and $R.V$ are the special validity columns.  
Without loss of generality, we assume that $L$ and $R$ are joined on $j$ common columns $K = \{K_1, \dots, K_j$\}, $ 0 < j \leq p, q$.
We collectively refer to the columns in $K$ (and their values) as the \emph{join keys} $L.K$ and $R.K$. The join keys can be arbitrary columns and should not be confused with Primary Keys~(PK) or Foreign Keys~(FK). Our operator is more general than a PK-FK join, in that it handles cases where rows in either input may not have matching keys on the other side. It is also not equivalent to a sort-merge join, in that it computes a join and an (optional) aggregation within the same oblivious control flow (to reduce MPC costs). Instead, be build a fused join-aggregation operator based on oblivious sort and an aggregation network.

\begin{figure*}[t]
    \centering
            \begin{subfigure}{\textwidth}
            \centering
            \includegraphics[width=\textwidth]{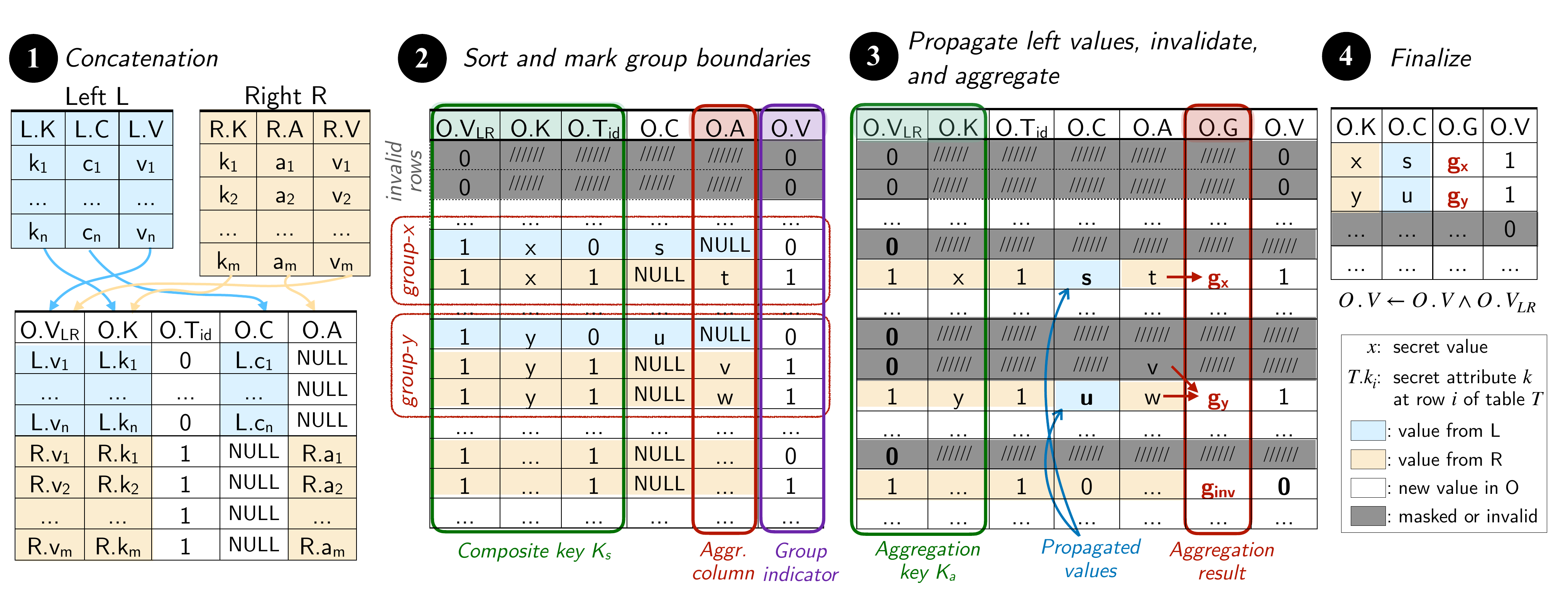}
            \end{subfigure}
         \caption{Steps of the basic \ours operator. The input tables are concatenated and sorted to bring valid rows with the same join key $K$ (e.g., \texttt{x} and \texttt{y}) next to each other. Values in $O.C$ are copied to matching rows from $R$. Values in $O.A$ are aggregated per key $K$ and stored in $O.G$.}\label{fig:approach}
\end{figure*}

\paragraph{Overview of basic operator.} 
We first describe the basic skeleton for a simple equality join that requires $L.K$ to be unique, but allows $R.K$ to contain duplicates.
Later on, we discuss how to further generalize the operator to other join types, including joins with duplicate keys in both inputs.
The operator joins $L$ with $R$ (denoted as $L \bowtie_K R$) to identify all pairs of rows from $L$ and $R$ with the same key $K$.
The result is a new table $O$ that includes all columns from $R$ along with \tocs{the specified} columns $C$ from $L$. 
Given that keys in $L$ are unique, the cardinality of table $O$ is at most $m=|R|$, enabling \ours to chain multiple such join operators while keeping memory bounded. 
If an aggregation function is provided, the operator also applies the function to each group of rows in $O$ with the same key~$K$.
Protocol~\ref{alg:inner-join} shows the pseudocode and Figure~\ref{fig:approach} depicts the internal steps of the operator.

\paragraph{Step 1: Concatenation.} In the first step, the operator initializes the output table $O$ with $n+m$ rows and $p+q-1$ columns from the two input tables. It combines the validity columns $L.V$ and $R.V$ into $O.V_{LR}$ and the join key columns $L.K$ and $R.K$ into $O.K$. Next, it appends an auxiliary column $O.T_{id}$ and initializes rows $[1 \dots n]$ with $0$ and rows $[n+1 \dots n+m]$ with $1$, indicating the origin table ($L$ and $R$, respectively) of each row.  Finally, it expands the non-key columns from $L$ and $R$ with $n$ (resp. $m$) \texttt{NULL} values and appends them to $O$. Figure~\ref{fig:approach} shows a concatenation example.  For simplicity, we assume one key ($K$) and one non-key column per table ($C$ and $A$ resp.). The resulting table $O$ has $p+q-1=5$ columns.

\paragraph{Step 2: Sort and mark group boundaries.}
The operator sorts table $O$ on the composite key $K_{s} = \{O.V_{LR}$,~$O.K$,~$O.{T_{id}}\}$ using the table sort protocol from~\S\ref{sec:sort}. 
This way, valid rows with the same key $K$ are clustered together, and each row from $L$ becomes the first row of its respective $K$ group.
Each group of rows in $O$ has at most one row from $L$ and zero or more rows from $R$.
Next, the operator calls oblivious deduplication from~\S\ref{sec:approach} with key $K_a = \{O.V_{LR}$,~$O.K\}$ to mark the first row of each group. 
It then flips the resulting bit and ANDs it with $O.V_{LR}$ to update the validity column $O.V$.  
The first row of each group now has $O.V = 0$.

\begin{algorithm}[t]
\caption{\textsc{Join-Aggregation (JoinAgg)}}\label{alg:inner-join}
\Input{Tables $L$, $R$; keys $K = \{K_1, \cdots, K_j\}$, columns $C = \{C_1, \cdots, C_k\}$ to propagate from $L$ to $R$, 
and an aggregation function $f_\mathrm{agg}$ to apply to column $A$}
\Result{Table $O$}\vspace{1mm} 
\Comment{\ Step 1: Initialize table $O$}
$O \gets \textsc{Concat}(L, R)$ \Com{\ Schema shown in Fig.~\ref{fig:approach}}
\Comment{\ Sorting \& Aggregation keys}
let $K_{s} = \{O.V_{LR},~O.K,~O.T_{id}\}$, $K_{a} = \{O.V_{LR},~O.K\}$\\
\Comment{\ Step 2: Sort table $O$ on $K_s$ (ASC)}
$\textsc{TableSort}(O, K_{s})$\\ \label{line:sort}
$O.V~\gets~O.V_{LR}~\land~\neg \textsc{Distinct}(O,~K_a)$\Com{\ Mark groups}\label{line:valid}
\Comment{\ Step 3: Join and aggregate per group}
\protoPrefixAgg{}($O,~\{K_{a},~C,~f_\mathrm{copy}\},$\Com{\ Copy $C$ into $R$}\nonl\label{line:prefixagg}
{$\quad\{K_{s},~V,~f_\mathrm{copy}\},$\Com{\ Propagate $V$}}\nonl
$\quad\{K_{s},~A,~f_\mathrm{agg}\})$\Com{\ Aggregate}
\Comment{\ Step 4: Remove redundant rows and columns}
$\textsc{Finalize}(O)$
\end{algorithm}

\paragraph{Step 3: Propagate values from $L$, invalidate rows, and apply aggregation.} 
In this step, the operator calls the \protoPrefixAgg{} protocol from~\S\ref{sec:approach} to apply two functions: $f_\mathrm{copy}$ and $f_\mathrm{agg}$. Function $f_\mathrm{copy}$ is an internal function that takes as input a pair of values from $O.C$ and simply returns the first one, i.e., $f_\mathrm{copy}(x,y) = x$. It is used by \ours to 
copy non-key values from $L$ into $R$ by obliviously populating the \texttt{NULL} values, and also to invalidate rows from $R$ that did not match with any row from $L$.
Function $f_\mathrm{agg}$ is a user-provided aggregation function. The operator applies $f_\mathrm{agg}$ to column $A$ and stores its result in a new column $O.G$. 
The function $f_\mathrm{agg}$ can be one of the common aggregations from~\S\ref{sec:workloads} or a custom function constructed as we explain in Appendix \ref{apdx:join}.

\paragraph{Step 4: Finalize result.} 
At this point, $O$ has $n+m$ rows; however, we know that the actual output of the join contains at most $m$ valid rows.
In the final step, \ours removes redundant columns (e.g., $O.V_{LR}$, $O.A$, $O.T_{id}$ in Figure~\ref{fig:approach}) and trims all unnecessary rows which came from $L$.
To remove these rows, we sort $O$ on $O.V$ (in descending order), placing valid rows above invalid ones, and trim the last $n$ rows from the sorted table.

In practice, we observe that the trimming operation is cheap but not costless, even when using our single-bit radixsort protocol: if $n\ll m$, the overhead of sorting by $O.V$ may be higher than the improvement gained in subsequent operations with the trimmed table. To address this tradeoff, we develop a heuristic that \emph{automatically} estimates whether trimming pays off, based on the sizes of the input tables. 
For example, in our semi-honest honest-majority protocol, the operator trims when $9m<n \lg n \lg \ell$, where $\ell$ is the share representation bitwidth. Using this (conservative) heuristic has improved end-to-end execution times modestly across all queries of \S\ref{sec:eval} and bounds the maximum working table size in queries with many joins. 
For a detailed analysis, see Appendix \ref{apdx:join}.

\subsection{Generalizing to different join types}\label{sec:generalization}

\textsc{JoinAgg} computes inner joins with equality predicates; however, its control flow 
can support different join types with only minor changes to the way we invalidate~rows. 
All operators of this section require unique keys on the left input (one-to-many) like the basic operator.  
We describe how to handle duplicates in both inputs in \S\ref{sec:duplicates}.  
For all join types we discuss, we provide correctness proofs in Appendix \ref{apdx:join}.

\paragraph{Semi-join.} A semi-join $L \ltimes_K R$ returns all rows in $L$ that matched with any row from $R$ on $K$. \ours implements the semi-join by simply running \textsc{JoinAgg} with the two input tables swapped.
After sorting $O$, rows from $R$ will now precede rows from $L$ in all groups that contain rows from both tables. The operator will therefore invalidate (and later trim) all rows from $R$ and will only invalidate those rows from $L$ that do not match with any row from $R$ (i.e., they appear at the beginning of a group). In all other cases, the validity of rows from $L$ remains as is (equal to $O.V_{LR}$), matching the semantics of semi-join without any changes in Protocol~\ref{alg:inner-join}.

\paragraph{Outer joins.} 
A left-outer join ($L\leftouterjoin_K R$) returns all pairs of rows from the two inputs that have the same key $K$ plus all remaining rows from $L$ that did not match with any row from $R$. To support a left-outer join, we only need to replace line~\ref{line:valid} in \textsc{JoinAgg} with:
$$
O.V~\gets~O.V_{LR} \land \neg (O.T_{id}~\land~\textsc{Distinct}(K_a))
$$
This will invalidate the first row of each group only if it comes from $R$. If the row comes from $L$, its $O.V$ is set equal to $O.V_{LR}$. This ensures that valid rows from $L$ do not get invalidated, even if they have no matches in $R$, but unmatched rows from $R$ will be invalidated as in the basic operator.

A right-outer join ($L\rightouterjoin_K R$) works the opposite way: it returns all pairs of rows from the two inputs that matched on $K$ plus all remaining rows from $R$ that did not match with any row from $L$.
In this case, we only need to replace line~\ref{line:valid} in Protocol~\ref{alg:inner-join} with $O.V \gets O.V_{LR} \land O.T_{id}$. The line invalidates all rows from $L$, since they have $O.T_{id}=0$, and leaves the validity of rows from $R$ intact (equal to $O.V_{LR}$). A right-outer join does not need to propagate valid bits and therefore omits this step while executing the \protoPrefixAgg{} protocol in line~\ref{line:prefixagg}.

Finally, the full-outer join ($A\fullouterjoin_K B$) does not invalidate any rows from either input. In this case, the operator sets $O.V \gets O.V_{LR}$ in line~\ref{line:valid} of Protocol~\ref{alg:inner-join} and does not propagate any valid bits, as in the right-outer join above. All outer joins return $n+m$ rows and never trim rows from $O$.

\paragraph{Anti-join.} An anti-join $L \rhd R$ returns all rows from $L$ that do not match any row from $R$. \ours implements the anti-join by executing Protocol~\ref{alg:inner-join} with the two input tables swapped and line~\ref{line:valid} replaced with $O.V \gets O.V_{LR} \land O.T_{id}$, as in right-outer join. However, here we must propagate valid bits from $R$ to $L$, which is done by invoking $f_{\mathrm{copy}}$. When a row exists in $R$, it will be sorted above the rows from $L$ and invalidated. Then, its valid bit ($=0$) will be copied down to all other rows, including those from $L$, in its group.

\paragraph{Theta-join.} \ours's operator can support one-to-many joins with $\theta$-predicates, as long as $\theta$ is conjunctive and contains one or more equality conditions, e.g., $$L.\texttt{Name} = R.\texttt{Name}~\land~L.\texttt{Time} \leq R.\texttt{Time}$$
In this case, \ours executes Protocol~\ref{alg:inner-join} using the equality predicate(s) that bound the output size and reduces all other predicates in $\theta$ into~oblivious~filters. 
\subsection{Supporting aggregations}\label{sec:composition}
In line~\ref{line:prefixagg} of \textsc{JoinAgg}, we execute an aggregation in the same oblivious control flow (\protoPrefixAgg{}) as the join.
In fact, \ours can evaluate multiple aggregation functions $f_\mathrm{agg}$ with keys $K_a$ on any combination of columns from $L$ and $R$.  If the query includes aggregations with keys different than $K_a$,  an additional invocation of \textsc{TableSort} followed by \protoPrefixAgg{} is necessary.  In case the join or aggregation keys are a prefix of the other,  we can sort on the largest prefix in line~\ref{line:sort} of \textsc{JoinAgg} to avoid the extra \textsc{TableSort} call.

\ours supports decomposable aggregation functions of the form
$$f_\mathrm{agg}(X)=f_\mathrm{final}(f_\mathrm{post}(f_\mathrm{pre}(X)))$$
where $f_\mathrm{pre}$ is a partial aggregation function over a set $X$, $f_\mathrm{post}$ combines partial aggregates, and $f_\mathrm{final}$ generates the final result. Self-decomposable functions are those with $f_\mathrm{final}=\mathbb{I}$, the identity function. Function $f_\mathrm{pre}$ must itself be self-decomposable, $f_\mathrm{pre}((X, Y)) = f_\mathrm{post}(f_\mathrm{pre}(X), f_\mathrm{pre}(Y))$. Examples include $\mathtt{min}$, $\mathtt{max}$, $\mathtt{sum}$ (where $f_\mathrm{post} = f_\mathrm{pre}$), and $\mathtt{count}$ (where $f_\mathrm{post}=\mathtt{sum}$).  Note that self-decomposability implies commutativity and associativity~\cite{jesus2014survey}, which are necessary for \protoPrefixAgg{} correctness. Decomposable aggregations enable efficient incremental computation and data-parallelism, and have been studied extensively in the context of sensor networks~\cite{10.1145/1061318.1061322} and dataflow systems~\cite{10.1145/1629575.1629600}. In \ours, we employ similar ideas in oblivious computation to bound the size of intermediate results, as we explain next.

\subsection{Supporting joins on duplicate keys}\label{sec:duplicates}
We can use \textsc{JoinAgg} to handle many-to-many joins (with duplicate keys on both sides), as long as the join is followed by a decomposable aggregation and the aggregation keys belong to one input table (see \S\ref{sec:workloads}). This ensures that the worst-case cardinality of the output table is equal to the cardinality of the table with the aggregation key(s). 

Decomposable functions are amenable to partial evaluation; therefore, we can apply them without having to materialize the Cartesian product of the two input tables.  As such, we apply $f_\mathrm{pre}$ before the join,  $f_\mathrm{post}$ after the join,  and, optionally, $f_\mathrm{final}$ to compute the final result.
Let $K_j$ and $K_a$ be the join and aggregation keys, respectively.
\ours applies function $f_\mathrm{pre}$ to one of the input tables to compute a partial aggregation on $K_j$. The result is a table with unique join keys that can now be joined with the second table using the basic operator. Function $f_\mathrm{post}$ is evaluated on the output table $O$ with aggregation key $K_a$. 

As an example, consider a join $L \bowtie_{K_j} R$ on $K_j$ with duplicate keys in both input tables followed by $f_\mathrm{agg} = \mathtt{count}$ that outputs the number of rows in the join result. We can obliviously compute the aggregation without a quadratic blowup by first evaluating function $f_\mathrm{pre} = \mathtt{count}$ on $L$ to count the number of rows per key $K_j$ (multiplicity). Then, we can join $L$ with $R$ using \textsc{JoinAgg} and apply $f_\mathrm{post} = \mathtt{sum}$ to add the multiplicities. The idea also applies to acyclic queries with multi-way~joins like in plaintext query evaluation~\cite{DBLP:conf/pods/JoglekarPR16}.
\subsection{Putting it all together}\label{sec:full-example}

\begin{figure}[h!]
    \centering
        \centering
        \includegraphics[width=0.8\columnwidth]{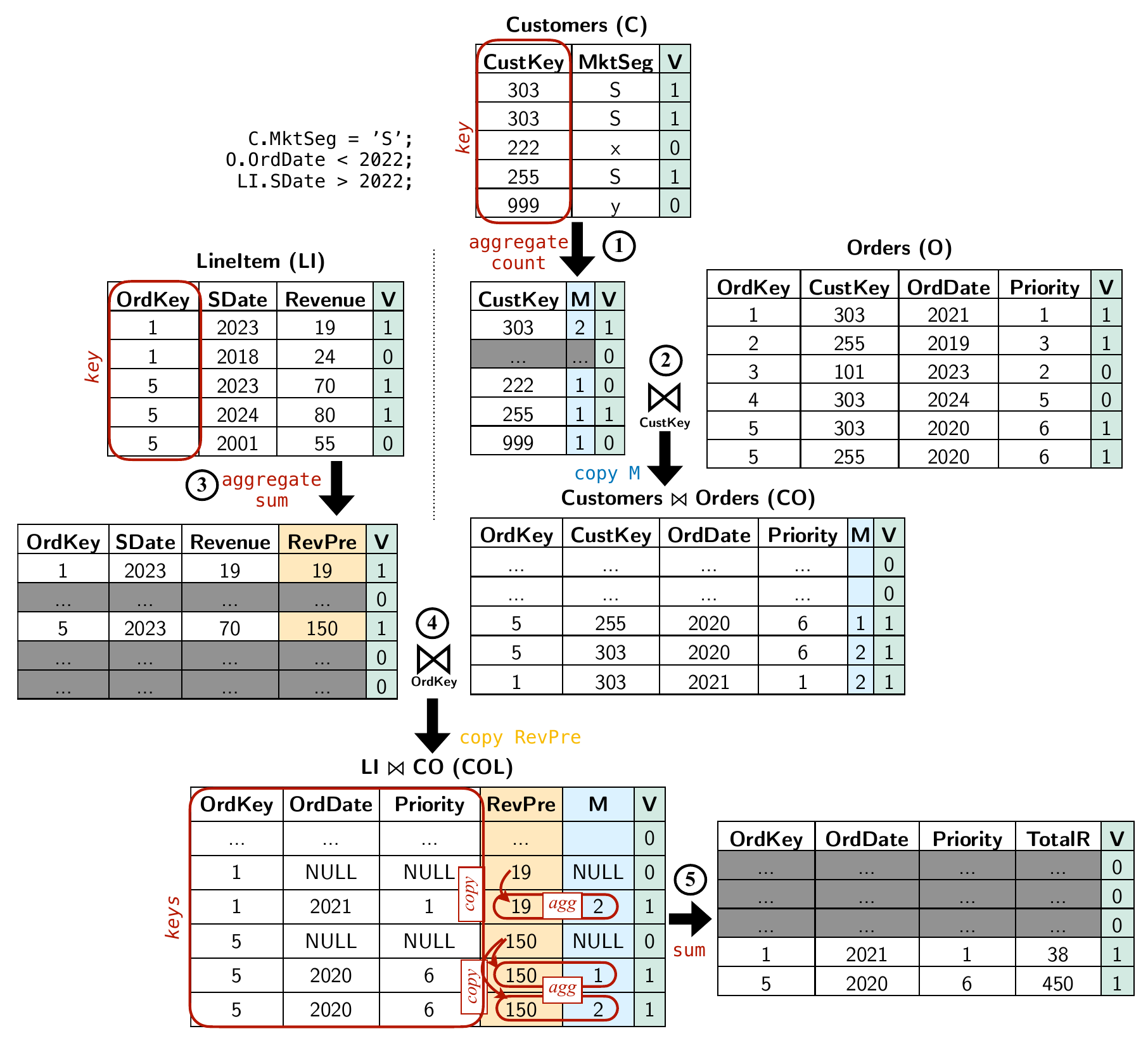}
        \caption{Step-by-step oblivious evaluation of TPC-H Q3 assuming no PK-FK constraints (all input tables have duplicate join keys).}\label{fig:q3}
\end{figure}

Figure~\ref{fig:q3} shows an example with many-to-many joins,  using a generalization of Q3 (see Listing~\ref{example}),  where we allow duplicates in all columns. 
That is, we assume no public knowledge of the PK-FK constraints in the TPC-H specification of this query. 
We omit filter operations to keep the example simple.

To evaluate $C \bowtie_\texttt{CustKey} O$, we first apply a pre-aggregation to $C$ that computes the multiplicity of each \texttt{CustKey} using Protocol~\ref{alg:prefixagg}. The pre-aggregation makes \texttt{CustKey} in $C$ unique and stores multiplicities in a new column \texttt{M}. Next, we join $C$ with $O$ on \texttt{CustKey} using Protocol~\ref{alg:inner-join} to produce a new table $CO$. The second join $LI \bowtie_\texttt{OrdKey} CO$ is evaluated similarly: we first apply a pre-aggregation to $LI$ that computes the \texttt{Revenue} per \texttt{OrdKey} and stores it in column \texttt{RevPre}. Then, we compute the join using Protocol~\ref{alg:inner-join} and apply the post-aggregation to sum the product \texttt{RevPre}*\texttt{M} per key $K_a = \{\texttt{OrdKey}, \texttt{OrdDate}, \texttt{Priority}\}$ using Protocol~\ref{alg:prefixagg}.
Listing~\ref{example2} shows the implementation using the \ours API.

\begin{lstlisting}[label={example2}, caption={The example of Figure~\ref{fig:q3} implemented in \ours},language=C++,float]
// Pre-aggregate and apply first join
auto CO = C.aggregate (   {"CustKey"}, {"M", "M", count})
           .inner_join(O, {"CustKey"}, {"M", "M", copy});
// Pre-aggregate and apply second join
auto COL = LI.aggregate (
                    {"OrdKey"}, {"Revenue", "RevPre", sum}
            ).inner_join(
                CO, {"OrdKey"}, {"RevPre",  "RevPre", copy});
// Apply post-aggregation
COL["TotalR"] = COL["RevPre"] * COL["M"];
auto RES = COL.aggregate(
                    {"OrdKey", "OrdDate", "Priority"},
                    {"TotalR", "TotalR", sum});
\end{lstlisting}

\section{System implementation}\label{sec:impl}

We have implemented \ours from scratch in C++. \tocs{Figure \ref{fig:stack} shows the software stack run by each computing party, which} includes (i) an MPC layer with vectorized implementations of secure primitives, (ii) a library of  oblivious relational operators,  (iii) a custom communication layer, and (iv) a data-parallel execution engine.  Users can configure the share size $2^{\ell}$ ($\ell=64$ by default), the MPC protocol, and the sorting protocol (quicksort by default), according to their application needs.

We use \texttt{libOTe}~\cite{libote} to produce random oblivious transfers, from which we generate OLE (Oblivious Linear Evaluation) correlations and Beaver triples, as described in \S \ref{sec:triples}. For permutation correlations, we use the technique by Peceny et al.~\cite{permutation-correlations}. For random number generation, we use \texttt{AES} from \texttt{libsodium}~\cite{libsodium}.

\paragraph{Protocol-agnostic secure operator abstractions.}
All relational operators, sorting, multiplexing, and division with private denominator are protocol-agnostic.
\tocs{
The protocol-specific primitives are the functionalities defined by each MPC protocol (e.g., $\times$, $\land$, etc.), division by public denominator, and two shuffling primitives that we use to generate correlated randomness (random sharded permutations).
Higher-level shuffle operators, such as composing, inverting, converting, and applying permutations, are protocol agnostic, and make use of these protocol-specific correlations. For more details, see Appendix \ref{apdx:shuffle}.
}

\paragraph{Vectorization and data parallelism.}
\ours uses a columnar data format that allows for efficient vectorized execution. 
Recall that each secure multiplication and bitwise AND operation incurs a message exchange between parties.
With vectorization, \ours can efficiently group independent messages within operations and exchange them in the same network call. All low-level primitives and oblivious operators in \ours are vectorized to work this way.

Each \ours server runs as a single process with the same (configurable) number of threads $\tau$. 
Execution within a server is data-parallel: a coordinator thread assigns tasks to workers, which operate on disjoint partitions of the input vector.
\ours's communication layer establishes a static mapping between workers, so that thread with id $t$ communicates only with thread $t$ in other servers, $0 \leq t < \tau$.

\paragraph{Communication layer.}
In \ours, we implement a TCP-based communicator, tailored to I/O-intensive MPC execution. The communicator manages network sockets via a configurable number of dedicated threads.  A high number of communication threads can significantly improve performance in WAN settings, where multiple parallel connections between parties can be used to leverage available bandwidth.

Worker threads interact with communication threads using lock-free ring buffers to avoid blocking computation. Communication threads are decoupled from worker threads and process ring buffers in a round-robin fashion. Data copies are avoided both ways; worker threads pass offsets in vectors they operate on, and the communicator directly pulls from (on send) and pushes to (on receive) these vectors.

\begin{figure}
    \centering
    \includegraphics[width=0.8\columnwidth]{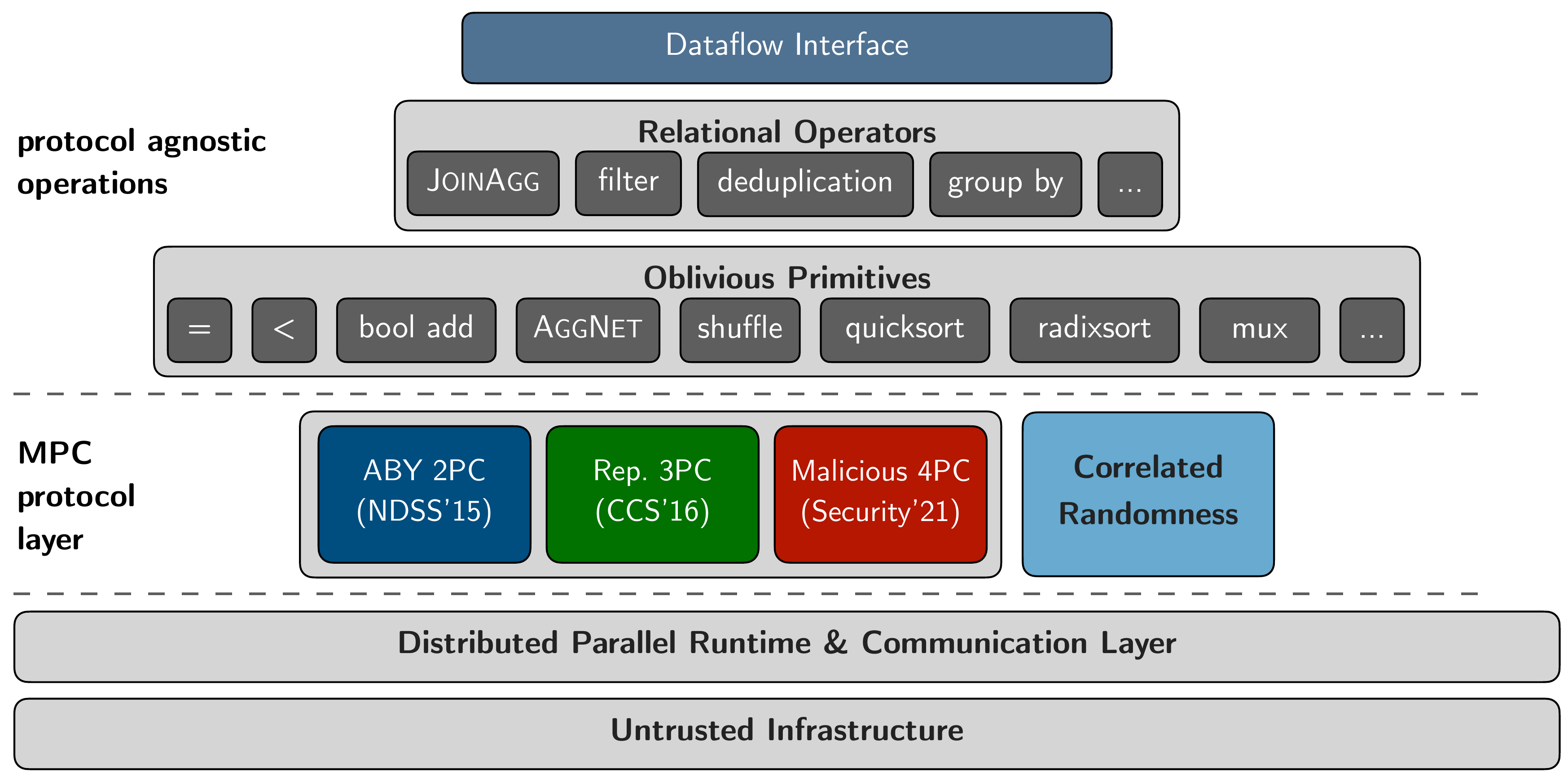}
    \caption{\tocs{The \ours stack. Analysts interact with the framework through its C++ dataflow interface. All oblivious primitives and relational operators inherit the security guarantees of the underlying MPC protocols.}}\label{fig:stack}
\end{figure}

\section{Experimental evaluation}\label{sec:eval}
Our experimental evaluation is structured into five parts:

\paragraph{\ours's performance on complex analytics.} In \S\ref{sec:eval-queries}, we present a comprehensive performance evaluation of \ours in LAN and WAN with three MPC protocols. Our results demonstrate that \ours provides excellent performance and can compute complex multi-way join queries on millions of input rows within minutes.

\paragraph{Comparison with state-of-the-art.} In \S\ref{sec:eval-comparison},  we compare \ours with the following state-of-the-art systems:
\begin{enumerate}[label=(\roman*)]
    \item \camera{Secrecy~\cite{secrecy},  a relational MPC framework targeting the 3-party outsourced setting with semi-honest security and no leakage.  Secrecy is the only open-source relational system for this setting.}
    \item \camera{SecretFlow~\cite{fang2024secretflow}, a relational MPC framework targeting the 2-party peer-to-peer setting with semi-honest security.  SecretFlow is the only system we are aware of that scales some TPC-H queries to millions of rows per~table,  albeit by leaking information to  parties.}
    \item \camera{MP-SPDZ~\cite{keller2020mp}, a general-purpose MPC compiler used in many prior works.  While MP-SPDZ does not support relational analytics, we compare with its sorting implementation, as sorting is the most expensive operator in relational queries.}
\end{enumerate}
\ours provides up to $827\times$ lower query latency than Secrecy and its secure sorting implementations are up to $5.9\times$ and $189.1\times$ faster than the ones provided by SecretFlow and MP-SPDZ, respectively. \camera{We also achieve modest speedups over SecretFlow on queries with joins (up to $1.5\times$),  despite the fact that SecretFlow introduces leakage while \ours does not.}

\paragraph{Scalability evaluation.} In \S\ref{sec:eval-scalability}, we show that \ours's operators scale to hundreds of millions of rows under all protocols.  When configured with our most expensive protocol for malicious security,  \ours sorts $2^{29}$ rows in less than $2$ hours. \tocs{We also report results on \ours's performance in a geodistributed setting in \S \ref{apdx:geo}, where we show that \ours can be practically deployed over the Internet.

\paragraph{Ablation study.} In \S\ref{sec:eval-ablation}, we consider each of our major optimizations separately and measure how they affect end-to-end runtime on a representative query.

\paragraph{Preprocessing.} Finally, we benchmark protocols for generating boolean and arithmetic Beaver triples in \S\ref{sec:eval-preproc}.
}

\subsection{Evaluation setup}\label{sec:eval-setup}
We use \texttt{c7a.16xlarge} AWS instances with Ubuntu \texttt{22.04.5 LTS} and \texttt{gcc~11.4.0}, in two configurations: (i) \texttt{LAN} is an unconstrained setting in the \texttt{us-east-2} (Ohio) region,  with a maximum bandwidth of $25$Gbps and a $0.3$ms round-trip time (RTT),  (ii) \texttt{WAN} is a restricted setting in \texttt{us-east-2}, with symmetric RTT of $20$ms and $6$Gbps bandwidth \camera{(as measured between \texttt{us-west-2} and \texttt{us-east-1} with 16 parallel \texttt{iperf3} streams)}.  Unless otherwise specified,  \ours is configured with $16$ compute threads and $4$ parallel connections in \texttt{LAN} and $16$ connections in \texttt{WAN}.

\paragraph{Workloads.} We use 31 queries for evaluation. Our most comprehensive benchmark is the full set of TPC-H queries~\cite{tpc-h}, implemented with the \ours dataflow API.
We replace floating point values with equivalent integral values and pattern matching expressions (\verb|X like `Y%'|) in six queries with oblivious (in)equality. To perform fully private averages, we implement a \textit{non-restoring division} circuit inspired by the hardware literature \cite{lu2005arithmetic}. Various prior works have implemented selected TPC-H queries~\cite{secrecy,Poddar2021Senate,bater2018shrinkwrap,fang2024secretflow}, but to our knowledge, we are the first to report performance numbers for \textit{all} TPC-H queries (22 in total) entirely under MPC.  
\camera{We also implement all other queries we could find in prior relational MPC works~\cite{Volgushev2019Conclave,secrecy,Poddar2021Senate,bater2018shrinkwrap,fang2024secretflow, Wang2021Secure,Bater2017SMCQL,luo2024secure, bater2020saqe}:}
\begin{itemize}
\item \camera{\textit{Aspirin}, \textit{C. Diff}, \textit{Password}, \textit{Credit Score}, \textit{Comorbidity}, and \textit{SecQ2}, used in Secrecy~\cite{secrecy}, Conclave~\cite{Volgushev2019Conclave}, and Senate~\cite{Poddar2021Senate}, among others.}
\item \camera{\textit{Market Share} from Conclave.}
\item \camera{\textit{SYan}, Example 1.1 from Wang et al.~\cite{Wang2021Secure}.}
\item \camera{\textit{Patients}, a 3-way join query used in Shrinkwrap \cite{bater2018shrinkwrap} to showcase the cascading effect (which \ours avoids).}
\end{itemize}

\noindent Additionally, we implement the five peer-to-peer variations of TPC-H queries in SecretFlow~\cite{fang2024secretflow}, denoted as \textit{S1}--\textit{S5}.  As an extra correctness check, we also implemented all queries in SQLite~\cite{sqlite} \tocs{and compared its output with \ours's}.

\paragraph{Inputs. } In all experiments, we use $\ell=64$ bit shares by default. However, because our sorting protocols pad inputs to guarantee uniqueness and compute permutations, sorting is actually performed over $128$-bit shares. We standardize a notion of query size: in the TPC-H specification \cite{tpc-h}, input tables grow according to the \textit{scale factor} (SF) parameter, where SF1 corresponds to an (approximately) 1GB database. The smallest table (\texttt{Supplier}) has $10\mathrm{k}\cdot\mathrm{SF}$ rows, while the largest (\texttt{LineItem}) has $6\mathrm{M}\cdot \mathrm{SF}$ rows. TPC-H queries have total input sizes between $810\mathrm{k}\cdot\mathrm{SF}$ (Q11) to $8.5\mathrm{M}\cdot\mathrm{SF}$ (Q9) rows, and average $5.8\mathrm M\cdot\mathrm{SF}$ rows.  For non-TPC-H queries, we accordingly define input sizes with $\approx 5\mathrm{M}$ rows per SF.

\paragraph{Protocols.}  We label results according to the MPC protocol: \textbf{SH-DM} (semi-honest, dishonest majority) for ABY \cite{aby}, \textbf{SH-HM} (semi-honest, honest majority) for Araki et al. \cite{ArakiFLNO16}, and \textbf{Mal-HM} (malicious-secure, honest majority) for Fantastic Four \cite{fantastic4}.  In all \emph{SH-DM} experiments, we report online time, as in prior works \cite{fang2024secretflow}.
\tocsCam{
    A recent attack \cite{bs26} was published on the security guarantees of the Fantastic Four protocol. The authors notified us of their attack while this paper was under submission; we include notes below on how the fixed, secure implementation affects our results. Notably, Br\"uggemann et al.~\cite{bs26} attack the protocol itself, as opposed to any particular implementation, so MP-SPDZ \cite{keller2020mp}, which we compare to, was likewise vulnerable. For an in-depth discussion of how we patched our implementation, see our technical report~\cite{cryptoeprint:2026/1152}.
}

\subsection{Performance on complex analytics}\label{sec:eval-queries}

In this experiment, we run \ours at Scale Factor 1 (SF1) on the full workload,  which includes queries of varying complexity.  
Queries like Q5,  Q7,  Q8,  Q9 and Q21 are expensive, involving 4-7 joins each, with Q21 also performing a self-join on the largest input table, \verb|LineItem|.  Other queries include multiple filters, semi-joins, outer joins,  group-by aggregations,  distinct, and order-by operations, covering a wide range of query patterns.  Q6 is the least expensive query, as it does not require sorting. \tocsCam{We observe than the fixed version of Mal-HM \cite{cryptoeprint:2026/1152} has an overhead of $\approx2\times$ on some queries at Scale Factor 1, primarily due to the additional malicious checks in sorting, but is up to $\approx3\times$ faster in other cases, due to the better resource utilization of batched malicious checks. Our results here correspond to the \texttt{custom} protocol in that paper.}

\begin{figure}[t]
\centering
\smaller
\begin{tabular}{r|rr|rr}
& \multicolumn{2}{c}{TPC-H} & \multicolumn{2}{c}{Other} \\
       & Median & Max & Median & Max \\ \hline\hline
SH-DM  LAN &  3.8 &  15.0 &  1.6 &  3.1 \\
       WAN & 13.5 &  52.7 &  6.2 & 11.7 \\ \hline
SH-HM  LAN &  4.4 &  17.4 &  2.0 &  3.7 \\
       WAN & 10.9 &  41.4 &  4.8 &  9.1 \\ \hline
Mal-HM LAN & 10.9 &  42.3 &  4.9 &  8.3 \\
       WAN & 27.1 & 108.4 & 11.9 & 21.7 \\
\end{tabular}
\vspace{1mm}

\includegraphics[width=0.7\columnwidth]{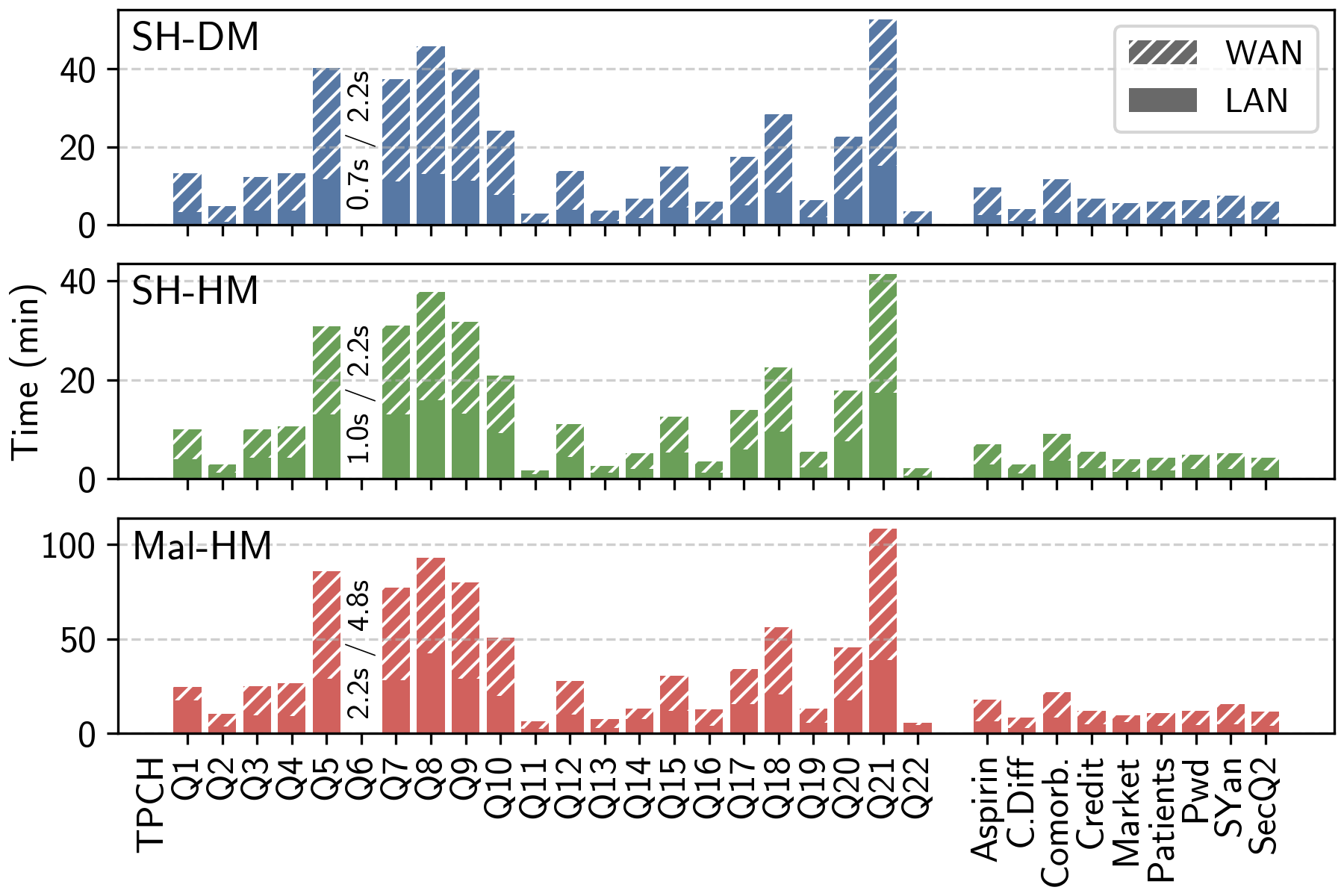}

\caption{Query execution time (min) at SF1 in \texttt{LAN} (solid) and \texttt{WAN} (hatched).  Bars are overlapped, not stacked.  TPC-H queries shown on the left and other queries on the right. Q6 times are annotated.}\label{fig:sf1-comparison} 
\end{figure}

Figure~\ref{fig:sf1-comparison} shows the end-to-end times in \texttt{LAN} and \texttt{WAN} and summarizes the result statistics.  Overall, \ours computes complex multi-way join queries on millions of input rows in a few minutes.  The most expensive query, Q21, calls the sorting operator 12 times and, under malicious security, completes in $42$ minutes over \texttt{LAN}.  In the same setting, \ours executes all other queries from prior work in under $10$ minutes.  Our results also demonstrate the effectiveness of \ours's vectorization and message batching.  In the \texttt{WAN} environment,  we see $1.2\times$-$6.9\times$ higher execution times for a $75\times$ higher~RTT.

\subsection{Comparison with state-of-the-art systems}

In this section, we compare \ours with three state-of-the-art MPC frameworks on  relational queries and oblivious sort.  

\begin{figure}[t]
    \centering
    \begin{minipage}{0.45\columnwidth}
        \centering
        \includegraphics[width=\linewidth]{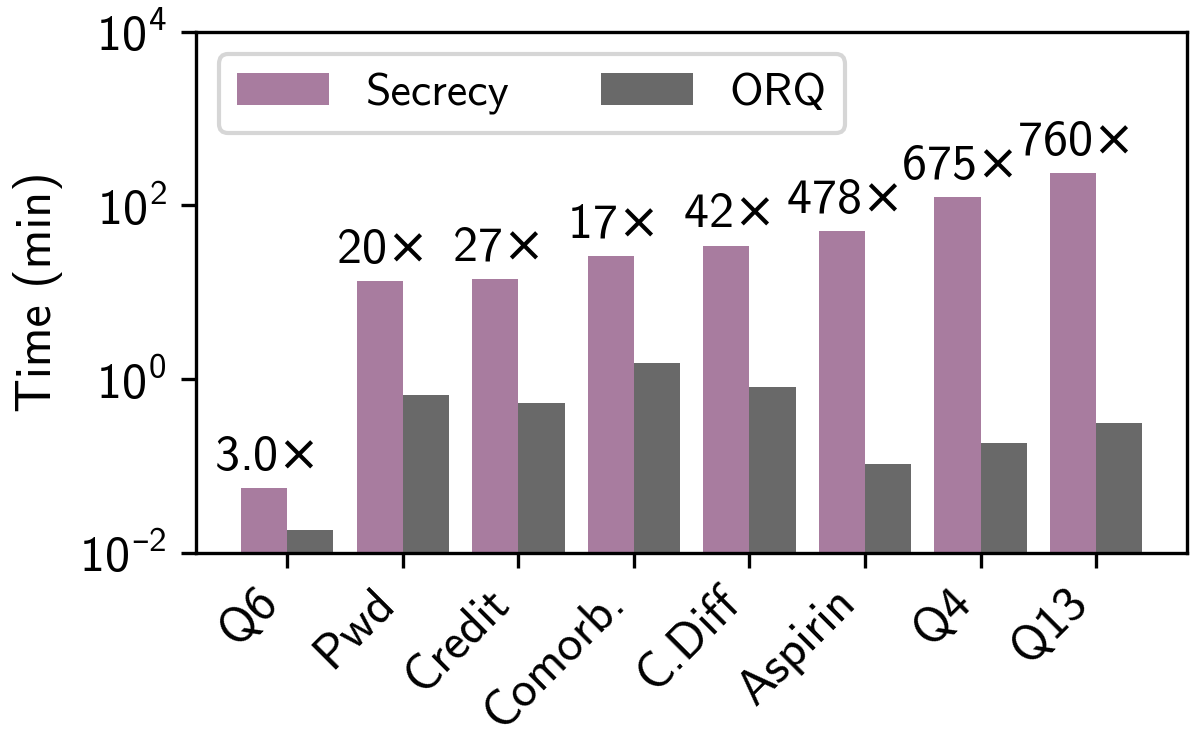}
    \end{minipage}
    \begin{minipage}{0.35\columnwidth}
        \centering
        \includegraphics[width=\linewidth]{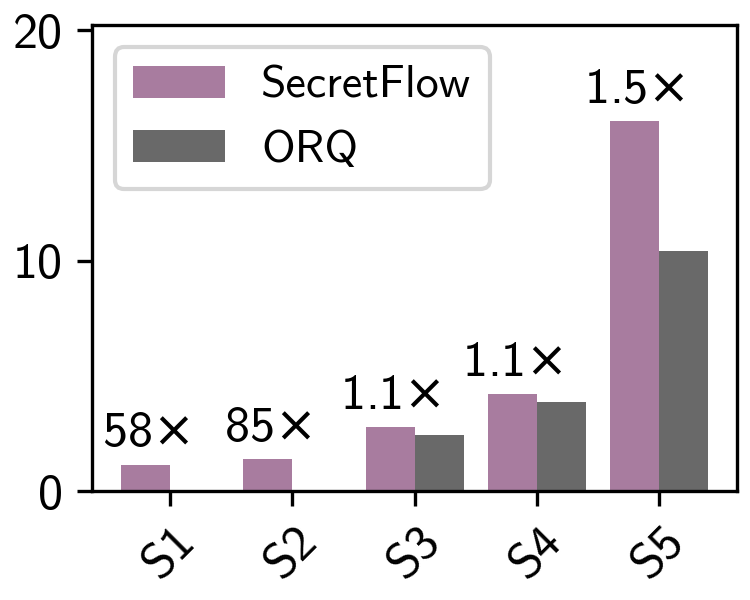}
    \end{minipage}
    \caption{\ours query times compared to Secrecy and SecretFlow. Secrecy operates in the outsourced setting without leakage. SecretFlow is a peer-to-peer system that offloads operations to the data owners' trusted compute and leaks the result of the join to parties.
    }\label{fig:secrecy-comparison}
\end{figure}

\paragraph{Comparison on queries.}\label{sec:compare-queries}
We first compare \ours with Secrecy~\cite{secrecy}, the only open-source relational MPC system without leakage that operates in the outsourced setting. We configure \ours with the SH-HM protocol, which is the one used by Secrecy, and we execute all queries from the Secrecy paper, using the maximum input size they report~\cite[Fig.~4]{secrecy}. 

Figure~\ref{fig:secrecy-comparison} (left) shows the results. 
For the most expensive queries (Aspirin,  Q4,  and Q13) that perform a join or semi-join, \ours achieves $478\times$ --- $760\times$ lower latency.
Password, Credit, Comorbidity, and C.Diff include group-by and deduplication operators.  For these queries, \ours's optimized sorting protocols also improve latency by $17\times$ --- $42\times$.
Q6 does not have any join or sorting, yet \ours outperforms Secrecy by~$3\times$. 
\camera{The significant improvements come from better asymptotic complexity of our join and sorting algorithms.  Secrecy uses an $O(n^2)$ join and an  $O(n\log^2n)$ bitonic sort,  while \ours evaluates join-aggregation queries in $O(n\log n)$.  Also,  Secrecy's runtime is single-threaded.}

Next, we compare with SecretFlow~\cite{fang2024secretflow}, a recent relational MPC framework that operates in the peer-to-peer setting. To our knowledge, SecretFlow is the only open-source system of this type that scales TPC-H-based queries to millions of input rows per table. 
However, we emphasize that SecretFlow \emph{leaks which rows match} to the parties, while \ours does not introduce any leakage. Nevertheless, we include this comparison to show that \ours can achieve competitive performance even in a peer-to-peer setting.
We configure \ours with the SH-DM protocol that uses the same underlying MPC primitives as SecretFlow (ABY~\cite{aby}), and we allow both systems to perform operations in the data owners' trusted compute, otherwise SecretFlow's optimizations are not applicable.

Figure~\ref{fig:secrecy-comparison} (right) shows results for the five queries used in the SecretFlow paper,  run on 16M input rows per table.  \ours delivers superior performance in all cases,  while protecting all data and intermediate result sizes.  It achieves $58-85\times$ lower latency for simple queries without joins (S1, S2) and $1.1-1.5\times$ lower latency for join queries that include aggregations (S3, S4) or oblivious group-by (S5). \tocsCam{S1 and S2 see such significant speedups because --- under SecretFlow's peer-to-peer model --- those queries largely reduce to local SQL operations rather than communication-intensive MPC evaluations.} \camera{Despite SecretFlow's leakage,  \ours still wins on performance due to its more efficient sorting and fused join-aggregation operator. Both systems use radixsort,  but SecretFlow cannot leverage parallelism.  In contrast, \ours's operators are vectorized and data-parallel.}

\begin{figure}[t]
    \centering
    \includegraphics[width=0.75\columnwidth]{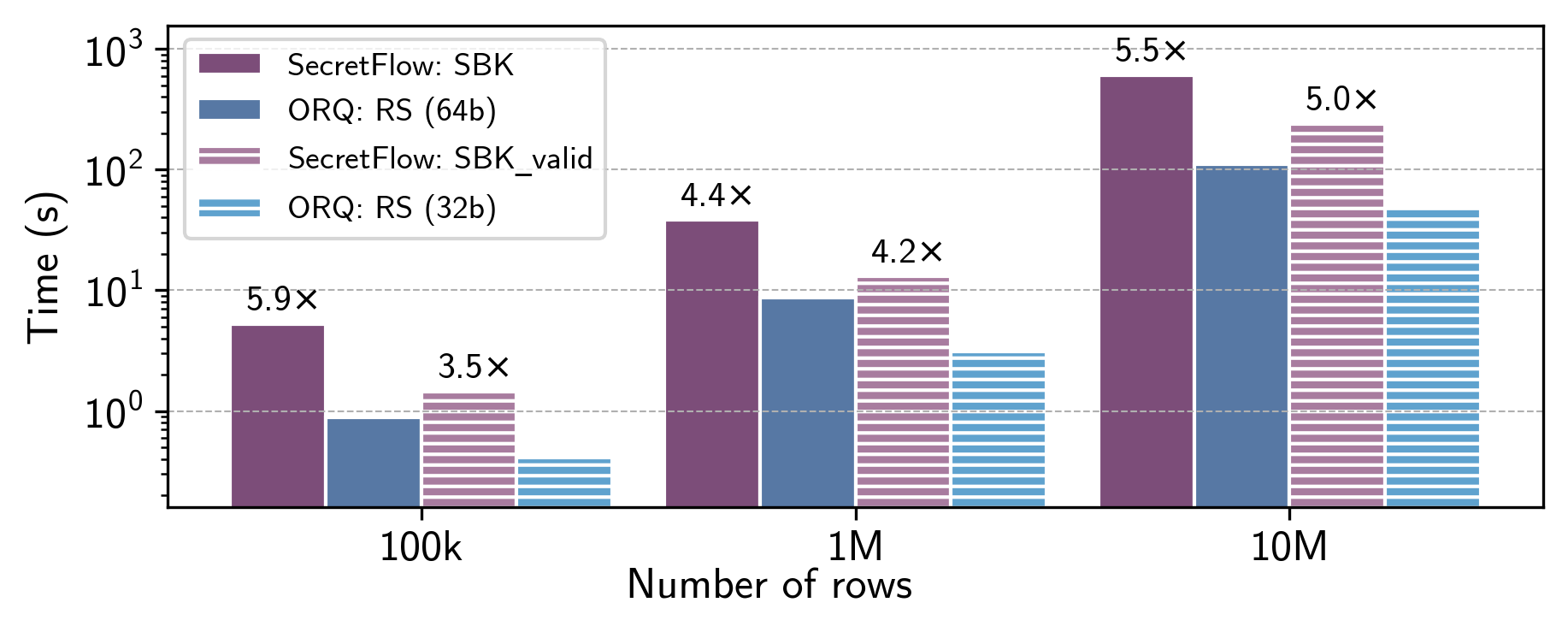}
    \caption{Performance of oblivious sort in SecretFlow and \ours.}\label{fig:secretflow-comparison}
\end{figure}

\paragraph{Comparison on oblivious sort.}\label{sec:compare-sort}
We now compare \ours's oblivious radixsort against two publicly available state-of-the-art implementations. In Figure~\ref{fig:secretflow-comparison}, we compare our SH-DM radixsort with SecretFlow's \texttt{SBK} ($64$-bit radixsort) and \texttt{SBK\_valid} protocols.  \texttt{SBK\_valid} is a modified radixsort that relies on data owners locally mapping plaintext data to use at most $\lceil\lg n\rceil$ bits, where $n$ is the input size ($n = 17, 20, 24$ bits in this experiment). Even though it favors the baseline, we compare \texttt{SBK\_valid}  to our $32$-bit radixsort and \texttt{SBK} to our $64$-bit protocol. 
We see that \ours is up to $4.4\times$ and $5.5\times$ faster when sorting $1\mathrm M$ and $10\mathrm M$ elements, respectively.

\begin{figure}[t]
    \centering
    \includegraphics[width=0.8\columnwidth]{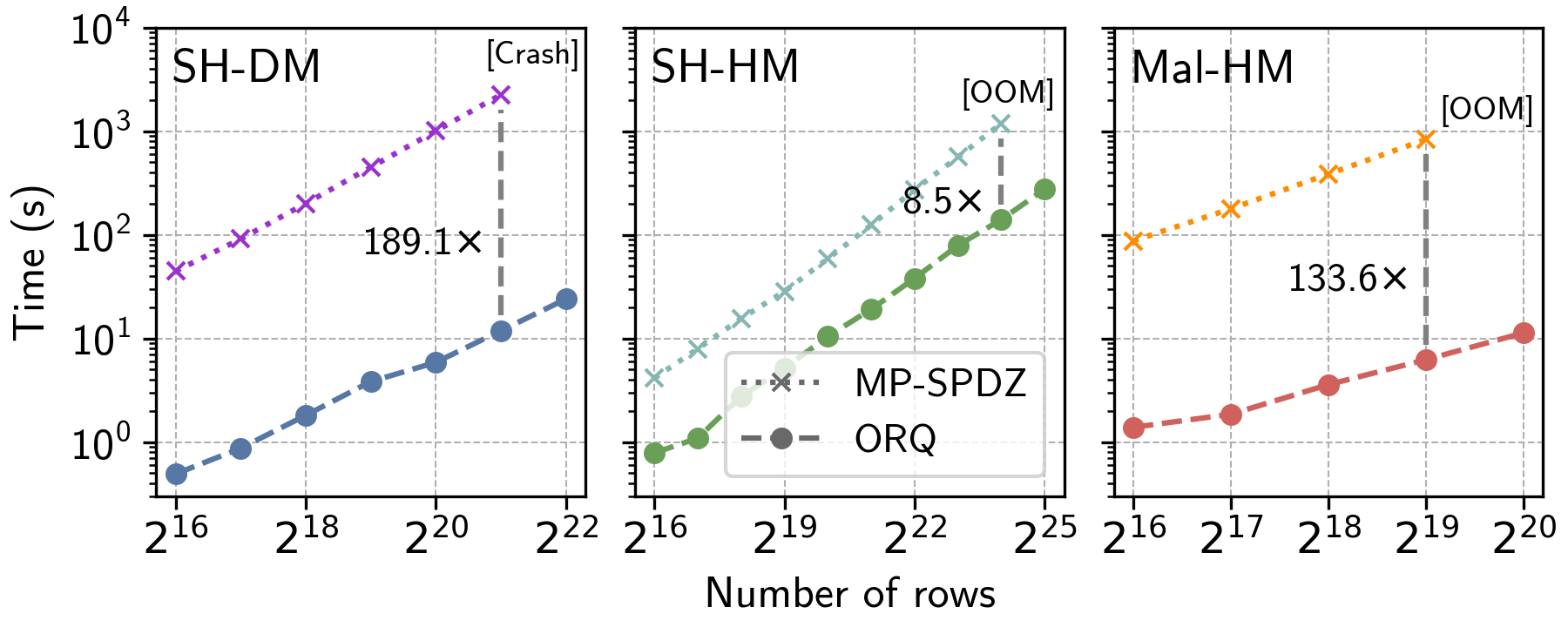}
    \caption{Performance of oblivious radixsort protocols in MP-SPDZ and \ours . We run MP-SPDZ until it crashes or runs out of memory.}\label{fig:mpspdz-comparison}
\end{figure}

Our next comparison is with MP-SPDZ \cite{keller2020mp}, a state-of-the-art general-purpose MPC library that provides secure sorting implementations for all MPC protocols supported by \ours. We compare with its radixsort implementation, which is its most efficient sorting algorithm.
Figure~\ref{fig:mpspdz-comparison} shows the results.
We vary the input size by powers of two, starting from $2^{16}$ elements and increasing the input size until MP-SPDZ either runs out of memory or crashes: up to $2^{22}$ elements with SH-DM,  $2^{25}$ with SH-DM, and $2^{20}$ with Mal-HM. In the SH-DM and Mal-HM settings,  \ours's radixsort is $189\times$ and $134\times$ faster than the corresponding MP-SPDZ implementations.  Under the SH-HM protocol, which is heavily optimized in MP-SPDZ,  \ours achieves a $8.5 \times$ speedup when sorting $2^{24}$ elements, however, MP-SPDZ could not complete the experiment for larger inputs.
\camera{The large performance improvements over MP-SPDZ come from data-parallelism; although MP-SPDZ supports parallelism and advanced vectorization,  it does not parallelize sorting.}
\tocsCam{Due to the attack by Br\"uggemann et al.~\cite{bs26}, our Mal-HM results are expected to be slightly slower. However, we note that the implementation of Fantastic Four in MP-SPDZ was \textit{also} vulnerable to the same attack, so is expected to pay a similar, if not greater, cost. All experiments were run before the attack was made public. With our initial fixed implementation \cite{cryptoeprint:2026/1152}, we observe a sorting overhead of about $3$-$6\times$ in this regime, which quickly amortizes with larger input sizes.}

\paragraph{Bandwidth consumption.} \camera{We also evaluate bandwidth consumption in \ours and the baselines.  When compared to Secrecy,  \ours exhibits up to two orders of magnitude lower bandwidth consumption, due to our asymptotic improvements in sorting and join. Likewise,  our sorting implementations achieve $1.6\times$ lower bandwidth than MP-SPDZ SH-HM and more than an order of magnitude improvement against MP-SPDZ SH-DM and Mal-HM.  \ours
also achieves a $1.8\times$ lower bandwidth than SecretFlow in sorting.}
\tocsCam{
On queries, SecretFlow has a lower bandwidth than \ours, as it leaks matching rows in its join operator and performs subsequent operations locally. We report detailed bandwidth measurements for all experiments in Appendix \ref{apdx:bw}.
}
\label{sec:eval-comparison}
\subsection{\ours scalability}
In this section, we show that \ours can deliver practical performance for complex analytics on large inputs, without sacrificing security.  To the best of our knowledge,  we are the first to report results at this scale, for end-to-end MPC execution without leakage or the use of trusted parties.

\paragraph{TPC-H at SF10.} 
We repeat the experiment of \S\ref{sec:eval-queries} on the TPC-H benchmark at \camera{Scale Factor 10.  In this setting}, queries operate on inputs of $58\mathrm M$ rows, on average. The most expensive query, Q21, has a total input of size $75\mathrm M$ rows and performs two of the largest joins,  sorting tables with $120\mathrm M$ rows each time.

\begin{figure}[t]
    \centering
    \includegraphics[width=0.8\columnwidth]{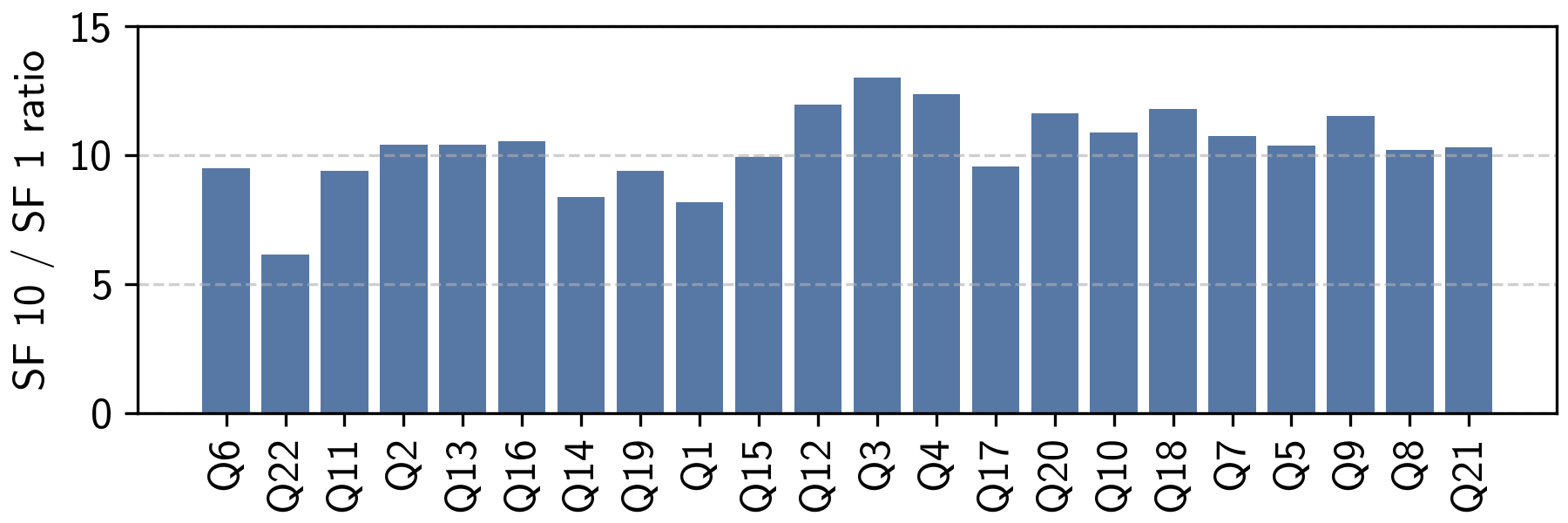}
    \caption{Ratio of TPC-H query execution times at SF10 over SF1 for the SH-DM protocol in \texttt{LAN}. The 22 queries are sorted from left to right by increasing SF10 latency.}\label{fig:lan-scaling}
\end{figure}

Figure~\ref{fig:lan-scaling} plots the ratio of execution times at SF10 over SF1 for the SH-DM protocol in \texttt{LAN}.
Assuming that query latency is dominated by sorting and aggregation, we would expect the runtime to scale as $O(n\log n)$. If the queries consisted \textit{only} of ideal sorts, then theoretical scaling would be $10\log(10n)/\log n\approx 11.5$, when $n\approx 5.8\mathrm M$ (SF1).
Overall, this is indeed the trend we see in Figure~\ref{fig:lan-scaling}, with some outliers.  In practice,  queries are more complex and consist of operations with different costs. Some overheads are sublinear: for example, oblivious division is round-constrained for most inputs, so scaling appears lower for queries with division, like Q22. Other queries suffer slightly suboptimal scaling, due to the power-of-two padding required in \protoPrefixAgg{}.  
For example, Q12 aggregates over a table of size $7.5$M at SF1, which is padded to $2^{23}\approx8\text{M}$ rows. However, at SF10, this same aggregation occurs on a table of size \textit{75M}, which now must be padded to $2^{27}\approx134\text{M}$ rows, $16\times$ larger than the table at SF1.

\begin{figure}[t]
    \centering
    \includegraphics[width=0.8\columnwidth]{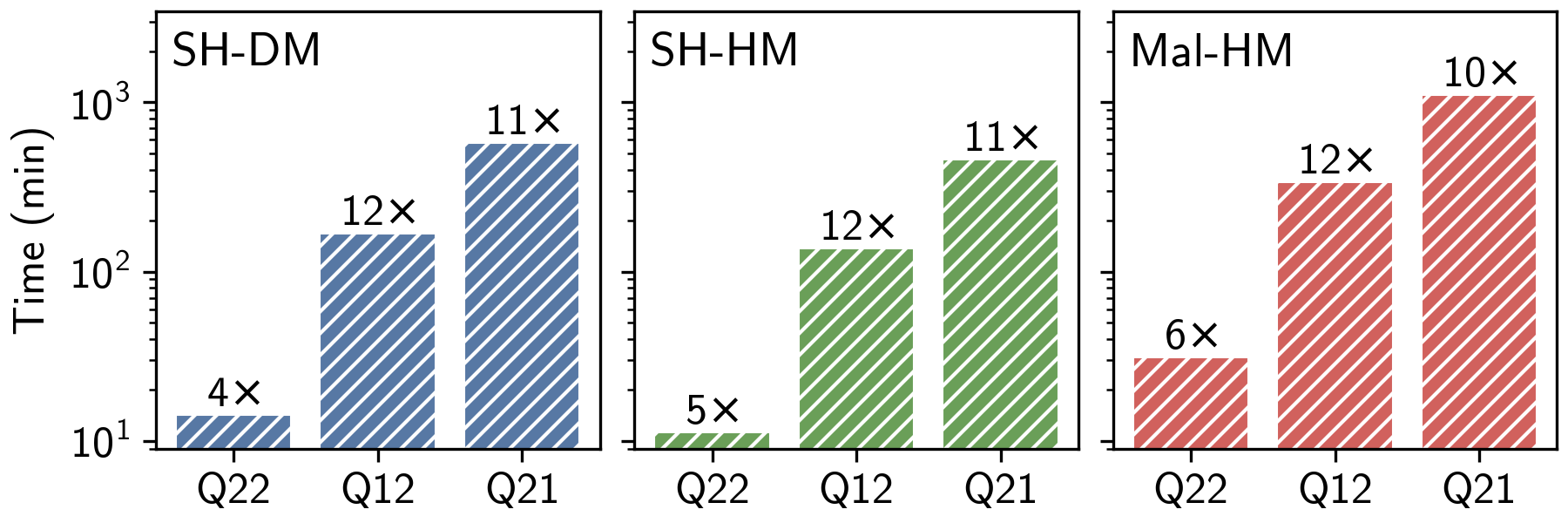}
    \caption{\ours performance at SF10 in \texttt{WAN}. Bars are labeled with the scaling ratio compared to SF1 in the same environment.  Q21 is the most expensive query in the TPC-H benchmark across protocols.}\label{fig:query-scaling}
\end{figure}

Next, we evaluate \ours's scalability in \texttt{WAN},  using the two outlier queries mentioned above, Q22 and Q12,  and the most expensive query, Q21. Figure~\ref{fig:query-scaling} shows the results for all protocols, where each bar is annotated with the scaling ratio compared to SF1.  \ours's scaling behavior in this environment is consistent with our previous findings.  Under malicious security,  Q22 completes in $31$ minutes and Q21 in $18$ hours.

\paragraph{Scalability of sorting protocols.}
In the previous experiment, we showed that \ours provides practical performance on complex queries that repeatedly sort tables with over $100$M rows. We now stress \ours's oblivious sorting protocols further and evaluate their performance on even larger inputs. We run radixsort and quicksort in \texttt{LAN} with 32 threads, increasing the input size from $2^{20}$ to $2^{29}$ elements.  Figure~\ref{fig:sorting} shows the results.  
The slowest protocol, Mal-HM radixsort, can sort $134$ million ($n=2^{27}$) elements in about $35$ minutes, while the fastest protocol, SH-HM quicksort, sorts $537$ million ($n=2^{29}$) in just over $70$ minutes. \tocsCam{For Mal-HM, the results on large inputs do not noticeably change after fixing the attack by Br\"uggemann et al.~\cite{bs26}, since the fully-secure checks only incur a constant-communication overhead $O(\log n)$ times. Concretely, we observe an overhead on Quicksort of less than $1\%$ at $n=2^{28}$ rows \cite{cryptoeprint:2026/1152}.}

\begin{figure}[t]
    \centering
    \includegraphics[width=0.8\columnwidth]{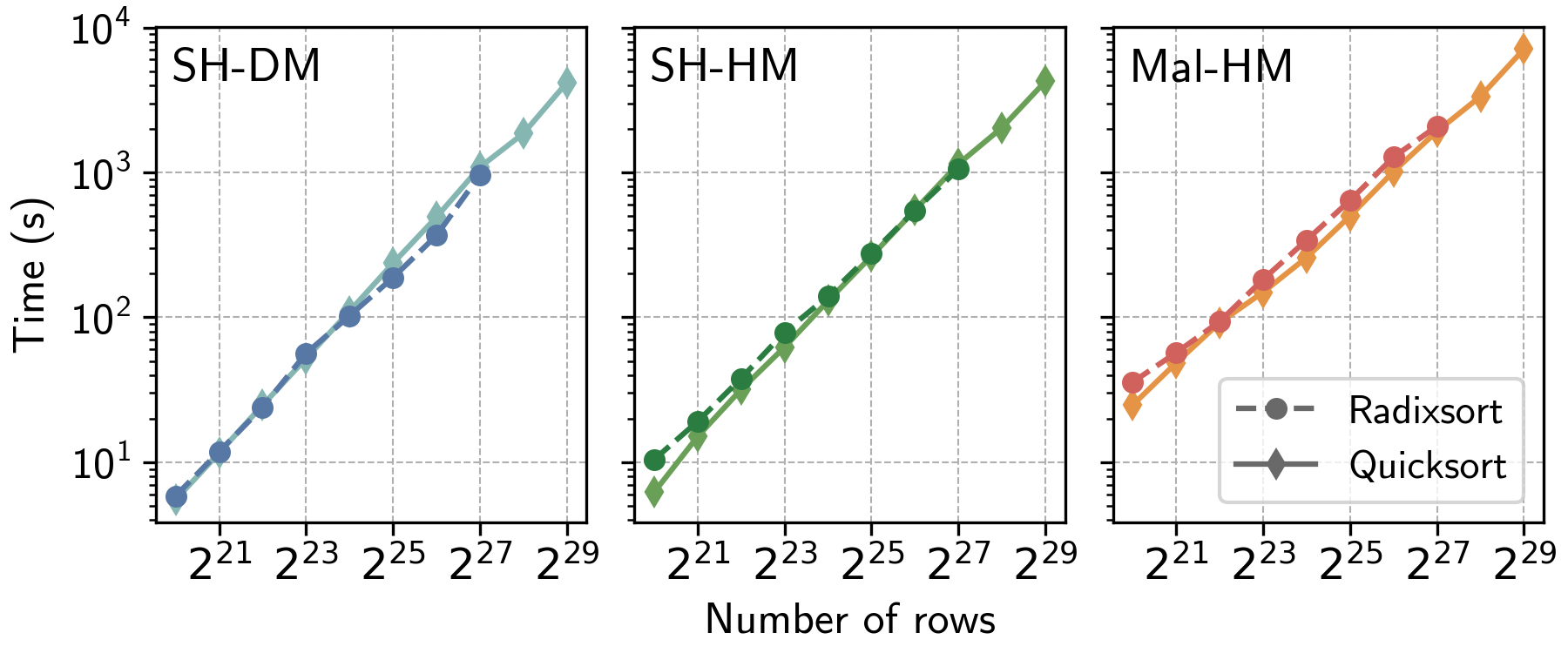}
    \caption{Scaling \ours's oblivious sorting protocols to very large inputs.  Quicksort scales to larger inputs than radixsort, which has higher time and memory requirements for permutation generation.}
    \label{fig:sorting}
\end{figure}

We find that quicksort and radixsort are competitive across settings.  Quicksort is consistently faster under the malicious protocol, while radixsort is slightly faster in the SH-DM setting.
Overall, quicksort (\ours's default sorting protocol) scales to larger inputs in our experimental setting, as radixsort has higher memory demands: it requires computing and storing $\ell+2$ secret-shared permutations of length $n$.

\label{sec:eval-scalability}
\tocs{
\subsection{Performance in a geodistributed setting}\label{apdx:geo}

In this section, we evaluate \ours's ability to amortize network costs in a geodistibuted setting.  We create a deployment spanning four AWS regions: \texttt{us-east-1} (N. Virginia),  \texttt{us-east-2} (Ohio), \texttt{us-west-1} (N. California),  and \texttt{us-west-2} (Oregon). The deployment is constrained by minimum link bandwidth between $4.23-8.47$ Gbps, and RTT between $50-61$ ms.  To make the setting challenging, we deploy parties so that at least one link exists between the U.S. east and west coasts.

We show results for five TPC-H queries: two of the fastest (Q11, Q21), the median (Q12), and the two slowest (Q8, Q21).  These five queries account for approximately 30\% of the total time to run the entire TPC-H benchmark. 
Figure~\ref{fig:sf1-real-wan-comparison} shows the execution time, along with the overhead compared to the symmetric \texttt{WAN} environment used in the experiments of Fig.~\ref{fig:sf1-comparison}.  Latency increases by up to $2.4\times$ for SH-DM, $1.7\times$ for SH-HM, and $1.9\times$ for Mal-HM. These results demonstrate that \ours's overhead is lower than the corresponding $3\times$ increase in RTT between the two WAN environments,  indicating its ability to effectively amortize network costs and achieve competitive performance. \tocsCam{The current implementation of the fixed malicious checks \cite{bs26, cryptoeprint:2026/1152} are not round-optimized, so may increase WAN times slightly. With further engineering effort, however, we can lower the round complexity; our analysis shows that only four rounds are necessary to run each malicious-secure check.}

\begin{figure}
    \centering
    \includegraphics[width=0.8\columnwidth]{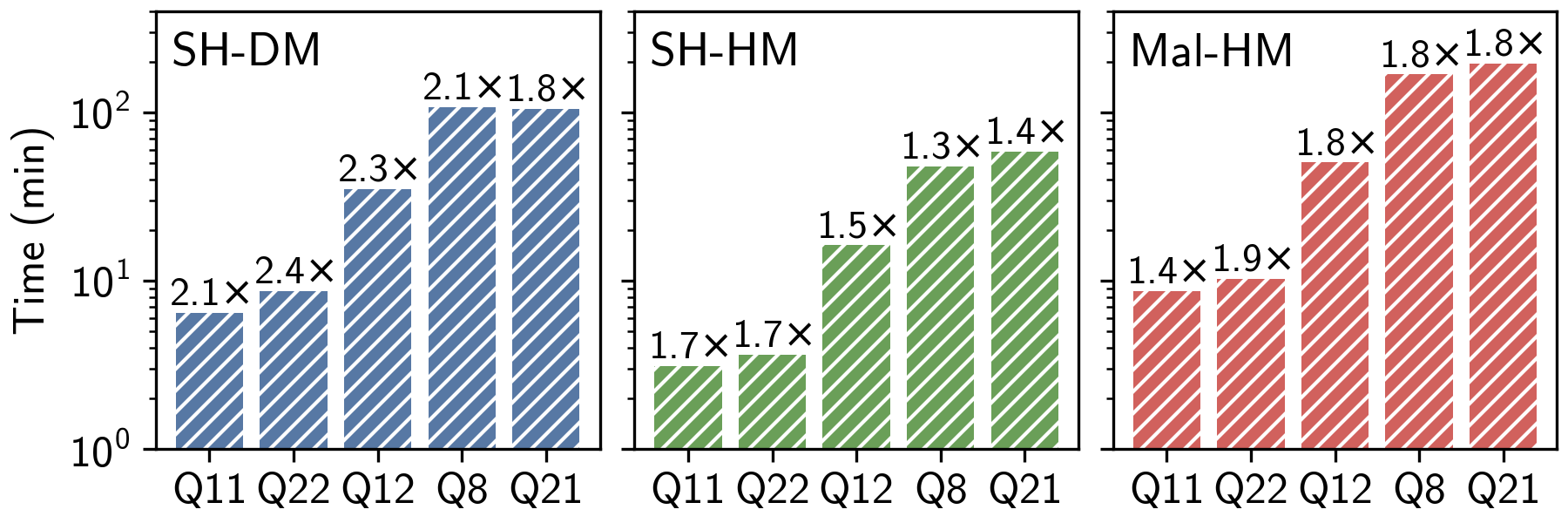}
    \caption{TPC-H query times at SF1 in a geo-distributed WAN deployment, with ratios over times in our symmetric WAN.}\label{fig:sf1-real-wan-comparison}
\end{figure}
}
\tocs{
\subsection{Ablation study}
\label{sec:eval-ablation}

We perform ablation studies in both LAN and WAN to measure the effectiveness of various optimizations, one at a time.
Our target workload is TPC-H Q15, which had a median runtime and includes all of our major operators. We focus on the SH-HM protocol and benchmark with 600k rows (SF0.1) on a cluster of 96-core AMD EPYC 9655P CPUs with 256 GB of memory each.
We choose SF0.1 because the baseline does not scale beyond that. At this size, the quadratic join must materialize a table with $600\mathrm M$ rows; at SF1, this jumps to $60\mathrm B$ rows, well beyond the memory capacity of our machines.

As part of the ablation study, we implement the pairwise sorting network \cite{pairwise}, which may be of independent interest. Like bitonic sort, it has $O(n\log^2n)$ asymptotic complexity, but it uses $O(n\log n)$ fewer comparisons --- the minimal number for any practical sorting network (and the same as Batcher's odd-even mergesort \cite{Batcher68}).

The baseline configuration uses no multithreading, no message batching, and asymptotically slower algorithms for join and sort (both implemented in \ours): a Cartesian product with cost $O(n^2)$ and a pairwise sorting network \cite{pairwise} with cost $O(n\log^2n)$, as used in prior work \cite{secrecy,FengScape2022}. 
Tables~\ref{tab:opts-lan} and \ref{tab:opts-wan} summarize the optimizations  and the resulting speedups in LAN and WAN, while Figure~\ref{fig:ablation} plots the end-to-end execution times. 

We do not evaluate the baseline join and sorting algorithms in WAN, as the network overhead of the Cartesian product and sorting network dwarfs the system optimizations that we wish to highlight. We also omit multiple connections in the LAN test, since a single connection is able to take full advantage of available network resources in that setting. We evaluate the effect of message batching in both LAN and WAN: in LAN, we compare no message batching at all with twelve batches per thread (our default), whereas for WAN, we compare twelve versus one batch per thread, to demonstrate how the decreased round complexity improves performance in that setting (batches execute sequentially per thread, so fewer batches corresponds with lower round complexity). We do not test without batching in WAN, as the increased network latency makes the experiment impractical to complete.

In LAN, we see that each of our asymptotic improvements (sorting and join algorithm) give about a $10\times$ and $50\times$ speedup, respectively. In both WAN and LAN, system optimizations (threading, multiple connections, and batching) also give sizeable speedups. In WAN, the greatest speedup comes from adding multiple parallel connections, since in this setting MPC execution is bound by communication, not computation. While multithreading is a powerful tool to improve the end-to-end running times of \ours queries, it only becomes beneficial at much larger inputs than SF0.1.

Overall, we observe about $1800\times$ higher query throughput due to our optimizations. Since our improvements are \textit{asymptotic} --- from $O(n^2)$ to $O(n\log n)$ --- this performance gap will only grow with larger input sizes. We observe similar speedups in the other two protocols supported by \ours (SH-DM, Mal-HM).

\begin{table}
\centering
\begin{tabular}{c|c|r}
    Baseline & Optimized & Speedup\\ \hline
    Sorting Network & Quicksort & $11.4\times$\\ 
    Cartesian Product & \textsc{JoinAgg} & $58.2\times$ \\
    No Batching & Multiple Batches & $1481.6\times$ \\
    Single Threaded & Multithreaded & $1851.0\times$
\end{tabular} 
\caption{\tocs{Ablation study optimizations in LAN}}
\label{tab:opts-lan}
\end{table}

\begin{table}
\centering
\begin{tabular}{c|c|r}
    Baseline & Optimized & Speedup\\ \hline
    Single Thread & Multithreaded & $1.2\times$ \\ 
    Multiple Batches & Single Batch & $1.7\times$\\
    Single Connection & Multiple Connections & $6.8\times$
\end{tabular}
\caption{\tocs{Ablation study optimizations in WAN}}
\label{tab:opts-wan}
\end{table}

\begin{figure}
    \centering
    \includegraphics[width=0.7\columnwidth]{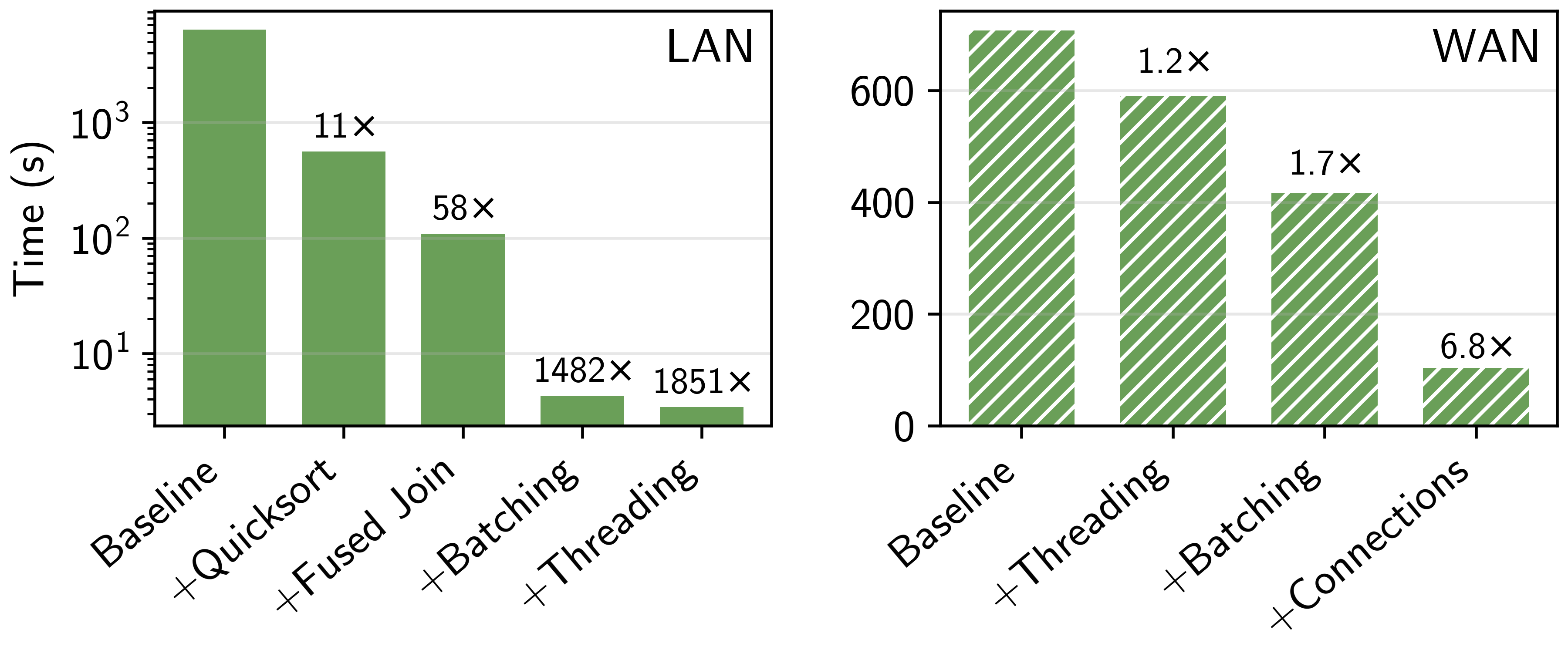}
    \caption{\tocs{Ablation study results in LAN and WAN. Above each bar we label the relative speedup compared to the baseline.}}\label{fig:ablation}
\end{figure}

}

\tocs{

\subsection{Preprocessing phase}\label{sec:eval-preproc}


\ours uses both boolean and arithmetic Beaver triples.
Our final experiment is an independent benchmark of the preprocessing phase for the SH-DM protocol, with a particular focus on the challenging case of generating arithmetic Beaver triples (over the integers $\bmod\ 2^\ell$). Such triples are needed for executing the secure multiplication protocol and the arithmetic multiplexer (\S\ref{sec:approach}). We are concerned with the \textit{runtime} and \textit{communication overhead} of generating Beaver triples. We require one Beaver triple per multiplication gate; a vector of length $n$ requires $n$ Beaver triples. For comparison purposes, we also include measurements for boolean (AND) triples. %
As described in \S \ref{sec:triples}, we construct Beaver triples via a local reduction to twice as many oblivious linear evaluations (OLE).
For all experiments of this section, we use the same hardware setup as in \S\ref{sec:eval-ablation}.

One might reasonably ask why we even bother with arithmetic triples, or, for that matter, arithmetic secret sharing in the first place. In this work, we optimize for a fast \textit{online} phase; our goal is to evaluate queries with as little latency as possible, but allow for (possibly expensive) preprocessing which may occur before input data --- or even the query itself --- has arrived. Without arithmetic triples, all multiplications and additions would have to be decomposed down to their expensive circuit representations (which require multiple rounds and have $O(\ell^2)$ and $O(\ell)$ communication complexity, respectively), leading to orders-of-magnitude slowdowns in the online phase.

We evaluate Gilboa's original OLE protocol \cite{gilboa}, which has communication cost $O(\ell^2)$ to generate an $\ell$-bit Beaver triple. We also implement the recent landmark protocol by Doerner et al. \cite{doerner2025} which uses a number-theoretic construction to break the two-decade-old quadratic barrier, and achieves $O(\ell \log \ell)$ communication per triple. To our knowledge, we are the first to publish an implementation for this protocol. Our results confirm that despite the heavier computational overhead of the subquadratic protocol, it is concretely more communication efficient at practical input sizes. In real-world deployments, such as wide-area networks, where bandwidth may be more expensive than compute, it is a viable alternative for arithmetic triple generation.

We base triple generation on the pseudorandom correlation generator (PCG) implementation (\texttt{SilentOTExt}) from \texttt{libOTe} \cite{libote}, in order to take advantage of sublinearities in the number of triples generated. AND triples directly follow from oblivious transfer, so those figures represent an effective benchmark of OT generation. However, we acknowledge that more engineering work is needed to make silent OT generators fully practical. Due to memory constraints, we are often unable to generate enough triples at once to fully amortize the computational overheads of the silent constructions.


In this experiment, we benchmark triple generation with up to 64 threads, and in batches of up to $2^{24}$ elements. We report \textit{normalized} throughput in Table~\ref{tab:triples}, which shows the results for the fastest batch size, and Figure~\ref{fig:triple-time}, which plots the fastest amortized execution time per bit. To enable a clearer comparison between protocols and bitwidths, we express throughput in terms of normalized megabits per second, and execution time in terms of normalized nanoseconds per generated bit. These values represent the best performance achievable with our implementation. The actual number of triples generated, in each case, is a factor of $\ell$ lower.

Gilboa's protocol achieves a maximum throughput of 165 Mbps at a width of 32 bits, corresponding to about 5M triples per second. The highest throughput observed for the CRT protocol was almost three times lower: 56 Mbps at 128 bits, or just over 400k triples per second. However, at this bitwidth, Gilboa only achieves 650k triples per second.

Figure~\ref{fig:triple-comm} shows how many bits of communication are required per output bit, for both 64-bit triples (our default bitwidth) and 128-bit triples (our sorting bitwidth; see~\S\ref{sec:sort}). In this experiment, we see the real benefit of the CRT protocol: it uses less than half the bandwidth of Gilboa's protocol at 64 bits, and $5\times$ lower bandwidth at 128 bits.
The amortized overhead of each Boolean triple in our implementation is around two communicated bits per output bit. We omit these values from Figure~\ref{fig:triple-comm} since they would not be clearly visible, but remark that this implies approximately $26\times$ greater communication overhead, per bit, for the CRT protocol compared to Boolean triples, and $128\times$ greater overhead for Gilboa, in the case of 128-bit outputs. With further work on batched silent preprocessing, we could, in principle, bring the communication cost of a Boolean triple close to zero.


The CRT protocol trades off substantially lower communication for greater computation. While we do not see lower end-to-end latencies, in practice, due to greater computational overheads, our results show that, with a more efficient use of computation resources, we could achieve a concrete improvement over Gilboa's protocol. This is an exciting avenue for future research, but will require new approaches to meeting the storage demands of correlated randomness.

\begin{table}
\tocs{
\centering
\begin{tabular}{cc|l|c}
Triple Type & Bitwidth $\ell$ & Max. Tput & Opt. Batch Size \\ \hline\hline
Boolean
&   8 &  87.4 Mbps & 524k \\
\texttt{SilentOT}~\cite{libote} &  16 & 130.2      & 262k \\
&  32 & 123.8      & 131k \\
&  64 & 119.7      & 65k \\
& 128 & 123.4      & 32k \\ \hline
Arithmetic
& 8 & 144.4 Mbps & 4M \\
Gilboa \cite{gilboa}
& 16 & 163.6 & 8M \\ 
& 32 & 165.1 & 4M \\
& 64 & 154.0 & 8M \\ 
& 128 & 82.7 & 16M \\ \hline
Arithmetic
& 8 & 11.0 Mbps & 4M \\
Doerner et al. (CRT)~\cite{doerner2025}
& 16 & 25.1 & 8M \\
& 32 & 37.1 & 8M \\
& 64 & 48.0 & 8M \\
& 128 & 56.0 & 8M
\end{tabular}
\vspace{2mm}
\caption{Maximum throughput observed for triple generation. Measurements were taken at 64 threads. Batch sizes are the closest powers of two.}
\label{tab:triples}
}
\end{table}

\begin{figure}[t]
    \centering
    \includegraphics[width=0.9\columnwidth]{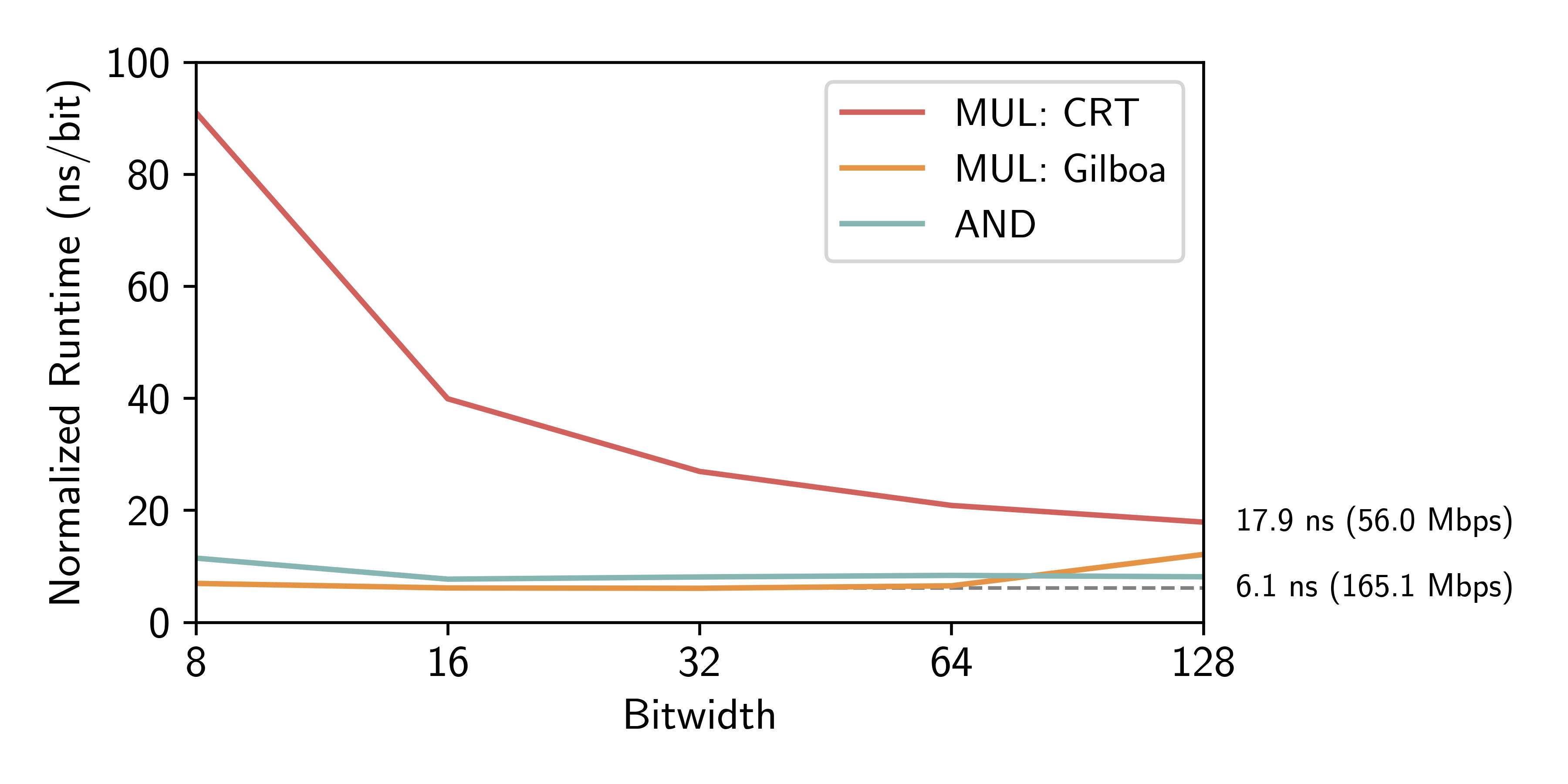}
    \caption{\tocs{Normalized fastest runtime for our triple-generation implementations. Dashed line represents the minimum Gilboa time observed; it slows down slightly at 128 bits. Reciprocal of the values in Table \ref{tab:triples}.}}\label{fig:triple-time}
\end{figure}

\begin{figure}[t]
    \centering
    \includegraphics[width=0.7\columnwidth]{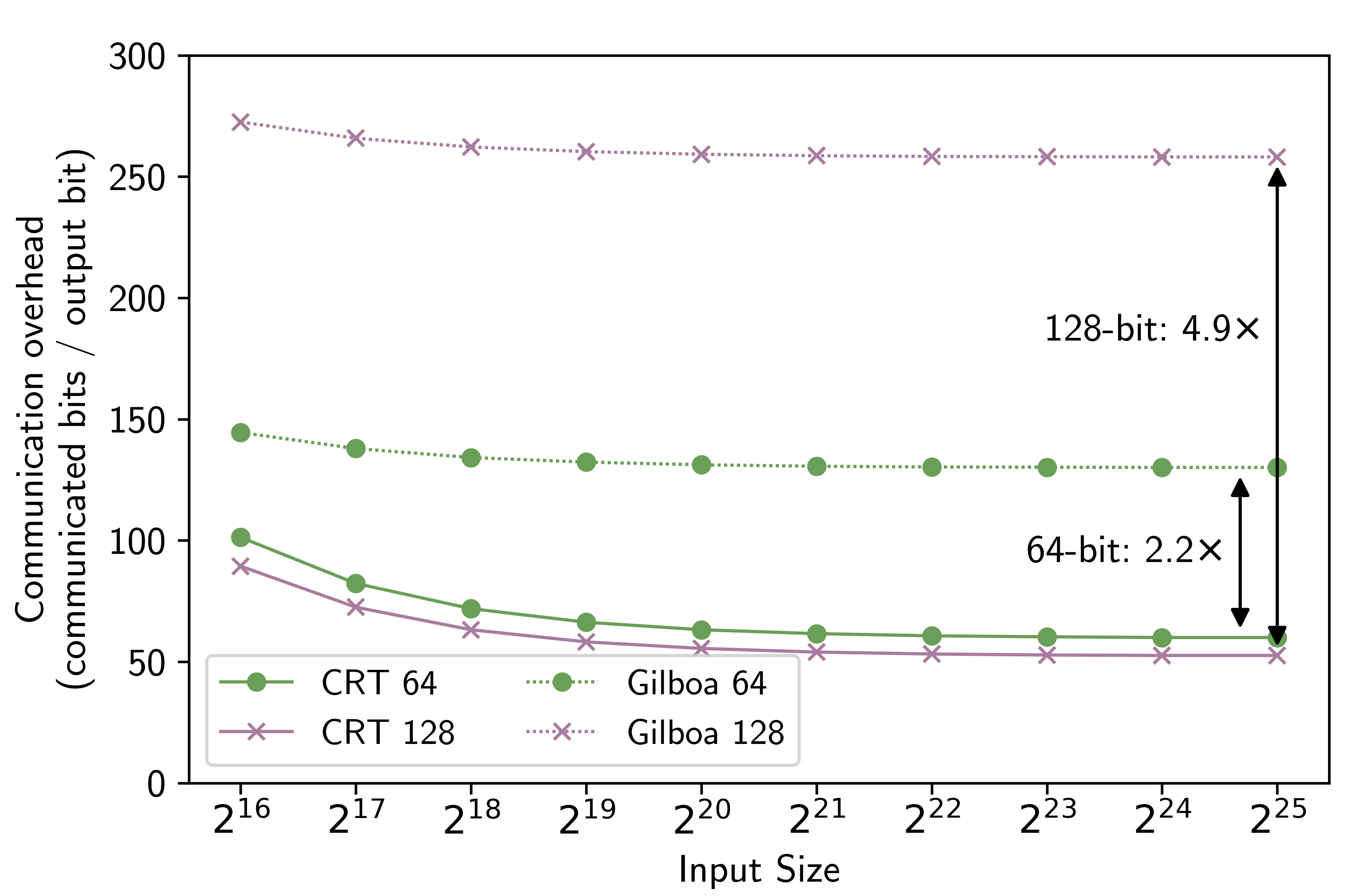}
    \caption{\tocs{Normalized communication overhead of the two arithmetic triple algorithms, shown for 64 and 128 bit outputs.}}\label{fig:triple-comm}
\end{figure}

}

\section{Related Work}\label{sec:related}

There is a large body of work on systems for secure computation.
In this section, we only focus on works that are directly related to \ours. 

\paragraph{Relational MPC}
Most systems and protocols for secure relational analytics target \emph{peer-to-peer} settings and propose techniques tailored to a small, fixed number of semi-honest data owners, who also act as computing parties: either two \cite{fang2024secretflow, Wang2021Secure, peng2024mapcomp, Bater2017SMCQL, bater2018shrinkwrap, bater2020saqe} or three~\cite{Volgushev2019Conclave, luo2024secure}.
In such settings,  performance can be improved by splitting the query into (i) a plaintext part that data owners compute locally (e.g., filters, hashing, sorting, and grouping), and (ii) an oblivious part that is executed under MPC (e.g., joins,  aggregations).  
Some of these works report results for a small subset of TPC-H queries on a maximum input size of $10\mathrm M$ rows per table~\cite{fang2024secretflow} (including rows processed in plaintext).
Senate~\cite{Poddar2021Senate} proposes a custom malicious-secure \camera{circuit} decomposition protocol that supports any number of parties and provides strong security against dishonest majorities. Even though the codebase is not publicly available,  published results indicate that the protocol's scalability is limited:
the paper reports performance for a subset of TPC-H queries
with 16k input rows (Scale Factor $\sim$0.004), and three queries from ~\S\ref{sec:eval} (\texttt{Pwd}, \texttt{Comorb.}, and \texttt{Credit}) with up to 80k rows.

Contrary to \ours, join optimizations in the above works rely on the assumption that each party owns one or more input tables. 
For example, they can improve performance in a scenario where a hospital and an insurance company wish to identify common patients with high premiums, but they are not effective (or do not even apply) in cases where multiple hospitals contribute to the patients table. \ours's operators are designed for the more general outsourced setting, hence, they are independent of the number of participants and data schema.  Moreover,  while \camera{circuit} decomposition techniques pay off in early stages of complex queries, later stages still need to be evaluated under MPC by all parties.  As a result, \ours could also be used in the peer-to-peer setting to speed up the oblivious (and by far most expensive) part of complex queries. Other systems for outsourced MPC like Scape~\cite{FengScape2022} (not publicly available) and Secrecy \cite{secrecy} report results for queries with up to two joins and maximum input sizes of 2M rows (for joins) and 8M rows (for other operators).
\camera{Sharemind~\cite{student-taxes} (also not publicly available) reports results for queries with a single join and up to 17M input rows}.

\paragraph{Oblivious join operators.}
Several works from the cryptography community study individual join operators in the outsourced setting, building from research into private set intersection (PSI).
Mohassel et al.~\cite{Mohassel2020Fast} propose efficient one-to-one joins with linear complexity, but rely on the LowMC library that has been cryptanalyzed~\cite{10.1007/978-3-030-84252-9_13}.
For one-to-many joins, Krastnikov et al.~\cite{Krastnikov2020Efficient}, Badrinarayanan et al.~\cite{Badrinarayanan2022Secret}, and Asharov et al.~\cite{Asharov2023Secure} propose protocols based on sorting. They report results with up to 1M rows per input table and do not evaluate complex queries like in \ours.
These operators have the same asymptotic costs as our join-aggregation operator, but are tailored to 3-party computation protocols~\cite{Badrinarayanan2022Secret, Asharov2023Secure} or target Trusted Execution Environments (TEEs)~\cite{Krastnikov2020Efficient}. \tocsCam{DJoin \cite{djoin} hides intermediate result sizes using differentially-private noise, but is only able to evaluate cardinality queries with simple predicates.} \tocs{While the one-to-many join protocol of Badrinarayanan et al.~\cite{Badrinarayanan2022Secret} is reminiscent of ours (their ``aggregation trees'' are actually just prefix networks, like our \textsc{AggNet}), their many-to-many join protocol requires output-size leakage or a quadratic blowup.}
None of these works support many-to-many joins without suffering from the limitations we described in~\S\ref{sec:intro}.

State-of-the-art systems for relational analytics in TEEs~\cite{opaque, DBLP:conf/eurosys/DaveLPGS20, cryptoeprint:2025/183} also provide oblivious join and group-by operators. To protect access patterns from side-channel attacks,  join and aggregation  in these systems rely on bitonic sort, which requires $O(n\log^2 n)$ operations. Early systems perform the grouping using fast parallel scans~\cite{opaque}, while recent approaches rely on prefix sums~\cite{cryptoeprint:2025/183}. These TEE-based systems perform the join and the group-by one after the other. Consequently, they either leak intermediate result sizes~\cite{cryptoeprint:2025/183}, resort to worst-case padding~\cite{DBLP:conf/eurosys/DaveLPGS20}, or require users to provide an upper bound on the query result size~\cite{opaque}. In contrast,  \ours uses a \textit{composite} join-aggregation operator,  with oblivious quicksort and an \camera{aggregation network},  that requires $O(n\log n)$ operations in total.  The techniques used by \ours in the MPC context can also be used to achieve oblivious execution in enclaves, and this is an exciting direction for future~work.

\section{Conclusion}

We have presented \ours, a scalable, secure framework for evaluating complex relational analytics under MPC. 
We instantiate \ours in multiple threat models, contribute a unified join-aggregation operator to the literature, and implement state-of-the-art sorting protocols. 
We deploy \ours to obliviously execute queries with multi-way joins larger than any prior work. As presented, \ours requires data analysts to translate queries into our dataflow API; future work includes integrating \ours with an automatic query planner. Cryptographic advances in oblivious sorting will also directly lead to performance improvements in \ours.

\section*{Acknowledgements}

\camera{We thank the anonymous reviewers and artifact evaluators for their constructive feedback that substantially improved the paper. We also thank Selene Wu for implementing low-level optimizations in the \ours runtime, and Sakshi Sharma for testing our artifact. This work was developed in part with computing resources provided by the Chameleon~\cite{keahey2020lessons} and CloudLab~\cite{duplyakin2019design} research testbeds,  supported by the National Science Foundation.  Final results and evaluation used Amazon Web Services clusters through the CloudBank project (National Science Foundation Grant No.\ 1925001).  This material is based upon work supported by the National Science Foundation under Grant No.\ 2209194, REU supplement awards No.\ 2432612 and 2326580,  by the DARPA SIEVE program under Agreement No.\ HR-00112020021, and by a gift from Robert Bosch GmbH.}

\makeatletter
\@ifundefined{pdfdest}{%
  \def\pdfdest name#1XYZ{}%
}{}
\makeatother

\bibliography{references}

\appendix
\section*{Appendix}

\tocs{
    In the following sections, we provide additional details on our oblivious shuffling (Appendix \ref{apdx:shuffle}), radixsort (Appendix \ref{sec:radixsort-analysis}), quicksort (Appendix \ref{sec:bound-quicksort}), and join primitives (including arguments of correctness and more information on the trimming heuristic; Appendix \ref{apdx:join}). Appendix~\ref{apdx:bw} provides detailed bandwidth measurements for \ours and the systems we compare against. Appendices are not included in the peer-review process.
}

\section{Oblivious Shuffle}\label{apdx:shuffle}

This section describes a set of primitives related to obliviously shuffling a secret-shared vector. Existing works typically discuss such primitives for a fixed number of parties and specific MPC protocol (e.g., \cite{ahi22-radixsort,permutation-correlations}).
Our primary contribution in this section is
a single stack of primitives that can be used across a variety of MPC protocols and threat models. We also provide novel algorithms for obliviously inverting elementwise permutations and converting elementwise permutations between arithmetic and boolean sharings.

\subsection{Preliminaries}

A permutation is a bijective mapping from $[n]$ to $[n]$, where $[n]$ represents the set $\{1, \dots, n\}$. We denote the composition of permutations $\pi$ and $\rho$ by $\pi \circ \rho(\cdot) = \pi(\rho(\cdot))$. We represent permutations as index maps: starting from a data array $\vec{x}$, if $\pi_i=j$ then $\pi(\vec{x})$ maps the value $x_i$ to position $j$. As a special case, a \textit{sorting permutation} for an input vector $\vec{x}$ is a permutation $\sigma$ such that $\sigma(\vec{x})$ is sorted. We use $\mathds{I}$ to denote the identity permutation $\mathds{I} = (1, \dots, n)$. We use $\share{x}$ to denote that a value $x$ is secret-shared, and we do not distinguish in notation between an arithmetic and boolean sharing.

Like Asharov et al. \cite{ahi22-radixsort}, in this work we consider two different forms of secret-sharing for permutations:

\begin{description}

    \item[Elementwise Permutation.] A vector of secret-shared indices, denoted $\share{ \pi } = (\share{ \pi_1 }, \dots, \share{ \pi_n })$. The elements can either be arithmetic or boolean secret-sharings.

    \item[Sharded Permutation.] A composition of random permutations, denoted $\langle \pi \rangle = \pi_n \circ \cdots \circ \pi_1$. 

\end{description}

We reproduce an important fact about permutations from Observation 2.4 of Asharov et al. \cite{ahi22-radixsort}.

\begin{fact} \label{perm-fact}
    Let $\pi$ and $\sigma$ be permutations over the set $[m]$ and let $\vec{\sigma} = (\sigma(1), \dots, \sigma(m))$ be the vector of destinations of the permutation $\sigma$. Then $\pi(\vec{\sigma}) = \sigma \circ \pi^{-1}([m])$, i.e. a vector of destinations for the permutation $\sigma \circ \pi^{-1}$.
\end{fact}

\subsection{Local permutations}

Many of the oblivious shuffling primitives that will be discussed in the following sections involve locally generating and locally applying random permutations as subroutines. In this section, we discuss the algorithms used for generating and applying random local permutations and our methods for parallelizing these operations.

\paragraph{Generating random local permutations.}
To generate local random permutations, we use the Fisher-Yates shuffle algorithm \cite{fisher-yates-shuffle}, a single-threaded algorithm. While there exist algorithms for generating random permutations in parallel (such as MergeShuffle \cite{mergeshuffle}), we typically need to generate many random permutations at once. As a result, we parallelize the generation by distributing a batch of permutations among all cores, so that each permutation is generated in a single-threaded fashion.

We may additionally wish to generate identical random permutations among multiple parties. This can be achieved by having each party generate the permutation locally using a pseudorandom generator (PRG) with the same seed.

\paragraph{Applying random local permutations.}
Unlike generation, we apply random local permutations one at a time, so we need to parallelize each individual local permutation application. To apply a local permutation, we want each element $x_i$ of a vector $\vec{x}$ to be placed at index $\pi_i$ in the output vector $\vec{y}$. We can easily divide $\vec{x}$ and $\pi$ into contiguous blocks and give one to each thread. Since we must access random elements in $\vec{y}$, we cannot give each thread a contiguous block of $\vec{y}$. However, we can give each thread full write access to $\vec{y}$ and mathematically guarantee that no element will be written to more than once.

\subsection{Sharded permutation protocols}
\label{sub:sharded-permutation}

In this section, we describe protocols to generate and apply sharded permutations, which we denote as \texttt{genSharded}- \texttt{Perm} and \applyperm, respectively. These permutations apply to all three MPC protocols used in \ours. We explain separately our protocols in the honest-majority and dishonest-majority settings.

\paragraph{Honest majority.}
In the three-party setting with replicated secret-sharing, we use the methods for generating and applying sharded permutations by Asharov et al. \cite{ahi22-radixsort}, which we restate below for completeness. A sharded permutation $\langle \pi \rangle = \pi_3 \circ \pi_2 \circ \pi_1$ can be expressed as a replicated secret sharing
$$\langle \pi \rangle = ((\pi_{1}, \pi_{2}), (\pi_{2}, \pi_{3}), (\pi_{3}, \pi_{1})).$$

\noindent That is, each party $P_i$ holds two permutations: one in common with $P_{i+1}$ and one in common with $P_{i-1}$. To generate such a sharing, each pair of parties generates a random local permutation using their common PRG seed. However, there is one permutation that $P_i$ does not know, so $\pi$ remains secret.

Later when the parties call \applyperm to apply a sharded permutation, each pair of parties locally applies their common permutation and reshares the permuted result to the excluded party. First, $P_1$ and $P_3$ permute under $\pi_1$ and reshare to $P_2$, and then so on for $\pi_2$ and $\pi_3$.

We generalize the protocol to work with any replicated secret-sharing scheme in the honest-majority setting, with an emphasis on the four-party Fantastic Four~\cite{fantastic4} protocol. The three-party protocol proceeds in a sequence of rounds, where in each round a subset of the parties --- which we will call a \textit{shuffle group} --- locally permutes the vector and reshares it to the remaining party.
We generalize this idea of shuffle groups and design protocols in which each round contains a local permutation application and resharing by a single shuffle group. For security, we require that shuffle groups obey two properties:
\begin{enumerate}
    \item Each shuffle group can collectively hold a sharing of the input data. As a result, the size of a shuffle group is greater than $T$, which is why this approach is limited to the honest-majority setting.
    \item There must exist at least one shuffle group containing no corrupted parties, so that the composed permutation $\pi$ is unknown to the adversary.
\end{enumerate}

Using the Fantastic Four protocol as a concrete example \cite{fantastic4}, here is a set of shuffle groups in the semi-honest setting: $\mathcal{G} = \{\{P_0, P_1\}, \{P_2, P_3\}\}$.
Since there is only $T = 1$ corrupted party, each shuffle group possesses at least one copy of all secret shares, and there is a shuffle group with no corrupted parties.

There are two ways to extend this idea to the malicious-secure four-party setting. One trivial approach is to combine the above idea with the generic compiler from semi-honest to malicious security for the $\texttt{reshare}$ functionality proposed by Asharov et al. \cite{ahi22-radixsort}.
However, we take a different approach that provides ease of implementation and consistency with the other functionalities in the Fantastic Four protocol \cite{fantastic4}.
We use four shuffle groups of three parties each, so that each share is contained twice in each group, allowing for redundancy in resharing the value to the excluded party:
$$\mathcal{G} = \{\{P_0, P_1, P_2\}, \{P_1, P_2, P_3\}, \{P_2, P_3, P_0\}, \{P_3, P_0, P_1\}\}.$$

\noindent We can then base the malicious security of the oblivious sharded permutation application on the black-box security of the $\texttt{INP}$ protocol described in Fantastic Four \cite{fantastic4}. Specifically, the receiving party receives either two copies of each share or one copy and its hash. In either case, a single corrupted party cannot corrupt both values received, so cheating will be detected with overwhelming probability.

\paragraph{Dishonest majority.}
In the two-party setting with one dishonest party, we cannot use shuffle groups and therefore instead adopt the framework of Peceny et al.~\cite{permutation-correlations} for generating and applying sharded permutations.
Their approach is based on a \textit{permutation correlation}, which is a pair of tuples, one held by each party:
$(A, B_0), (B_1, \pi)$ such that $\pi(A) = B_0 + B_1.$
Here, $A, B_0, B_1 \in \mathbb{F}^n$ and $\pi$ is a random permutation over $n$ elements. $B_0$ and $B_1$ can be either arithmetic or boolean secret-shares, with the addition operator defined correspondingly.

To generate a sharded permutation, it suffices to generate two permutation correlations in a preprocessing phase, one with each party as the sender. Peceny et al. \cite{permutation-correlations} \camera{propose} two protocols for generating permutation correlations, both based on oblivious pseudorandom functions \camera{(OPRFs)}, with one of them additionally using pseudorandom correlation generators to achieve communication sublinear in the length of the elements (as opposed to the length of the permutations).

\camera{To apply a sharded permutation obliviously, we use the protocol $\Pi_{\mathrm{Comp-Perm}}$ from Peceny et al. \cite{permutation-correlations}. We make use of their library's OPRF implementation \cite{aprr24-oprf,permutation-correlations} and construct the rest of the permutation correlations from scratch. Similar ideas can be used to build \applyinvperm that applies the inverse of a sharded permutation. Since the output of the OPRF is a boolean secret sharing, we generate all permutation correlations with a boolean sharing and convert to an arithmetic sharing where needed using a boolean-to-arithmetic conversion protocol.}

\subsection{Shuffling framework}
\label{sub:shuffle-framework}

Our framework provides generic interfaces to the functionalities \genperm, \\ \applyperm, and \texttt{applyInverseShardedPerm} that call the protocols to generate and apply (respectively) a sharded permutation as described in the previous section depending on the desired MPC protocol and threat model.
In this section, we describe our abstract interface and the protocols for oblivious shuffling primitives. A detailed analysis of the complexity of each protocol in each setting can be found in Table \ref{tab:shuffle}.

\paragraph{The Permutation Manager abstraction.}
We now describe the $\texttt{PermutationManager}$, our abstraction for generating sharded permutations in a setting-agnostic manner. The goal is to make all differences in the generation and application of sharded permutations invisible to the higher level shuffle protocols. The $\texttt{PermutationManager}$ exposes two functions:

\begin{itemize}
    \item $\genperm\texttt{<T>(size, enc)}$
    \item $\genpermpair\texttt{<T1,T2>(size, enc1, enc2)}$
\end{itemize}

The function \genperm returns a random sharded permutation over $\texttt{size}$ elements of type $\texttt{T}$ (e.g., $\texttt{int32}$), where the elements have encoding $\texttt{enc}$ (either an arithmetic or boolean sharing).
The function \genpermpair returns two random sharded permutations over $\texttt{size}$ elements such that they correspond to the same random permutation. This is necessary if we wish to permute multiple columns of data according to the same sharded permutation, as is the case, for example, when applying or composing elementwise permutations. The two vectors have types $\texttt{T1}$ and $\texttt{T2}$ and encodings $\texttt{enc1}$ and $\texttt{enc2}$.
Both functions allow for batched, parallel generation of many sharded permutations. Since all computation is data-independent, our framework generates sufficiently-many sharded permutations in preprocessing; when called, these functions fetch an already-generated permutation.

In the honest-majority setting, sharded permutations are identical for all input types, regardless of the bitwidth or sharing type of the vectors to be permuted. To reserve a pair of sharded permutations which represent the same permutation, we simply generate a single sharded permutation $\langle \pi \rangle$ and return the pair $(\langle \pi \rangle, \langle \pi \rangle)$. The need to distinguish between a single sharded permutation and a pair of sharded permutations only arises in the dishonest-majority setting, as the underlying permutation correlations are type-dependent and cannot be securely reused. 

\paragraph{Applying and composing permutations.}
As described in Section \ref{sub:sharded-permutation},
our framework provides methods that apply a sharded permutation and its inverse.
It also provides methods for oblivious shuffling, permutation composition, and applying an elementwise permutation.

Our oblivious shuffle protocol $\texttt{shuffle}$ simply generates and applies a sharded permutation, as shown in Protocol \ref{shuffle}.

\begin{algorithm}[ht]
    \caption{$\Pi_{\mathrm{Shuffle}}: \texttt{shuffle}$}
    \label{shuffle}
    \SetKwInOut{Input}{input}
    \SetKwInOut{Output}{output}
    \Input{$\share{x}$}
    \Output{$\share{\pi(x)}$ for random $\pi$}
    $\langle \pi \rangle \leftarrow \genperm\texttt{<type}(x)\texttt{>}(\texttt{size}(x), \texttt{enc}(x))$ \\
    $\share{\pi(x)} \leftarrow \texttt{applyShardedPerm}(\share{x}, \langle \pi \rangle)$ \\
    \Return{$\share{\pi(x)}$}
\end{algorithm}

Next, Protocol \ref{apply-elementwise-perm} shows $\texttt{applyElementwisePerm}$, which generalizes protocol 4.2 of Asharov et al.~\cite{ahi22-radixsort}.
The main idea here (which we use also in subsequent protocols) is that if we obliviously shuffle the elementwise permutation, then it is safe to open the shuffled vector in line \ref{line:open-perm} while revealing nothing about the unshuffled vector.

\begin{algorithm}[ht]
    \caption{$\Pi_{\mathrm{ApplyElem}}: \texttt{applyElementwisePerm}$}
    \label{apply-elementwise-perm}
    \SetKwInOut{Input}{input}
    \SetKwInOut{Output}{output}
    \Input{$\share{x}, \share{\rho}$}
    \Output{$\share{\rho(x)}$}
    $\tau_1, \tau_2 \leftarrow \texttt{type}(x), \texttt{type}(\rho)$ \\
    $\epsilon_1, \epsilon_2 \leftarrow \texttt{enc}(x), \texttt{enc}(\rho)$ \\
    $\langle \pi_1 \rangle, \langle \pi_2 \rangle \leftarrow \genpermpair\texttt{<}\tau_1, \tau_2 \texttt{>}(\texttt{size}(x), \epsilon_1, \epsilon_2)$ \\
    $\share{\pi(x)} \leftarrow \texttt{applyShardedPerm}(\share{x}, \langle \pi_1 \rangle)$ \\
    $\share{\pi(\rho)} \leftarrow \texttt{applyShardedPerm}(\share{\rho}, \langle \pi_2 \rangle)$ \\
    $\pi(\rho) \leftarrow \texttt{open}(\share{\pi(\rho)})$ \label{line:open-perm} \\
    $\share{\rho(x)} \leftarrow \texttt{localApplyPerm}(\share{\pi(x)}, \pi(\rho))$ \\
    \Return{$\share{\rho(x)}$}
\end{algorithm}

Finally, Protocol \ref{compose-perm} shows $\texttt{composePerms}$, our method to compose two permutations $\share{\sigma}$ and $\share{\rho}$ that is based on protocol 4.3 of Asharov et al. \cite{ahi22-radixsort}.
For simplicity, we describe the case where both permutations have the same encoding (either arithmetic or boolean).
Otherwise, we can convert one encoding using the conversion protocol described next. Alternatively, we could allow the input elementwise permutations to have different sharing types, but we would then have to pay the additional cost of generating a pair of sharded permutations rather than a single sharded permutation.

\begin{algorithm}[ht]
    \caption{$\Pi_{\mathrm{Compose}}: \texttt{composePerms}$}
    \label{compose-perm}
    \SetKwInOut{Input}{input}
    \SetKwInOut{Output}{output}
    \Input{$\share{\sigma}, \share{\rho}$}
    \Output{$\share{\rho \circ \sigma}$}
    $\langle \pi \rangle \leftarrow \genperm\texttt{<type}(\sigma)\texttt{>}(\texttt{size}(\sigma), \texttt{enc}(\sigma))$ \\
    $\share{\pi(\sigma)} \leftarrow \texttt{applyShardedPerm}(\share{\sigma}, \langle \pi \rangle)$ \\
    $\pi(\sigma) \leftarrow \texttt{open}(\share{\pi(\sigma)})$ \\
    $\share{\pi \circ \sigma^{-1}(\rho)} \leftarrow \texttt{localApplyPerm}(\share{\rho}, (\pi(\sigma))^{-1})$ \\
    $\share{\rho \circ \sigma} \leftarrow \texttt{applyInverseShardedPerm}(\share{\pi \circ \sigma^{-1}(\rho)}, \langle \pi \rangle)$ \\
    \Return{$\share{\rho \circ \sigma}$}
\end{algorithm}

\begin{table*}[h!]
    \centering
    \begin{tabularx}{\textwidth}{lCCCCCC} 
    \toprule
     & \multicolumn{2}{c}{\textbf{2PC}} & \multicolumn{2}{c}{\textbf{3PC}} & \multicolumn{2}{c}{\textbf{4PC}} \\
    \cmidrule(lr){2-3} \cmidrule(lr){4-5} \cmidrule(lr){6-7}
    \textbf{Primitive} & \textbf{Comm.} & \textbf{Rounds} & \textbf{Comm.} & \textbf{Rounds} & \textbf{Comm.} & \textbf{Rounds} \\
    \midrule

    \texttt{genSharded} & $\mathrm{preproc.}$ & $\mathrm{preproc.}$ & - & - & - & - \\

    \texttt{applySharded} & $2 \ell n$ & $2$ & $6 \ell n$ & $3$ & $24 \ell n$ & $4$ \\

    \texttt{shuffle} & $2 \ell n$ & $2$ & $6 \ell n$ & $3$ & $24 \ell n$ & $4$ \\

    \texttt{applyElementwise} & $2 \ell n + 3 \ell_{\sigma} n$ & $5$ & $6 \ell n + 7 \ell_{\sigma} n$ & $7$ & $24 \ell n + 25 \ell_{\sigma} n$ & $9$ \\

    \texttt{compose} & $5 \ell_{\sigma} n$ & $5$ & $13 \ell_{\sigma} n$ & $7$ & $49 \ell_{\sigma} n$ & $9$ \\

    \texttt{invertElementwise} & $5 \ell_{\sigma} n$ & $5$ & $13 \ell_{\sigma} n$ & $7$ & $49 \ell_{\sigma} n$ & $9$ \\

    \texttt{convertElementwise} & $5 \ell_{\sigma} n$ & $5$ & $13 \ell_{\sigma} n$ & $7$ & $49 \ell_{\sigma} n$ & $9$ \\
    \bottomrule
    \end{tabularx}
    \caption{Communication and round complexities for 2PC, 3PC, and 4PC shuffling primitives. Here, $n$ denotes the number of elements in the input vector, $\ell$ denotes the bitwidth of the elements, and $\ell_{\sigma}$ denotes the bitwidth of a permutation ($\ell_{\sigma} = 32$ in our system).}
    \label{tab:shuffle}
\end{table*}

\paragraph{Novel shuffle primitives.}
We propose two novel protocols related to oblivious shuffling. The first converts an elementwise permutation from either an arithmetic sharing to a boolean sharing or from a boolean sharing to an arithmetic sharing. The second inverts an elementwise permutation.

To convert between arithmetic and boolean encodings, one trivial approach is to use arithmetic to boolean share conversions for each element of the permutation. \camera{In the dishonest-majority setting, we use this trivial approach.}
However, \camera{in the honest-majority setting,} we can take advantage of the additional structure imposed by permutations by recognizing that we already know every value contained in the plaintext vector, but not the ordering.
Specifically, we can open a shuffled version of the vector, re-share the opened vector under the desired type, and obliviously apply the inverse of the shuffle.
This approach, shown in Protocol \ref{convert-elementwise-perm}, is faster than the trivial protocol in the honest-majority.

\begin{algorithm}[ht]
    \caption{$\Pi_{\mathrm{Conv}}: \texttt{convertElementwisePerm}$}
    \label{convert-elementwise-perm}
    \SetKwInOut{Input}{input}
    \SetKwInOut{Output}{output}
    \Input{$\share{x}_T$, $T \in \{A:0, B:1\}$}
    \Output{$\share{x}_{1-T}$}
    $\tau \leftarrow \texttt{type}(x)$ \\
    $\langle \pi_1 \rangle, \langle \pi_2 \rangle \leftarrow \genperm\texttt{<}\tau, \tau\texttt{>}(x, T, 1-T)$ \\
    $\share{\pi(x)}_T \leftarrow \texttt{applyShardedPerm}(\share{x}_T, \langle \pi_1 \rangle)$ \\
    $\pi(x) \leftarrow \texttt{reveal}(\share{\pi(x)}_T)$ \\
    $\share{\pi(x)}_{1-T} \leftarrow \texttt{secretShare}(\pi(x), 1-T)$ \\
    $\share{x}_{1-T} \leftarrow \texttt{applyInverseShardedPerm}(\share{\pi(x)}_{1-T}, \langle \pi_2 \rangle)$ \\
    \Return{$\share{x}_{1-T}$}
\end{algorithm}

We now describe a protocol for inverting an elementwise permutation. Whereas in $\texttt{applyInverseShardedPerm}$ the inverse was not computed directly but applied to a vector, here we wish to compute the inverse so it can be composed with other elementwise permutations.
The protocol $\Pi_{\mathrm{Inv}}$ is simple: to invert an elementwise permutation, we simply obliviously apply it to the identity permutation. Let $\share{\pi}$ be the permutation we wish to invert. We create an identity vector $\mathds{I} = (1, \dots, n)$ and secret-share it. We then obliviously apply $\share{\pi}$ to $\share{\mathds{I}}$. The permuted identity vector is $\pi^{-1}$. By Fact \ref{perm-fact}, $\pi(\mathds{I}) = \mathds{I} \circ \pi^{-1} = \pi^{-1}$.
We provide a formal description in Protocol \ref{invert-elementwise-perm}, and we discuss how to use such an inversion protocol to extract the sorting permutation from a sorting protocol in section \ref{sec:wrapper}.

\begin{algorithm}[ht]
    \caption{$\Pi_{\mathrm{Inv}}: \texttt{invertElementwisePerm}$}
    \label{invert-elementwise-perm}
    \SetKwInOut{Input}{input}
    \SetKwInOut{Output}{output}
    \Input{$\share{\pi}$}
    \Output{$\share{\pi^{-1}}$}
    $\share{\pi^{-1}} \leftarrow \texttt{applyElementwisePerm}(\share{\mathds{I}}, \share{\pi})$ \\
    \Return{$\share{\pi^{-1}}$}
\end{algorithm}

\section{Radixsort analysis} \label{sec:radixsort-analysis}

Our radixsort is similar to the protocol of Bogdanov et al. \cite{sharemind-radixsort}, except we make use of recent advances in oblivious shuffling discussed in \S \ref{sub:shuffle-framework}.
For reasonable bitwidths like $\ell = 32$, our protocol can be seen as an optimization of Asharov et al. \cite{ahi22-radixsort} that achieves a significant improvement in round complexity and a mild improvement in communication complexity.

\paragraph{Comparison.}
The key difference between our works is that, whereas Asharov et al.\ runs \texttt{composePerms} after sorting each bit, we apply the permutation for each bit to a larger vector. While this increases the communication required for the permutation application step, it does not add additional rounds and it allows us to eliminate the compose step and therefore reduce the round complexity. We demonstrate analytically in Table \ref{tab:radixsort} and empirically in Figure \ref{fig:ahi22-comparison} that our protocol typically outperforms Asharov et al. \cite{ahi22-radixsort} for both $\ell = 32$ and $\ell = 64$ bits.
Because the codebase of Asharov et al. \cite{ahi22-radixsort} is proprietary, we provide our own reimplementation of their protocol that we use in the Figure \ref{fig:ahi22-comparison} benchmarks.

This may seem contradictory to the claims of Asharov et al. In fact, they claim that it is precisely the permutation composition step that leads to their improvement in performance over Bogdanov et al. The distinction is that the protocol of Asharov et al.\ has constant factors that scale better than Bogdanov et al. (and better than our protocol) as $\ell$ grows without bound, whereas our analysis focuses on the values $\ell = 32$ and $64$ that we use in our shuffling and sorting frameworks.

\paragraph{Relationship to Bogdanov et al. \cite{sharemind-radixsort}.}
Here we show that by instantiating the components of the Bogdanov et al. protocol with the more efficient variants proposed in recent works, we arrive at our protocol.
Concretely, there are three differences in our protocols.
First, their work converts the bit to sort into an arithmetic sharing and subsequently computes the $\share{ord}$ vector, whereas our work and Asharov et al.\ do the same using the $\texttt{genBitPerm}$ subprotocol.
Second, they apply the $\share{ord}$ permutation to the input vector using an oblivious shuffle, open, and local application, which we do with our equivalent $\texttt{applyElementwisePerm}$ protocol (Protocol 4.2 of Asharov et al.\ in the three-party setting).
Finally, the input padding and permutation extraction procedure described in Section 5 of Bogdanov et al. is equivalent to our procedure in \S \ref{sec:wrapper} for padding and permutation extraction, plus our protocol is generalized to include descending-order sorting.

\paragraph{Performance analysis.}
In the remainder of this section and in Table \ref{tab:radixsort}, we compare the costs of our protocol and Asharov et al. \cite{ahi22-radixsort}, specifically in the three-party setting as that is the only one they support.

\emph{Common elements.}
Both protocols invoke $\texttt{genBitPerm}$ $\ell$ times. Each invocation incurs $n$ calls to $\texttt{b2abit}$ and $n$ multiplications that cost $3 \ell_{\sigma} n$ and $\ell_{\sigma} n$ bits of communication, and $3$ and $1$ rounds of communication, respectively,
with elements of size $\ell_{\sigma}$ since $\texttt{genBitPerm}$ computes a permutation. In total, each protocol communicates $4 \ell_{\sigma} n$ bits over $4$ rounds for each call to $\texttt{genBitPerm}$, amounting to $4 \ell \cdot \ell_{\sigma} n$ bits of communication and $4 \ell$ rounds in total.

\emph{Costs of Asharov et al.}
There are two additional costs in Asharov et al.~\cite{ahi22-radixsort}: $\ell - 1$ calls to $\texttt{applyElementwisePerm}$ and $\ell - 1$ calls to $\texttt{composePerms}$.
First, each call to \\\texttt{applyElementwisePerm} operates over a single bit; using Table \ref{tab:shuffle}, we see that the total cost is $(\ell - 1) (6 n + 7 \ell_{\sigma} n)$ bits of communication over $7 (\ell - 1)$ rounds.
Second, the calls to $\texttt{composePerms}$ collectively cost $13 (\ell - 1) \ell_{\sigma} n$ bits of communication and $7 (\ell - 1)$ rounds of communication.
When adding the cost of the $\ell$ calls to $\texttt{genBitPerm}$, the total cost of the protocol is
$$24 \ell \cdot \ell_{\sigma} n - 
20 \ell_{\sigma} n +
6 (\ell - 1) n$$
bits of communication over $18 \ell - 14$ rounds.

\emph{Our costs.}
In our protocol, we eliminate the permutation compositions and instead apply the permutations to a vector of bitwidth $\ell + \ell_{\sigma}$. We make:
\begin{itemize}
    \item $\ell - 1$ calls to $\texttt{applyElementwisePerm}$, each costing $6 (\ell + \ell_{\sigma}) n + 7 \ell_{\sigma} n$ bits of communication and $7$ rounds.
    \item One call to each of $\texttt{convertElementwisePerm}$ and $\texttt{invertElementwisePerm}$, adding $26 \ell_{\sigma} n$ bits of communication and $14$ rounds of communication.
\end{itemize}
Therefore, in total, our protocol has a cost of
$$
17 \ell \cdot \ell_{\sigma} n +
13 \ell_{\sigma} n +
6 \ell^2 n
- 6 \ell n
$$
bits of communication and $11 \ell + 7$ rounds of communication.

\begin{figure}[t]
    \begin{subfigure}[b]{\columnwidth}
        \centering
        \includegraphics[width=0.8\columnwidth]{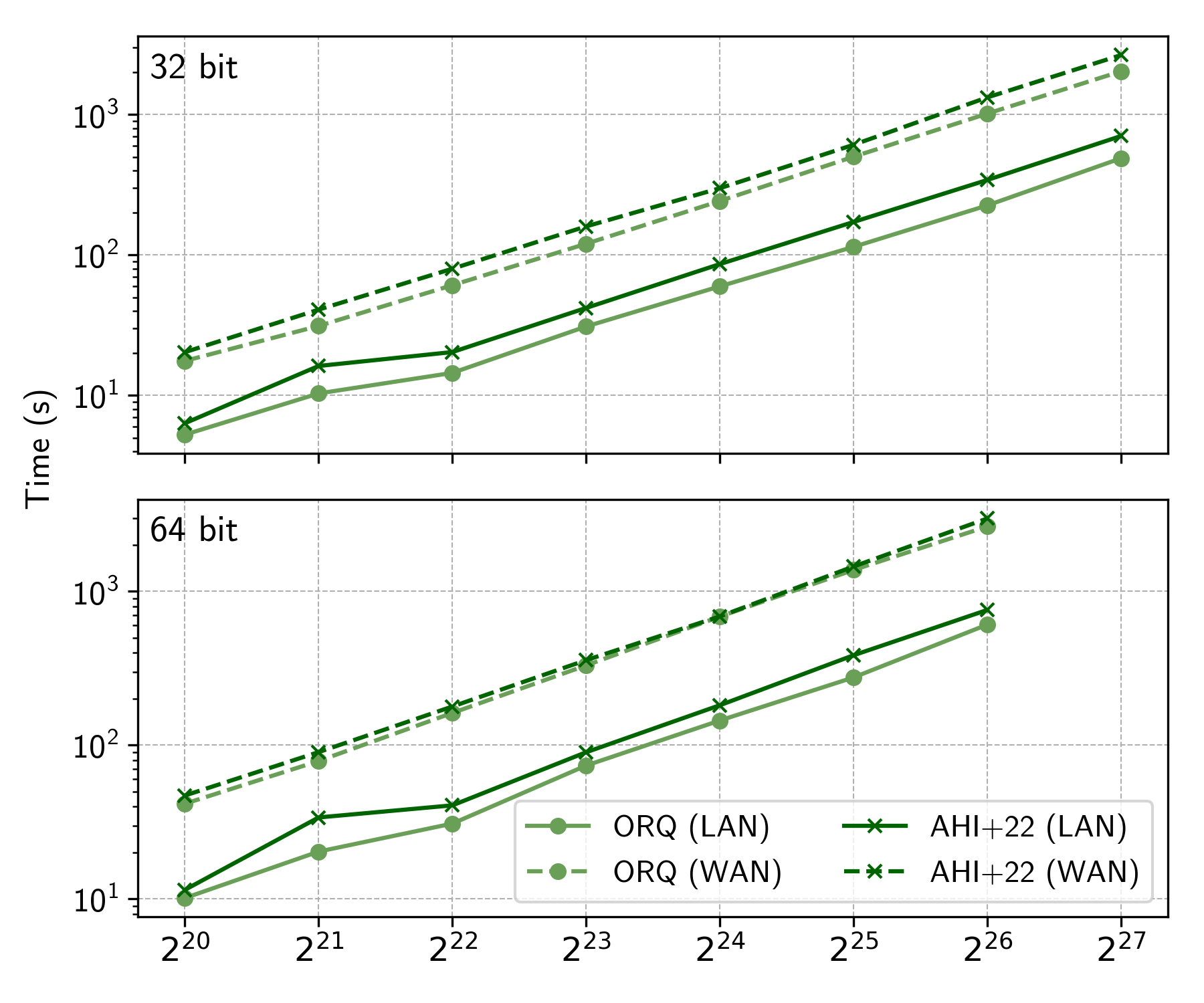}
        \label{fig:ahi22-comparison-32b}
    \end{subfigure}
    \caption{Comparison of our radixsort protocol with Asharov et al. \cite{ahi22-radixsort} for (a) $\ell = 32$ and (b) $\ell = 64$. All data points are the average of $3$ runs and are run with $34$ threads in the 3-party setting. WAN data points are collected in a WAN environment with 20ms latency. 64-bit radixsort runs out of memory for $2^{27}$ input rows. Our hybrid protocol wins in all scenarios by up to $1.44\times$.}
    \label{fig:ahi22-comparison}
\end{figure}

\begin{table*}[h!]
    \centering
    \begin{tabular*}{\textwidth}{@{\extracolsep{\fill}} >{\small\itshape}p{0.3\textwidth} p{0.35\textwidth} p{0.35\textwidth}}
    \toprule
    & \textbf{Asharov et al. \cite{ahi22-radixsort}} & \textbf{Ours} \\
    \midrule
    Asymptotic Comm. & $O(\ell n \log n)$ & $O(\ell^2 n)$ \\

    Concrete Comm. & $[\ell_\sigma(24 \ell - 20) + 6 (\ell - 1)] n$ &  $[\ell_{\sigma} (17 \ell + 13) + 6 \ell(\ell - 1)]n$ \\

    Asymptotic Rounds & $O(\ell)$ & $O(\ell)$ \\

    Concrete Rounds & $18 \ell - 14$ & $11 \ell + 7$ \\

    $\ell = 1$ Comm. & $128 n \quad (4$ rounds$)$ & $960 n \quad (18$ rounds$)$ \\

    $\ell = 32$ Comm. & $24122 n \quad (562$ rounds$)$ & $23776 n \quad (359$ rounds$)$ \\

    $\ell = 64$ Comm. & $48890 n \quad (1138$ rounds$)$ & $59424 n \quad (711$ rounds$)$ \\
    \bottomrule
    \end{tabular*}
    \caption{Cost analysis of the state-of-the-art radixsort protocol by Asharov et al. \cite{ahi22-radixsort} and our hybrid radixsort protocol.}
    \label{tab:radixsort}
\end{table*}

\paragraph{Discussion.}
In Table \ref{tab:radixsort}, we describe the costs both as a function of $\ell$ and $\ell_{\sigma}$ and for three values of the bitwidth. Our protocol has the better concrete round complexity, scaling as $11 \ell$ rather than $18 \ell$ (although they are asymptotically equivalent).
For communication complexity, because $\ell$ and $\log n$ are often similar, the two protocols are typically asymptotically equivalent.
Concretely, our protocol has the smaller constant factor for the $\ell \cdot \ell_{\sigma}$ term, which for typical values of $\ell$ close to $\ell_{\sigma}$ (e.g., $\ell = \ell_{\sigma} = 32)$ means our protocol will require less communication. However, only our protocol has an $\ell^2$ term, so it is a poor choice when $\ell \gg \ell_{\sigma}$.

We discuss a few specific examples that are shown in Table~\ref{tab:radixsort} for the setting $\ell_\sigma = 32$ (since we never sort more than $4$ billion elements due to speed constraints).
\begin{itemize}
    \item If $\ell = 32$ as well, then our protocol saves a modest $1.4\%$ improvement in total communication and a significant $36\%$ improvement in communication rounds. In practice, the improved round complexity leads to substantially better performance, which we demonstrate empirically in both the LAN and WAN settings in Figure \ref{fig:ahi22-comparison}.
    \item If $\ell = 64$, then our protocol requires $22 \%$ more communication but decreases the number of rounds by $37 \%$. Once again, Figure \ref{fig:ahi22-comparison} shows empirically that \ours performs better in both the LAN and WAN settings when $\ell = 64$, although by a slimmer margin than with $\ell = 32$.
    \item If $\ell = 1$ then the constant additive overhead of input padding and permutation extraction makes our protocol inferior to Asharov et al.
\end{itemize}

\section{Bounding preprocessing for Quicksort} \label{sec:bound-quicksort}

In the two-party SH-DM setting, we need to preprocess Beaver triples for the quicksort comparisons. Quicksort involves a nondeterministic number of comparisons, but we would like to probabilistically bound the running time so we can generate sufficient triples ahead of tim (generating triples in a large batch is substantially faster than generate multiple smaller batches). For $n$ input elements, we generate Beaver triples for $2 n \lg n$ comparisons, which is sufficient approximately $99.9\%$ of the time. We make use of the work of McDiarmid et al. \cite{large-deviations-quicksort}, which discusses this problem in the case of a random pivot selection. Since we shuffle the input list, our method of pivot selection chooses a random element, so the analysis applies.

The expected number of comparisons $q_n$ for quicksort with uniformly random pivot selection as $n \to \infty$ is
\begin{align*}
    q_n &= 2n \ln n - (4 - 2 \gamma)n + 2 \ln n + O(1) \\
    & \leq 2n \ln n = (\ln 4)\;  n \lg n\leq 1.39 n \lg n.
\end{align*}
This says nothing about how concentrated the distribution is around this expectation.  Let $Q_n$ be the random variable representing the actual number of comparisons required for a given run of quicksort. McDiarmid et al. \cite{large-deviations-quicksort} discuss the probability
$p = \Pr\left[\left|Q_n/q_n - 1\right| > \varepsilon\right]$
describing the concentration.
Here, $(\varepsilon+1)$ is the \textit{multiplicative} overhead over the expectation $q_n$. We want to solve for $\varepsilon$ as a function of both the input size $n$ and the probability $p$. We set the probability to a constant $p = 2^{-10}$. This is significantly larger than a typical statistical failure, because if we ``fail,'' we can simply generate more triples in the online phase with no impact on security and an acceptable performance penalty. The chosen value of $p$ makes our generated randomness sufficient for $99.9\%$ of executions.

Theorem 1 of McDiarmid et al. \cite{large-deviations-quicksort} gives the following expression of the above probability.
$$p = n^{-2\varepsilon(\ln(\ln(n)) - \ln(\frac{1}{\varepsilon}) + O(\ln^{(3)}(n)))}
\leq n^{-2\varepsilon(\ln(\ln(n)) - \ln(\frac{1}{\varepsilon}))}$$
To write $\varepsilon$ as a function of $p$, we can take the natural log of both sides of this inequality. Even so, it is not feasible to calculate 
an analytic solution to this equation due to the presence of both $\varepsilon$ and $\ln(1/\varepsilon)$ terms.
To simplify the equation, we use the bound $\varepsilon \geq 0.43$ because $1.39 \cdot 1.43 \approx 2$. We simplify the equation using this value.

$$\varepsilon \leq -\frac{\ln(p)}{2 \ln(n)(\ln(\ln(n)) - \ln(1 / 0.43))}$$

For $p = 2^{-10}$, $\varepsilon \leq 0.43$ for all $n \geq 1300$, meaning the chosen value of $\varepsilon = 0.43$ is suitable for all $n \geq 1300$. To offset the fact that values of $n < 1300$ are not covered by this case (however unlikely it may be for us to need a sort that small in practice), we handle this case separately with an additive buffer of $10,000$ triples whenever the input size is below $n = 2000$. This buffer is loose but sufficient, and it will not have a noticeable impact on performance for any reasonable use of the library.
\section{Analysis of the Join Operator}\label{apdx:join}

In this appendix, we address the security of our join operator and then describe its correctness in detail.

\paragraph{Security.}
Security comes naturally from our oblivious building blocks: since we never open any secret-shared data, and make no assumptions about the true cardinality of input or output tables, any execution of the join operator on tables $L_1, R_1$ is identically distributed to any other execution on similarly-sized tables $L_2,R_2$; $|L_1|=|L_2|$ and $|R_1|=|R_2|$. Note that here we refer to the observable size of the table in memory (both rows and columns) and not the number of valid rows (which would be the ``size'' of the table in the conventional sense). For example, even if $L_2$ and $R_2$ consisted entirely of (secret-shared) zero values, an execution of the join protocol would be indistinguishable from an execution on real values.

In more detail, we can analyze the security of the join operator using the arithmetic black-box model (e.g., \cite[\S 2.1]{DBLP:conf/crypto/EscuderoGKRS20}), which provides an ideal functionality abstraction of all primitive MPC protocols for secret-sharing and reconstructing secrets, performing arithmetic and boolean operations, and converting between arithmetic and boolean shares.
Through inspection of Protocols \ref{alg:prefixagg}, \ref{alg:sorting-wrapper}, and \ref{alg:inner-join}, we can see that the join operators in \ours make use of MPC primitives in a black-box way. Concretely, the join-aggregation in Protocol \ref{alg:inner-join} only uses oblivious logical operators (e.g., \texttt{AND} and \texttt{NOT}), (in)equality comparisons to find distinct keys, and calls to underlying methods like \protoPrefixAgg{}, \textsc{TableSort}, and Concat.
These methods, in turn, are shown in Protocols \ref{alg:prefixagg}--\ref{alg:inner-join} to rely exclusively on oblivious comparisons and logical operations along with the oblivious shuffling and sorting protocols defined in Appendix \ref{apdx:shuffle} and \S \ref{sec:sort}. Finally, inspection of all shuffling and sorting protocols (in Protocols \ref{shuffle}--\ref{sorting-wrapper}) reveals that they only rely on oblivious arithmetic and boolean operators.
Importantly, none of these algorithms opens sensitive data; the only openings are on permuted data (Protocols \ref{apply-elementwise-perm}--\ref{compose-perm}) or control flow data (Protocol \ref{iterative-quicksort}) that have been carefully constructed so that the distribution of openings is independent of the data. As a result, \ours join operators inherit the security guarantees of the underlying MPC protocols.

\paragraph{Correctness.}
In the remainder of this section, we argue the correctness of the inner join operation and its variants. We will use the following important fact when reasoning about joins:

\begin{fact} \label{fact:never-invalidate}
    Invalid rows are never revalidated.
\end{fact}
This fact is guaranteed by only operating on the valid column using logical AND. Thus, rows may move from valid to invalid ($\mathtt{Valid} \wedge 0 \rightarrow \mathtt{Invalid}$) but once invalidated can never be revalidated ($\forall x:\mathtt{Invalid} \wedge x \rightarrow \mathtt{Invalid}$). Rows are only validated when creating a \textit{fresh} table.

\begin{theorem}[Informal]
    The inner join protocol is correct.
\end{theorem}

\begin{proof}[Proof Sketch]

Define the key and valid columns of each input table as $L.k, L.V; R.k, R.V$, respectively. Without loss of generality we only consider a single abstract data column, $L.a$ and $R.b$, in each of the two tables. Assume the left table $L$ has unique keys --- each distinct value of $k$ occurs at most once --- and $R$ has an arbitrary key distribution. Note that primary-key foreign-key relations are a strict subset of unique-key many-key relations, since keys on the right in a one-to-many relation need not exist on the left. Finally, while we discuss the case of a single key $k$ below, our protocol transparently handles compound keys, $k_1\ ||\ k_2\ ||\ \dots$, exactly as if all such key columns were concatenated into one.

The first step of the protocol is to concatenate the two tables. We concatenate by merging the schemas of the two tables and copying each table's rows into the appropriate columns. All other values are zero (such as within left-table rows but under right-table columns). We additionally append a $T_{id}$ column, which is $0$ for left-table rows and $1$ for right-table rows, and sort by $V\ ||\ k\ ||\ T_{id}$. Any invalid rows in either table would remain invalid after concatenation and thus be sorted to the top of the table. The resulting table $O=\textsc{Concat}(L, R)$ may have the form:

\begin{table}[h!]
    \centering
    \begin{tabular}{|c|cc|cc|}
        \toprule
        $O.k$ & $O.a$ & $O.b$ & $O.V$ & $O.T_{id}$ \\ \hline
        - & - & - & 0 & - \\
        $k_1$ & $L.a_1$ & 0 & 1 & 0 \\
        $k_2$ & $L.a_2$ & 0 & 1 & 0 \\
        $k_2$ & 0 & $R.b_1$ & 1 & 1 \\ 
        $k_2$ & 0 & $R.b_2$ & 1 & 1 \\ 
        $k_2$ & 0 & $R.b_3$ & 1 & 1 \\ 
        $k_3$ & 0 & $R.b_4$ & 1 & 1 \\
        \bottomrule
    \end{tabular}
\end{table}

Next we discuss the two main functions of our inner join algorithm: (a) identifying matching rows, according to the join keys, and (b) copying data into the result of the join.

We first turn to the identification of matching rows. This step takes the form of invalidating rows which do not belong to the output of the join. In the case of an inner join, we only keep those (valid) rows from the right for which their key matches an associated unique key on the left:
\begin{align*}
O.V_i := O.V_i &\wedge [O.T_{id,i} = 1]&\text{valid row from }R\\
               & \wedge [\exists j : O.k_j = O.k_i &\text{exists matching key}\\
               & \wedge O.T_{id,j} = 0 &\text{from left}\\
               & \wedge O.V_j = 1] &\text{which is valid}
\end{align*}
Due to previously having sorted on $V||k||T_{id}$, we know that if any such $O.k_j$ exists, it will occur as the first position of this $k$-group. If not, a row from the right appears as the first position of this $k$-group, and no rows with this key are in the output of the join. We use the \textsc{Distinct} operator to mark the first row of each group with a $1$. Then, the join output is formed by updating a temporary valid bit, $O.V_o := O.V \land \neg\textsc{Distinct}(V,k)$. That is, $O.V_o$ invalidates any valid rows which are the \textit{first} of their $(V\!,k)$ group.\footnote{We cannot directly update $O.V$, because it will later be used as an aggregation key, during the execution of which we propagate changes to $O.V_o$. However, changing aggregation keys during execution breaks correctness, so we keep a temporary column.}

For each group four possibilities exist, shown in the tables below.


\begin{table}[h!]
\centering
\begin{floatrow}[2]
\ttabbox{%
    \begin{tabular}{|c|c|c|}
        \hline
        $O.k$ & $\neg\textsc{Dist}$ & $O.V$ \\ \hline\hline
        \dots  & \dots & 0   \\
        \dots  & \dots & \dots \\ \hline
        $L.k$    & 0     & 1   \\
        $L.k$    & 1     & 1   \\
        $L.k$    & 1     & 1   \\
        $L.k$    & 1     & 1   \\ \hline
    \end{tabular}
}{\caption*{Case I}}
\ttabbox{%
    \begin{tabular}{|c|c|c|}
        \hline
        $O.k$ & $\neg\textsc{Dist}$ & $O.V$ \\ \hline\hline
        \dots & \dots & 0   \\
        \dots & \dots & \dots \\ \hline
        $L.k$   & 0     & 1    \\
        $L.k$   & 1     & 1    \\ \hline
        $R.k$   & 1     & 1    \\
        $R.k$   & 1     & 1    \\ \hline
    \end{tabular}
}{\caption*{Case II}}
\end{floatrow}
\vspace{16pt}
\begin{floatrow}[2]
\ttabbox{%
    \begin{tabular}{|c|c|c|}
        \hline
        $O.k$ & $\neg\textsc{Dist}$ & $O.V$ \\ \hline\hline
        \dots & \dots & 0   \\
        \dots & \dots & \dots \\ \hline
        $L.k$   & 0     & 1    \\ \hline
        $R.k$   & 1     & 1    \\
        $R.k$   & 1     & 1    \\
        $R.k$   & 1     & 1    \\ \hline
    \end{tabular}
}{\caption*{Case III}}
\ttabbox{%
    \begin{tabular}{|c|c|c|}
        \hline
        $O.k$ & $\neg\textsc{Dist}$ & $O.V$ \\ \hline
        \dots& \dots & 0   \\
        \dots& \dots & \dots \\\hline
        $R.k$  & 0     & 1    \\
        $R.k$  & 1     & 1    \\
        $R.k$  & 1     & 1    \\
        $R.k$  & 1     & 1    \\ \hline
    \end{tabular}
}{\caption*{Case IV}}
\end{floatrow}
\end{table}

\paragraph{Case I} There are duplicate keys on the left, and no matching keys on the right. The validity update procedure marks the first row from the left invalid, and aggregation (below) will mark the entire group invalid. This represents (possibly duplicated) primary keys with no foreign key rows.

\paragraph{Case II} There are duplicate rows on the left, and greater than zero matching keys on the right. The first left row is marked invalid, and aggregation similarly invalidates the entire left group. The right group remains valid.

\paragraph{Case III} This is another matching case, but with exactly one primary key on the left. The single row from the left is invalidated, but the aggregation makes no changes. Case III is frequently observed in realistic workloads, where our join operator operates on an explicit primary-key foreign-key relation.

\paragraph{Case IV} There are no matching rows on the left. Now we invalidate the \textit{first row from the right}, and the subsequent aggregation invalidates \textit{all rows on the right}. With no key on the left, this key is not be part of the output table.
\\

In Cases I and II, there are multiple matching ``primary keys''; we operate as if only the first existed. We tend to avoid this mode of operation, since it does not follow the semantics of pk-fk relations, but it is useful in some environments, such as when performing semi-join. It also allows us to natively perform \verb|COUNT(DISTINCT(...))| operations over many-to-many  joins (with duplicates on both sides).

After the invalid rows are marked, a final aggregation is applied (as part of the copy step, below) to similarly invalidate all rows in the group. This aggregation is a special case of the aggregation step outlined below, as the aggregation keys here are $V\ ||\ k\ ||\ T_{id}$. This allows for an invalidated left row to invalidate all other left rows, and an invalidated right row to invalidate all other right rows, but invalidation will never cross the table boundary.

The size of the output of a one-to-many inner join is bounded by the size of the right table: assume every key on the right is present exactly once on the left. Then the output of the join is exactly the right table. Because our protocol is oblivious, it must be input-independent, even in the worst case. Thus our output must always exactly the size of the right table; we use the valid column to mark rows not actually present in the join output.

The cardinality equivalence means it is most efficient to build the output of the join within the rows of the right table. As such, all data from the right is included in the output by default. Data from the left table, on the other hand, must be explicitly copied. This operation is performed with a \textit{group-by aggregation} and implemented via \protoPrefixAgg{}. In practice, we have noticed that this default behavior also matches the semantics of many queries, so few \texttt{copy} aggregations are required.

Groups are defined as $V\ ||\ k$ (note the exclusion of the table ID column). The aggregation performed is a simple \texttt{Mux}, where the selection bit denotes whether values belong to the same group. We perform a reverse aggregation, which has the effect of aggregating \textit{down}, and copying the first row of each multiplex group into all other rows of that group. Where a row on the left exists, it is the first row, due to the final sort by $T_{id}$; should multiple left rows exist, we only copy the value from the first. All left rows will be invalidated (Case II, above). The join operator does not support, at this stage, many-to-many-key joins. However, with a prior aggregation, or under certain conditions (e.g., the join performs a $\textsc{Distinct}(\cdot)$ aggregation), our operator still applies. See Section \ref{sec:join} and \ref{sec:duplicates} for more details.\end{proof}

\paragraph{Unifying join and aggregation.} The validity-update procedure of join looks very similar to aggregations performed over joined tables. For improved efficiency, we would like to combine these two operations into a single unified control flow, which will approximately halve the number of operations required.

\begin{theorem}[Informal]
    Aggregations can be performed in the same control flow as the join operator provided the following conditions are met:
    \begin{enumerate}
        \item The aggregation keys are the same as the join keys.
        \item No aggregation key is also an aggregation output.
        \item No aggregation input is also an aggregation output for a different aggregation.
    \end{enumerate}
\end{theorem}

The intuition behind this theorem is that the conditions preclude any interaction between individual aggregations, so they can be evaluated in parallel: \camera{calling $\protoPrefixAgg{}(f_1, \dots)$, $\protoPrefixAgg{}(f_2, \dots)$, $\protoPrefixAgg{}(f_3, \dots)$ is equivalent to}
$$
\protoPrefixAgg{}(\{f_1, f_2, f_3\}, \dots)
$$
If any of the conditions are not true, these calls are \textit{not} independent (and may have side effects), so sequential calls to the operator are required.

\paragraph{Handling other scenarios.} Our operator only provides correct semantics for the case where there are at most unique keys on the left. To handle duplicate keys, we must first pre-aggregate over the left key. This requires an additional call to \protoPrefixAgg{}, after which the aggregated table has unique keys and aggregated values. Then, our operator can be applied. We note that this collapsing is only valid for self-decomposable aggregation functions.

This general technique is also applicable, as alluded to above, if one of the join or aggregation keys are a prefix of the other. (i.e., $K_j = K_a || K^*$ or $K_a = K_j || K^*$). Then, we sort on the longest combined key ($\mathrm{argmax}(|K_j|, |K_a|)$) and perform the join normally. If the aggregation keys are a prefix of the join keys, we need to handle the possible interspersed rows from the left ($T_{id}=0$): \protoPrefixAgg{} requires all rows in a group to be adjacent. We can accomplish this by masking rows from the left from an identity element for the given aggregation (e.g. for \texttt{sum}, 0; for \texttt{min}, $\infty$.) Then, while the join occurs over $K_j=K_a || K^*$, we can still aggregate over $K_a$ by ignoring all $K^*$ keys and taking advantage of the masked identity elements.

Finally, if the aggregation is applied to columns from both the left and the right (as in the Secure Yannakakis query, \cite{Wang2021Secure}) we must first break the aggregation function apart, apply either of the techniques above, and post-aggregate over the joined table. In the Secure Yannakakis example, we have:
\begin{lstlisting}[label={syan},language=SQL]
SELECT T.class, SUM(S.cost * (1 - R.coinsurance))
FROM R, S, T
WHERE R.person=S.person AND S.disease=T.disease
GROUP BY T.class;
\end{lstlisting}
We observe that we can first pre-aggregate the quantity \verb|1-R.coinsurance| per \texttt{person}, join $R\bowtie S$ on \texttt{person}, compute the product with \texttt{cost}, and then post-aggregate by summing over \texttt{disease} and \texttt{class}. This rewriting applies because the \texttt{SUM} function is self-decomposable.

\paragraph{Custom Aggregatons.} \ours users can also define custom aggregation functions, as in the following example that computes a product:
\begin{lstlisting}[label={prod},language=C++]
template<typename A>
A prod(const A& a, const A& b) {
    return a * b;  // User-defined logic here
}
\end{lstlisting}
\ours will call $\texttt{Mux}(g, a, \texttt{prod}(a, b))$ within the aggregation control flow, where $\texttt g$ denotes whether the elements belong to the same group and \texttt a, \texttt b are slices of the vector to be aggregated. For each pair of elements in the same group ($g_i=1$), the aggregation function \texttt{prod} is applied; otherwise, no update occurs. For the example above, this will result in obliviously multiplying all elements of each group.

\paragraph{Unique-key joins.} An additional optimization is possible in the case of unique keys on both sides of the join. (We assume this information is included in the tables' public schema.) In this case, the call to \protoPrefixAgg{} is unnecessary: since each row on the left has at most one match on the right, the call to \textsc{Distinct} is sufficient for computing the join; the output of is then bounded by $\min(|L|,|R|)$ rather than just $|R|$. Additionally, in this context, aggregating over the join keys is a nonsensical operation, so we omit it. This operation is now effectively just a PSI protocol; many more optimizations are possible. However, we only use this optimized join algorithm for the comparison with SecretFlow, since the only join implementation they provide is a PSI-join which requires unique keys.

\subsection{Correctness of other types of joins}

\paragraph{Outer joins.} Outer joins work similarly to an inner join but with different validation rules. That is,
\begin{itemize}
    \item A \textbf{left-outer join} is an inner join, plus all rows from the left.
    \item A \textbf{right-outer join} is an inner join, plus all rows from the right.
    \item A \textbf{full-outer join} contains all rows from both input tables, and is effectively just a concatenation with aggregation.
\end{itemize}

Full-outer join is the simplest case. Since we take all rows, there is nothing new to invalidate. We copy the valid bits from each input table and apply them to the joined output, which is just a concatenation $L||R$. (We may also sort, so that user-defined aggregations can be applied, if necessary.)

For a right-outer join, we do not invalidate any rows from the right. But by the semantics of our join operator, we only include new columns from the left if explicitly specified by the user with a \texttt{copy} aggregation. Any rows \textit{originally} from the left table should always be removed. This validation rule, therefore, is also quite simple: $O.V \leftarrow O.V_{LR}\wedge O.T_{id}$, where rows from the left have $T_{id}=0$.

\tocs{To compute a left outer join, we can compute the inner join but should not invalidate any left-table rows. The desired relation is shown in Table \ref{tab:loj}. We invalidate rows from $R$ which are the first of their group. When this occurs, it means no such row from $L$ exists, so this row is not in the output relation. To only keep \textit{unique} (rather than all) rows from $L$, the second clause could instead be replaced with $O.T_{id} \oplus \textsc{Distinct}(K)$.}

\begin{table}[h!]
    \centering
    \begin{tabular}{ccc||c}
    $V_{LR}$ & $T_{id}$ & \textsc{Distinct} & $V$ \\\hline
    0 & X & X &  0 \\
    1 & 0 & 0 &  1 \\
    1 & 0 & 1 &  1 \\
    1 & 1 & 0 &  1 \\
    1 & 1 & 1 &  0 \\
    \end{tabular}
    \caption{Left outer join validity-update rule $O.V\gets O.V_{LR}\wedge \neg (O.T_{id}~\wedge~\textsc{Distinct}(K))$}
    \label{tab:loj}
\end{table}

\paragraph{Semi-join.} The semi-join $L\ltimes R$ returns all rows in $L$ which match a row in $R$. Observe that we can reframe this operator into one we have already seen by swapping the tables: which operator returns all rows in $R$ which match a row in $L$? This is the inner join $R\bowtie L$. So, to implement $L\ltimes R$, we simply call $R\bowtie L$, and then project the result columns to maintain the semantics of semi-join (namely, only return columns in $L$). This final operation requires no extra work because our inner-join operator, by default, returns all data in columns from the ``right'' (here, $L$). We acknowledge that this transformation of semi-join to inner join is not correct in general, but rather a byproduct of the way our oblivious control flow is implemented for one-to-many joins.

To see why this rewriting works, consider what rows are \textit{removed} from $R \bowtie L$ by the validity-update procedure of our inner join protocol: all rows from $R$, and all rows from $L$ which do not have a matching ``primary key'' in $R$. Of course, in a semi-join, table $R$ is not part of the output relation, so we should always remove all rows from $R$. And, if a row in $L$ does not have a match in $R$, then it should be removed from the output, by the definition of semi-join.

\paragraph{Anti-join.} The inverse of semi-join is anti-join, $L\rhd R$, which \textit{removes} from $L$ all rows matching one in $R$. We apply a similar transformation: we need an operator which will ignore all rows from $R$, and keep rows from $L$ which do not match any in $R$. A partial answer is the \textit{right-outer join} from above, again applied in the opposite direction, $R\rightouterjoin L$. This would keep all rows in the ``right'' (here, $L$), but it would also incorrectly keep rows from $L$ which matched one from $R$.

The solution here is to apply a special aggregation over the valid bit\footnote{Specifically, this is just a \texttt{copy} aggregation, which already has special handling in order to copy rows from the left to right.} which copies the valid bit of the \textit{first} row of a group to all other rows in the group. We do not include $O.T_{id}$ as an aggregation key. Thus, we invalidate all rows for which the first row in that group was previously invalidated by the (right outer) join. Since this join will invalidate all ``left'' ($=R$) rows, any group for which \textit{a row exists in $R$} will have its \textit{first} row invalidated by the join, and in the subsequent aggregation, all such rows from this group will be invalidated. This correctly implements anti-join; to match the semantics of the operator, we again project to the columns on the left. Thanks to our unified join-aggregation control flow, this validation update is performed at the same time as the join, and incurs the same cost as a copy aggregation (i.e., $O(n \log n)$ multiplex operations) in a normal join. Again, we acknowledge that this rewriting of anti-join to right-outer join is not correct in relational algebra, but a convenient optimization in our setting.

One side effect of this rewriting (for both semi-join and anti-join) is that our implementations of these operators transparently handles duplicates in $R$, by Cases I and II, above. Both semi-join and anti-join only consider the \textit{existence}, in table $R$, of keys in $L$, so we need no special handling for duplicate keys.

\subsection{Proof Sketch of \protoPrefixAgg{} correctness}\label{apdx:agg}

\camera{In this section, we provide an argument for why the \protoPrefixAgg{} protocol achieves correctness. As an illustrative example, consider a table with $n = 2^\ell$ rows and three columns: a key column $K$, input column $A$, and output column $G$. We use these parameters for convenience but emphasize that the ideas generalize to any database. Additionally, while the aggregation network requires that the \textit{full} table size is a power of two, each key can have arbitrary cardinality. If the number of input rows is not a power of two, we pad the table with dummy (invalid) rows and remove them after the protocol terminates.}

\camera{The first step of the algorithm is to copy column $A$ into column $G$ (Line 1), since the aggregation function is applied in place. From this point, we only consider column $G$.}

\begin{claim}
    \camera{Values across distinct keys are never aggregated.}
\end{claim}

\camera{Take an arbitrary row $i$ and distance $d$. Assume the keys in rows $(i, i+d)$ do not match. Then $b \gets 0$ (Line 4) and we update $G_{i+d} \gets \mathtt{Mux}(0, G_{i+d}, g) = G_{i+d}$ (Line 6), so this row is not modified. Therefore, we only aggregate rows with matching keys.}

\begin{claim}
    \camera{The algorithm correctly aggregates groups.}
\end{claim}

\camera{\protoPrefixAgg{} expects an input table sorted on the key column. This puts all rows of a given key-group next to each other. Assume row $j$ is the final row of some arbitrary group $k$ of size $m\leq n$. Then, we want to argue that the final value $G_j = f(G_{j-m}, \dots, G_j)$.}

\camera{In the final round of the protocol, $G_j$ is updated with the value $n/2$ indices earlier (assume for the sake of argument that $m > n / 2$), $G_j \gets f(G_{j-n/2}, G_j)$. In the prior round, each of these two values were updated with their own value along with the values $n/4$ indices prior. The algorithm proceeds in this manner, taking the next power of two in each round. In the first round we have a distance $d=1$ so adjacent elements are aggregated.}

\camera{This procedure is easy to visualize if we consider this aggregation structure as a tree, where each “prior index” is a left branch, and the “same index” is a right branch. \protoPrefixAgg{} builds a binary tree, where the path through the tree from $G_j$ to $G_{j-h}$, for arbitrary $h \leq m$, corresponds to the binary representation of $h$. Each leaf node has a single unique path to $G_j$, since each offset $h$ has a unique binary representation. Thus, each input value is aggregated once and in order of its appearance in the original list. Since the aggregation function $f$ is self-decomposable, we have that $f(f(X), f(Y)) = f(X, Y)$, so by induction we can show that the final result is $f(G_{j-m}, \dots, G_j)$.} \qed

\subsection{Trimming the output of Join operations}\label{apdx:trim}

In Section \ref{sec:join}, we discussed bounding the outputs of joins to the size of the right table. Here, we elaborate on the heuristic used to automatically govern this trim operation. It is important that, due to the obliviousness of our operators, the ``size'' of tables discussed in this section refers to the secret-shared table size, \textit{not} the number of valid rows (which is not known at runtime).

First, we provide intuition by way of an example. Consider a join between two asymmetric tables, $|L|=10$ and $|R|=10^6$. After join, we have a larger table of size $|O|=10+10^6$. Clearly it is not worth trimming: we must perform valid-bit sort on a table of size $|O|\approx|R|=10^6$, but are only able to remove $|L|=10$ elements. However, if a join is more symmetric, say $|L|=10^6$ and $|R|=2\cdot10^6$, trimming is much more promising. While we sort on $|O|=3\cdot 10^6$ elements, we can then remove $1/3$ of the output rows.

Trimming brings significant improvements in composed queries where we also perform aggregations on the output of the join. Since the odd-even aggregation algorithm requires power-of-two-sized inputs, trimming frequently allows us to ``drop'' below the next power of two, leading to significant speedups in query execution. In the analysis below, we optimize for such composed operation, assuming a join will be followed by additional joins with other tables.

To establish the heuristic, we define the abstract costs of each of our major operators. Let $J(n)$, $V(n)$, $S(n)$, and $A(n)$ be the respective costs of join, valid-bit sort, full sort, and aggregation over $n$ rows (we abstract away aggregation functions at this phase). 

In this analysis, we will say we first join two tables of size $L$ and $R$, with $N:=L+R$ the combined size. A subsequent join is performed with a table of size $T$. Intuitively, we want to trim if the additional cost (valid-bit sorting over $N$ rows) is less than the future savings. Expressed formally, when $J(N + T) > V(N) + J(R + T)$, we should trim.

First, we observe that joins are superlinear: $\forall n,m>0:J(n+m)\geq J(n)+J(m)$. This follows from both the superlinearity of both sorting and aggregation.
\begin{align*}
J(N+T)&>J(L)+J(R+T)\\
J(L)+J(R+T)&>V(N)+J(R+T)\\
J(L)&>V(N)\implies\text{trim}
\end{align*}
That is, we should trim when the cost of performing a join over the trimmed rows would be more expensive than performing a valid-bit sort over the entire table.

We now decompose each operator:
\begin{itemize}
    \item Valid-bit sorting is implemented with radixsort so is approximately linear, $V(n+m)\approx V(n)+V(m)$.
    \item Joins consist of a two single-bit sorts (valid column plus $T_{id}$), a regular sort, and an aggregation; $J(n)\approx 2V(n)+S(n)+A(n)$.
\end{itemize}

Define $\alpha:=R/L$. For most applications involving unique-key joins, $\alpha>1$. Then,
\begin{align*}
V(N)&<J(L)\implies\text{trim}\\
V(L)+V(R)&<2V(L)+S(L)+A(L)\\
(\alpha - 1)V(L) &< S(L) + A(L)
\end{align*}

It is difficult to reason, in general, about the cost of aggregation, $A(\cdot)$, since this cost depends heavily on (possibly user-defined) aggregation functions. However, it can be bounded from below considering only the control flow of the aggregation operator itself. Let $\omega$ be the bitwidth of values in the table. Then both aggregation and quicksort perform at least $n\lg(n)$ comparisons, which, under MPC, each require $\log \omega$ operations. Valid-bit sorting, on the other hand, performs $6n$ $\omega$-bit permutation operations, each of which we can approximate as costing $N$ multiplications, where $N$ is the number of parties required by the MPC protocol. Inserting these costs into the above, we see that:
\begin{align*}
(\alpha - 1)6NL &< 2 L\lg(L) \lg(\omega)\implies\text{trim}\\
3\alpha N &< \lg(L)\lg(\omega)
\end{align*}

For our common use cases, $\omega=128$ bits (due to input padding) and $N\in\{2,3,4\}$. The table below gives some example sizes, and confirms our intuitions about only trimming when we can remove sufficiently many rows from the left table.

\begin{center}
\begin{tabular}{rrr}
$L$  & $\alpha$ for $N=3\text{PC}$ & Trim when $R<\dots$ \\ \hline
100  & 5.2  & 516 \\
10k  & 10.3 & 103k \\
1M   & 15.5 & 15.5M \\
100M & 20.7 & 2.07B \\
\end{tabular}
\end{center}

\tocs{A similar analysis can be performed for radixsort, giving a slightly different heuristic which nonetheless makes similar concrete decisions on real queries.}

\section{Bandwidth Measurements}\label{apdx:bw}

In Table \ref{tab:bw-tpch}, we provide the bandwidth usage for each query used in the paper. Since TPC-H queries have a wide variety of input sizes (nearly an order-of-magnitude spread at a fixed Scale Factor), we report bandwidth usage in kilobytes per row. While these results have been collected at SF10, query execution scales only slightly superlinearly, so the values below provide a good bandwidth estimate across all query sizes reported in \ours. Measurements are also normalized per party; some protocols are not perfectly symmetric, so we divide the \textit{total communication} by the number of computing parties (2 in SH-DM, 3 in SH-HM, and 4 in Mal-HM, respectively). We observe that SH-DM has just under twice the communication complexity of SH-HM, and Mal-HM about three times. \tocsCam{The fix for the attack by Br\"uggemann et al.~\cite{bs26} has minimal communication overhead (consisting of only $O(\log n)$ constant-bandwidth checks), so is not expected to noticeably change the results for Mal-HM.}
Results for other experiments (\S\ref{sec:compare-queries}) are provided in Tables \ref{tab:bw-secrecy}, \ref{tab:bw-sf}, \ref{tab:bw-sf-sort}, and \ref{tab:bw-mpspdz-sort}.

\begin{table}[b]
\centering
\makebox[\linewidth][c]{%
\begin{tabular}{r|rrr@{\hspace{8mm}}r|rrr}
Query & SH-DM & SH-HM & Mal-HM & Query & SH-DM & SH-HM & Mal-HM \\ \hline\hline
TPC-H Q1 & 33.1  & 18.5 &  51.4  & Aspirin      & 39.9 & 21.2 & 61.2 \\
Q2       & 65.0  & 35.2 & 100.2  & Comorbidity  & 42.5 & 23.5 & 65.8 \\
Q3       & 42.2  & 23.6 &  65.4  & Credit       & 26.6 & 14.0 & 40.6 \\
Q4       & 44.1  & 24.3 &  68.5  & SYan         & 42.2 & 23.9 & 65.6 \\
Q5       & 122.5 & 66.4 & 188.7  & Patients     & 22.3 & 12.0 & 34.5 \\
Q6       &  0.2  &  0.3 &   0.6  & Market Share & 17.8 & 10.0 & 27.3 \\
Q7       & 111.9 & 61.0 & 173.4  & Password     & 26.9 & 14.8 & 40.8 \\
Q8       & 120.2 & 65.5 & 186.6  & C. Diff      & 27.4 & 15.0 & 41.6 \\
Q9       & 113.4 & 62.5 & 176.1  & Secrecy Q2   & 32.7 & 18.4 & 50.6 \\
Q10      & 73.6  & 40.5 & 114.7  & & & & \\
Q11      & 41.8  & 23.4 &  65.7  & & & & \\
Q12      & 46.9  & 25.5 &  72.0  & & & & \\
Q13      & 45.6  & 25.6 &  70.4  & & & & \\
Q14      & 18.0  &  9.9 &  27.9  & & & & \\
Q15      & 54.1  & 30.0 &  83.6  & & & & \\
Q16      & 94.2  & 51.3 & 144.3  & & & & \\
Q17      & 53.0  & 29.4 &  82.8  & & & & \\
Q18      & 86.6  & 47.5 & 133.3  & & & & \\
Q19      & 21.2  & 11.5 &  32.8  & & & & \\
Q20      & 82.1  & 45.8 & 127.0  & & & & \\
Q21      & 160.2 & 87.0 & 245.8  & & & & \\
Q22      & 16.6  &  9.3 &  26.2  & & & & \\
\end{tabular}
}
\vspace{2mm}
\caption{Bandwidth (KB) per row and computing party, for all queries and protocols used in \ours.}
\label{tab:bw-tpch}
\end{table}

\begin{table*}[b]
\centering
\makebox[\linewidth][c]{%
    \begin{tabular}{r|rr|rr}
    Input Size & SecretFlow: SBK & \ours: RS (64b) & SecretFlow: SBKvalid & \ours: RS (32b) \\ \hline\hline
    100k & 921.4 & 517.6 & 282.0 & 210.0 \\
    1M & 9\,212.4 & 5\,176.0 & 3\,212.2 & 2\,100.0 \\
    10M & 92\,118.9 & 51\,760.0 & 37\,568.5 & 21\,000.0 \\
    \end{tabular}
}
    \vspace{2mm}
    \caption{Total bandwidth (MB) for SecretFlow and \ours in the experiment of Fig.~\ref{fig:secretflow-comparison}. \ours's optimized radixsort achieves between $1.34\times$ and $1.79\times$ lower bandwidth than the SecretFlow algorithms.}
    \label{tab:bw-sf-sort}
\end{table*}

\begin{table}
    \centering
    \begin{tabular}{r|rr}
    Query       & Secrecy & \ours   \\ \hline\hline
    TPCH-Q6     & 0.3     & 0.3   \\
    Password    & 30.9    & 10.7  \\
    Credit      & 39.2    & 9.3   \\
    Comorbidity & 56.4    & 23.7  \\
    C. Diff     & 63.1    & 12.1  \\
    Aspirin     & 2\,820    & 10.7  \\
    TPCH-Q4     & 4\,825    & 16.5  \\
    TPCH-Q13    & 1\,712    & 20.0  \\
    \end{tabular}
    \captionof{table}{Bandwidth (KB) per row and party for the queries used in Fig.~\ref{fig:secrecy-comparison} (left). Secrecy's high bandwidth is due to its $O(n^2)$ join operator and $O(n\log^2 n)$ sorting algorithm (bitonic sort).}
    \label{tab:bw-secrecy}
\end{table}

\begin{table}
    \centering
    \begin{tabular}{r|rr}
    Query & SecretFlow & \ours  \\ \hline\hline
    S1    &  <1     & 14       \\
    S2    &  <1     & 28       \\
    S3    &  88     & 7\,677    \\
    S4    &  286    & 7\,735      \\
    S5    &  2\,835   & 29\,834    \\
    \end{tabular}
    \captionof{table}{Bandwidth (bytes) per row and party for the queries used in Fig.~\ref{fig:secrecy-comparison} (right). SecretFlow's low bandwidth is due to its join operator, which leaks matching rows to parties and allows them to perform most subsequent operations locally (i.e., without communication).}
    \label{tab:bw-sf}
\end{table}

\begin{table*}
    \centering
    \begin{tabularx}{0.8\textwidth}{X|rr|rr|rr}
    & \multicolumn{2}{c|}{SH-DM (2PC)} & \multicolumn{2}{c|}{SH-HM (3PC)} & \multicolumn{2}{c}{Mal-HM (4PC)} \\
    $k$ & MP-SPDZ & \ours & MP-SPDZ & \ours & MP-SPDZ & \ours \\
    \cline{1-7}\noalign{\vskip 2pt}
    \cline{1-7}
    10 & 165.3          & 5.3       & 18.7       & 11.9         & 272.0       & 40.3 \\
    11 & 363.7          & 10.6      & 37.4       & 23.9         & 600.4       & 80.6 \\
    12 & 793.4          & 21.2      & 74.8       & 47.8         &  1\,315.3   & 161.3 \\
    13 & 1\,719.0       & 42.4      & 149.6      & 95.6         & 2\,860.5    & 322.7 \\
    14 & 3\,702.2       & 84.8      & 299.2      & 191.2        & 6\,182.2    & 645.5 \\
    15 & 7\,933.0       & 169.6     & 598.4      & 382.4        & 13\,288.1   & 1\,291.0 \\
    16 & 16\,922.9      & 339.2     & 1\,196.7   & 764.9        & 28\,424.6   & 2\,582.1 \\
    17 & 35\,959.8      & 678.4     & 2\,393.5   & 1\,529.8     & 60\,547.2   & 5\,164.2 \\
    18 & 76\,147.6      & 1\,356.8  & 4\,787.1   & 3\,059.7     & 128\,492.2  & 10\,328.4 \\
    19 & 160\,751.0     & 2\,713.7  & 9\,574.2   & 6\,119.4     &  271\,780.4 & 20\,656.9 \\
    20 & 338\,414.0     & 5\,427.4  & 19\,148.4  & 12\,238.9    & (OOM)       & \\
    21 & 710\,650.0     & 10\,854.8 & 38\,296.8  & 24\,477.9    &             & \\
    22 & (Crash)        &           & 76\,593.9  & 48\,955.9    &             & \\
    23 &                &           & 153\,187.5 & 97\,911.8    &             & \\
    24 &                &           & 306\,375.0 & 195\,823.6   &             & \\
    25 &                &           & (OOM)      &              &             & 
    \end{tabularx}
    \vspace{2mm}
    \caption{Total bandwidth (MB) for MP-SPDZ and \ours in the experiments of Fig.~\ref{fig:mpspdz-comparison} (oblivious radixsort on $2^k$ rows). \ours's optimized radixsort achieves between $1.57\times$ and $65.5\times$ lower bandwidth than the MP-SPDZ algorithm. MP-SPDZ crashes at $2^{22}$ in SH-DM and runs out of memory at $2^{25}$ in SH-HM and $2^{20}$ in Mal-HM.}
    \label{tab:bw-mpspdz-sort}
\end{table*}

\end{document}